\documentclass[10pt,aps,prx,reprint,superscriptaddress]{revtex4-2}
\usepackage{amsthm}
\usepackage{libertine}
\usepackage[table]{xcolor}
\usepackage[normalem]{ulem}
\usepackage[T1]{fontenc}
\usepackage[utf8]{inputenc}
\usepackage{tikz}
\usetikzlibrary{arrows.meta}

\usepackage{mwe}
\definecolor{TSUYUKUSA}{RGB}{46, 169, 223}
\definecolor{KURENAI}{RGB}{203, 27, 69}

\pdfoutput=1

\usepackage{hyperref}
\usepackage{amssymb}
\usepackage{mathtools}
\usepackage{amsthm}
\usepackage{eucal}
\usepackage{dsfont}
\usepackage{libertine}
\usepackage{newtxmath}
\usepackage[capitalise]{cleveref}
\usepackage{thm-restate}
\usepackage{subfig}
\usepackage{framed}
\usepackage{environ}

\usepackage{enumitem}
\usepackage{algorithmic}
\NewEnviron{Anonymous}{\ifx\AnonymousSwitch\undefined\BODY\fi}

\makeatletter
\renewcommand{\section}{\@startsection {section}{1}{\z@}%
             {-3.5ex \@plus -1ex \@minus -.2ex}%
             {2.3ex \@plus .2ex}%
             {\normalfont\Large\scshape\bfseries}}
\renewcommand{\subsection}{\@startsection{subsection}{2}{\z@}%
             {-3.25ex\@plus -1ex \@minus -.2ex}%
             {1.5ex \@plus .2ex}%
             {\normalfont\large\scshape\bfseries}}
\renewcommand{\subsubsection}{\@startsection{subsubsection}{2}{\z@}%
             {-3.25ex\@plus -1ex \@minus -.2ex}%
             {1.5ex \@plus .2ex}%
             {\normalfont\normalsize\scshape\bfseries}}
\makeatother

\allowdisplaybreaks[1]

\def\A{\CMcal{A}}
\def\B{\CMcal{B}}

\def\H{\CMcal{H}}

\def\M{\CMcal{M}}

\def\Q{\CMcal{Q}}

\theoremstyle{plain}
\newtheorem{theorem}{Theorem}[section]
\newtheorem{lemma}[theorem]{Lemma}
\newtheorem{prop}[theorem]{Proposition}

\theoremstyle{definition}
\newtheorem{definition}[theorem]{Definition}

\newtheorem{question}[theorem]{Question}

\newtheorem{remark}[theorem]{Remark}
\newtheorem{fact}[theorem]{Fact}

\newtheorem{protocol}[theorem]{Protocol}

\newtheorem{assumption}[theorem]{Assumption}
\makeatletter
\newenvironment{replemma}[1]{%
  \edef\@savedref{#1}%
  \let\origthelemma\thelemma 
  \renewcommand{\thelemma}{\ref{\@savedref}} 
  \begin{lemma}%
}{%
  \end{lemma}%
  \let\thelemma\origthelemma 
}
\makeatother

\newcommand {\set} [1] {\ensuremath{ \left\lbrace #1 \right\rbrace }}

\newcommand {\innerproduct} [2] {\ensuremath{\left \langle #1 , #2 \right \rangle}}

\newcommand{\normthree}[1]{{\left\vert\kern-0.25ex\left\vert\kern-0.25ex\left\vert #1 \right\vert\kern-0.25ex\right\vert\kern-0.25ex\right\vert}}
\newcommand {\br} [1] {\ensuremath{ \left( #1 \right) }}

\DeclarePairedDelimiter{\abs}{\lvert}{\rvert}

\newcommand {\bra} [1] {\ensuremath{ \left\langle #1 \right| }}
\newcommand {\ket} [1] {\ensuremath{ \left| #1 \right\rangle }}
\newcommand {\ketbratwo} [2] {\ensuremath{ \left| #1 \middle\rangle \middle\langle #2 \right| }}
\newcommand {\ketbra} [1] {\ketbratwo{#1}{#1}}

\DeclareMathOperator*{\bigE}{\mathbb{E}}

\newcommand{\nnorm}[1]{{\left\vert\kern-0.25ex\left\vert\kern-0.25ex\left\vert #1
		\right\vert\kern-0.25ex\right\vert\kern-0.25ex\right\vert}}

\newcommand {\Tr} {\ensuremath{ \mathrm{Tr} }}

\usetikzlibrary{calc}
\tikzset{meter/.append style={draw, inner sep=10, rectangle, font=\vphantom{A}, minimum width=30, line width=.8,
 path picture={\draw[black] ([shift={(.1,.3)}]path picture bounding box.south west) to[bend left=50] ([shift={(-.1,.3)}]path picture bounding box.south east);\draw[black,-latex] ([shift={(0,.1)}]path picture bounding box.south) -- ([shift={(.3,-.1)}]path picture bounding box.north);}}}

\crefname{prop}{Proposition}{Propositions}
\crefname{question}{Question}{Questions}

\newcommand{\posint}{{\mathbb Z}_{>0}}

\renewcommand{\i}{\mathrm{i}}
\newcommand{\reals}{{\mathbb R}}

\newcommand{\randS}{\ensuremath{\mathbf{S}}}
\newcommand{\randX}{\ensuremath{\mathbf{X}}}

\newtheorem*{theorem*}{Theorem}

\newcommand{\x}{\otimes}

\newcommand{\calX}{\mathcal{X}}
\newcommand{\calY}{\mathcal{Y}}
\newcommand{\calA}{\mathcal{A}}
\newcommand{\calB}{\mathcal{B}}

\newcommand{\1}{\mathbb{1}} 

\newcommand{\MIP}{\mathrm{MIP}}
\newcommand{\RE}{\mathrm{RE}}
\newcommand{\val}{\mathrm{val}}

\newcommand{\Corr}{\ensuremath{\mathrm{Corr}}}
\newcommand{\Fnorm}[1]{\ensuremath{\left\|#1\right\|_F}}
\newcommand{\opnorm}[1]{\ensuremath{\left\|#1\right\|}}
\newcommand{\Trerr}{\ensuremath{\epsilon_{\mathrm{tr}}}}
\newcommand{\winerr}{\ensuremath{\epsilon_{\mathrm{win}}}}

\allowdisplaybreaks

\begin{document}

\title{Nonlocal Games in the High-Noise Regime: Optimal Quantum Values and Rigidity}

%


\author{Honghao Fu}
\email{honghao.fu@concordia.ca}
\affiliation{Department of Cybersecurity and Intelligent Systems Engineering, Concordia University}

\author{Minglong Qin}
\email{mlqin@nus.edu.sg}
\affiliation{Centre for Quantum Technologies (CQT), National University of Singapore}

\author{Haochen Xu}
\email{hpx5065@psu.edu}
\affiliation{Department of Computer Science and Engineering, Penn State University}
\affiliation{Key Laboratory of System Software (Chinese Academy of Sciences) and State Key Laboratory of Computer Science, Institute of Software,
Chinese Academy of Sciences}

\author{Penghui Yao}
\email{phyao1985@gmail.com}
\affiliation{State Key Laboratory of Novel Software Technology, Nanjing University, Nanjing 210023, China}
\affiliation{Hefei National Laboratory, Hefei 230088, China}

\date{\today}

\begin{abstract}

Motivated by the limitations of near-term quantum devices, we study nonlocal games in the \emph{high-noise regime}, where the two players may share arbitrarily many copies of a noisy entangled state. 
In this regime, existing rigidity theorems are unable to certify any nontrivial quantum structure.
We first characterize the maximal quantum winning probabilities of the CHSH game\cite{clauser1969proposed}, the Magic Square game\cite{mermin1990quantum}, and their 2-out-of-$n$ variants \cite{chao2018test} as explicit functions of the noise rate. 
These characterizations enable the construction of device-independent protocols for estimating the underlying noise level. 
Building on these results, we prove noise-robust rigidity theorems showing that these games certify one, two, and 
$n$ pairs of anticommuting Pauli observables, respectively. 
To our knowledge, these are the first rigidity results of Pauli measurements that remain sound in the high-noise regime,
which has applications in Measurement-Device-Independent (MDI) cryptography and studying the computational power 
of Multi-prover Interactive Proof System with entanglement and a vanishing completeness-soundness gap ($\MIP^*_0$).
Our proofs rely on Sum-of-Squares decompositions and Pauli analysis techniques originating from quantum proof systems and quantum learning theory, respectively.

\end{abstract}

\maketitle

\section{Introduction}
Nonlocality is a central operational feature of quantum theory and admits direct experimental tests~\cite{bell64,fine1982hidden,tsirel1987quantum,brunner2014bell}. It refers to the input-output correlations generated by spatially separated parties that are consistent with the no-signaling principle, yet cannot be reproduced by any local hidden-variable (LHV) model. The detection of nonlocality
can be equivalently formulated through nonlocal games: in the bipartite setting, the referee chooses inputs $x,y$, the two players who cannot communicate respond with outputs $a,b$, and the referee decides whether they win according to a predefined predicate.
For some nonlocal games, quantum mechanics allows for a maximal winning probability that is higher than the probabilities achieved by any classical strategies, for example, the CHSH game~\cite{clauser1969proposed} and the Magic Square game~\cite{aravind2004quantum}.

In realistic implementations, imperfections and environmental noise in state preparation are unavoidable, and the shared resource is therefore more accurately modeled as a mixed state rather than an ideal pure state~\cite{brunner2014bell,horodecki1995violating}. Such noise inevitably restricts the attainable winning probability of a nonlocal game. This raises a basic question: how should one quantify the best possible performance of a nonlocal game when the source is noisy?


In this work, we study this question in a setting where the source preparing and distributing the bipartite state is assumed to be \emph{trusted but noisy}. By trusted, we mean that the source is designed to repeatedly prepare a fixed
bipartite quantum state. However, unavoidable imperfections during preparation result in deviations from the ideal target state, which makes the source noisy. 
Even in this setting, characterizing the optimal game performance is highly nontrivial and remains largely unexplored.
 In the idealized noiseless setting, the optimal quantum value of certain games, such as CHSH, is well understood \cite{kaniewski2016analytic}. However, in the high-noise regime, finding the optimal value is less clear because the players can share arbitrarily many copies of the mixed resource state and apply global measurements on them.  
Existing results in the noisy setting~\cite{horodecki1995violating,lifshitz2025noise} crucially rely on the assumption that the players share a bounded number of noisy states, whereas the unbounded-copy regime can behave very differently and is also the natural one from the viewpoint of quantum proof systems. This leads to the following question:

\begin{question}
    \label{q1}
    How can one characterize the optimal winning probability, i.e., the quantum value of a nonlocal game, when the players share arbitrarily many copies of a noisy state?
\end{question}

A uniquely quantum feature of nonlocal games is \emph{rigidity}: in certain games, observing a near-optimal winning probability implies that the underlying strategy, i.e. the shared state and implemented measurements, is close (up to local isometries) to a specific target state and measurements~\cite{mayers2003self,kaniewski2016analytic}. 
Rigidity results underpin device-independent certification (self-testing), and play a central role in device-independent randomness generation and quantum key distribution~\cite{miller2016robust,vazirani2019fully}. 
Moreover, in theoretical computer science, it lays the foundation for the studies of Multi-prover Interactive Proof Systems with Entanglement ($\MIP^*$), as self-testing of EPR pairs and Pauli measurements enables the exponential compression of $\MIP^*$ protocols in the proof of $\MIP^* = \RE$ \cite{JNVWY'20}.
These developments rely on rigidity theorems proved in the noiseless setting, which are typically formulated relative to the ideal optimal quantum value. 
However, in the noisy regime, the optimal achievable winning probability may already be bounded away from this ideal value, so such theorems do not directly apply.
Then a natural follow-up question to \cref{q1} is the following.
\begin{question}
    \label{q2}
   If the observed winning probability is sufficiently close to the optimal value in the noisy setting, what can we certify about the source and the measurement devices?
\end{question}

Our trusted-but-noisy setting is closely related to the measurement-device-independent (MDI) framework. In MDI protocols, the shared state and the number of copies are fixed, and the goal is to certify the security of tasks such as quantum key distribution, based solely on observed measurement statistics~\cite{lo2012measurement,buscemi2012all,branciard2013measurement}. In comparison, we assume a trusted model for the noisy source, but allow arbitrarily many copies of the shared state. We ask what near-optimal nonlocal-game performance certifies about the measurements and, in some cases, about the noise level of the source. This regime is both physically natural and theoretically challenging, as unbounded copies and collective measurements can potentially obscure the structure required for rigidity. 

We answer \cref{q1,q2} for four foundational nonlocal games: the CHSH game, the Magic Square game, and their 2-out-of-$n$ versions. We answer \cref{q1} by determining the exact optimal values for these games in the noisy setting. We answer \cref{q2} by proving rigidity theorems against these optimal values. In particular, our rigidity result is stronger as we prove the near-optimal observables are effectively supported on a single register, whereas this property does not hold in the noiseless setting.
\subsection{Main Results}

We study the nonlocal games above in a trusted-but-noisy setting, where the source is intended to prepare EPR pairs. 

\paragraph{CHSH games}

For the CHSH game, we consider any noise such that the resulting state has maximally mixed marginals. We use the quantum maximal correlation to label the noise level. This quantity is fundamental in quantum information theory and measures the largest correlation that can be extracted from the state using local observables~\cite{Beigi:2013}. It is directly connected to the CHSH task, since the game value is determined by correlations of this form. In particular, quantum maximal correlation characterizes the optimal CHSH value in our setting.
\begin{definition}[Noise model for CHSH]
\label{def:chsh}
Fix a constant \(\rho\in(0,1)\), where \(\rho\) denotes the correlation parameter and \(1-\rho\) quantifies
the noise level.
We consider bipartite states $\Psi$ with maximally mixed marginals whose quantum maximal correlation is upper bounded by $\rho$.
The players are allowed to share arbitrarily many copies of \(\Psi\).
\end{definition}

This model includes many types of noise. Here are some examples: 
A special case is the depolarized EPR pair
\[
\Phi_\rho := \Delta_\rho^{(4)}(\Phi)
= \rho \Phi + (1-\rho)\frac{I}{4},
\]
which has maximally mixed marginals and maximal correlation $\rho$.
More generally, we consider an EPR pair $\Phi$ subject to one-sided Pauli noise. For
\begin{align*}
\mathcal E_{a,b,c}(\rho)
=
(1-a-b-c)\rho+aX\rho X+bY\rho Y+cZ\rho Z,\\
a,b,c\ge 0,\quad a+b+c\le 1,\\\max\bigl\{|1-2(b+c)|,\ |1-2(a+c)|,\ |1-2(a+b)|\bigr\}\leq \rho,
\end{align*}
the state $(\mathcal E_{a,b,c}\otimes \mathrm{Id})(\Phi)$ also has maximally mixed marginals and quantum maximal correlation
\[
\max\bigl\{|1-2(b+c)|,\ |1-2(a+c)|,\ |1-2(a+b)|\bigr\}\leq \rho.
\]
We assume that the noise acts only on the shared quantum states, and that communication between the players and the referee is noiseless, which can be ensured by standard error-correcting techniques.

\begin{remark}
In the trusted-but-noisy setting, we impose two structural assumptions on the source. 

First, the source prepares i.i.d.\ copies of a fixed bipartite state. This is also the standard assumption in applications such as quantum state tomography, where repeated preparation of the same state underlies statistical estimation and concentration arguments.

Second, in \cref{def:chsh} we assume that the single-copy source state has maximally mixed marginals. This condition is essential in the many-copy regime: without it, the optimal CHSH violation may depend on the number of shared copies. For example, for the tilted EPR state
$\Psi_\theta=\cos\theta\,\ket{00}+\sin\theta\,\ket{11}, \theta\in (0,\pi/4)$,
the achievable CHSH violation can increase when multiple copies are available \cite[Eq. (7)]{liang2006better}.
\end{remark}

The following theorem summarizes our answers to \cref{q1,q2} for the CHSH game:
\begin{theorem}[CHSH, informal]\label{thm:intro-chsh}
Let $\Psi$ be a two-qubit state with maximally mixed marginals and quantum maximal correlation upper bounded by $ \rho \in (0,1)$. Suppose the players are allowed to share arbitrarily many i.i.d.\ copies of $\Psi$.

\begin{enumerate}
    \item The optimal quantum winning probability of the CHSH game, under traceless observables, is upper bounded by
    \[
    \frac{1}{2}+\frac{\sqrt{2}}{4}\rho.
    \]

    \item Moreover, if a strategy with shared state $\Psi^{\otimes n}$ and traceless observables wins the CHSH game with probability at least
    \[
    \frac{1}{2}+\frac{\sqrt{2}}{4}\rho-\epsilon,
    \]
    then the players' observables are $O(\sqrt{\epsilon})$-close to Pauli $Z$ and $X$ acting on a single shared qubit, up to local unitaries. 
\end{enumerate}
\end{theorem}
The traceless condition is easy to check in experiments. One can enforce it via a simple trace test, where the referee repeatedly plays the base game and checks that the players' answers are statistically unbiased. We also show that allowing small bias, quantified by a normalized trace bound $\Trerr$, increases the winning probability by at most $O(\Trerr^2)$.

The optimal quantum value can be attained using the same optimal observables as in the noiseless cases. In particular, allowing arbitrarily many shared copies does not improve the winning probability. Our result also yields the noise threshold for nonlocality: the CHSH game remains nonlocal as long as $\rho > 1/\sqrt{2}$.

The rigidity result in the noisy regime differs from the noiseless case. For constant $\rho \in (0,1)$, the robustness remains $O(\sqrt{\epsilon})$, matching the noiseless setting~\cite{mckague2012robust}.
At the same time, near-optimal strategies are forced to concentrate on a single register: the relevant observables act nontrivially on one qubit and trivially on the rest. As a result, to extract one pair of Pauli $Z$ and $X$, we only need to operate on one qubit, making them much simpler to implement in practice. However, this single-register structure does not hold in the noiseless case $\rho = 1$, where optimal strategies may be spread across multiple registers.

The method we use for the CHSH game is also applicable to other Bell inequalities, and we can 
use these methods to set the boundary for quantum behaviour, i.e. the correlations generated by quantum measurements in nonlocal games under noise. 
\Cref{fig correlation set} shows the set of quantum correlations, denoted by $\Q_\rho$, when the shared state is arbitrarily many copies of $\Phi_\rho$, and there are two questions and two answers per party.
We can show that $\Q_\rho$ is obtained by uniformly shrinking $\Q$, the set of all quantum correlations in the same setting, toward the origin, that is $\Q_\rho$ is a circle centered at $(0,0)$ with radius $2\sqrt{2}\rho$.

\begin{figure}[htbp]
    \centering
    
    \begin{tikzpicture}[
    scale=0.58,
    >=Stealth,
    every node/.style={font=\scriptsize}
]

    \def\tsirelson{2.828427} 
    \def\rhoVal{0.8}
    \def\rhoRadius{\tsirelson * \rhoVal}

    \draw[->, thick] (-4.2, 0) -- (4.2, 0) node[right] {$\tilde{S}$};
    \draw[->, thick] (0, -4.2) -- (0, 4.2) node[above] {$S$};

    \foreach \y in {2, \tsirelson, 4} {
        \draw[dashed, very thin, darkgray] (0, \y) -- (-4.0, \y);
    }
    \draw[dashed, very thin, red] (0, \rhoRadius) -- (-4.0, \rhoRadius);

    \draw[fill=teal!50, draw=black, thick, fill opacity=0.7]
        (4,0) -- (0,4) -- (-4,0) -- (0,-4) -- cycle;

    \draw[fill=violet!50, draw=black, thick, fill opacity=0.7]
        (0,0) circle (\tsirelson);

    \draw[fill=lightgray!90, draw=black, thick, fill opacity=1.0]
        (-2,-2) rectangle (2,2);

    \draw[fill=orange!70, draw=orange!40!black, thick, fill opacity=0.7]
        (0,0) circle (\rhoRadius);

    \draw[<-, thick]
        (-3.2, 0.5) -- (-4.15, 1.2)
        node[left, align=right] {No-Signaling\\Polytope $\mathcal{NS}$};

    \draw[<-, thick]
        ({cos(205)*2.45}, {sin(205)*2.45}) -- (-4.05, -1.2)
        node[left, align=right] {Quantum Set $\mathcal{Q}$};

    \draw[<-, thick]
        (-1.55, -1.55) -- (-3.55, -2.75)
        node[left, align=right] {LHV Polytope $\mathcal{L}$};

    \draw[<-, thick, orange!40!black]
        ({cos(45)*0.68*\rhoRadius}, {sin(45)*0.68*\rhoRadius}) -- (2.65, 2.75)
        node[right, align=left, text=orange!40!black] {$\mathcal{Q}_{\rho}$};

    \node[left] at (-4.0, 4) {4};
    \node[left] at (-4.0, \tsirelson) {$2\sqrt{2}$};
    \node[left] at (-4.0, 2) {2};

    \node[left, text=red] at (-4.0, \rhoRadius) {$2\sqrt{2}\rho$};

\end{tikzpicture}
    \caption{Correlation sets in the $(S,\tilde{S})$ plane. $S$ is the CHSH direction, and $\tilde{S}$ is the CHSH direction with relabeled inputs (i.e., swapping $Q_0$ and $Q_1$). $\Q\rho$ denotes the correlation set obtained when Alice and Bob share the state $\Phi_\rho^{\otimes n}$ and measure traceless observables.}\label{fig correlation set}
\end{figure}

We also consider the 2-out-of-$n$ variant of the CHSH game, where the players are associated with $n$ subsystems, and in each round the verifier samples only two indices and tests the corresponding instance of the game. This framework, introduced by \cite{chao2018test}, provides a way to test large multipartite systems using only local checks. 

In the many-copy setting, such a structure becomes essential. Since the players may perform collective measurements across many shared copies, the underlying tensor-product structure can be obscured. The 2-out-of-$n$ construction restricts such behavior by enforcing consistency across different indices, preventing strategies that spread correlations over many registers. Consequently, it enables the simultaneous certification of multiple independent qubit observables even in the presence of noise, which is difficult to achieve with a single-instance game. 

We now state our main result for the 2-out-of-$n$ CHSH game.
\begin{theorem}[2-out-of-$n$ CHSH, informal]
Let $\Psi$ be a two-qubit state with maximally mixed marginals and quantum maximal correlation $\rho \in (0,1)$. Suppose the players share arbitrarily many i.i.d.\ copies of $\Psi$.

\begin{enumerate}
    \item The optimal quantum winning probability of the 2-out-of-$n$ CHSH game, under traceless observables, is upper bounded by
    \[
    \frac{1}{2}+\frac{\sqrt{2}}{4}\rho.
    \]

    \item Moreover, if a strategy with shared state $\Psi^{\otimes n}$ and traceless observables wins the 2-out-of-$n$ CHSH game with probability at least
    \[
    \frac{1}{2}+\frac{\sqrt{2}}{4}\rho - \epsilon,
    \]
    then there exist $n$ distinct qubits such that for each index $i \in [n]$, the corresponding observables are $O(\sqrt{\epsilon})$-close to Pauli operators $Z$ and $X$ acting on the $i$-th qubit, up to local unitaries. 
\end{enumerate}
\end{theorem}

In~\cite{chao2018test}, the protocol additionally includes a consistency test to certify that the shared state is close to $n$ copies of EPR pairs.
In contrast, in our setting with the trusted but noisy source, the form of the noise acting on the source is assumed to be fixed. Consequently, we do not require an additional consistency test to certify the structure of the shared state. For a constant correlation parameter $\rho \in (0,1)$, the robustness is of the same order as in the noiseless setting~\cite[Theorem 3.2]{chao2018test}. As in the previous case, the register-concentration and simple-unitary properties hold only in the noisy setting.

    


\paragraph{Magic Square games}
For the Magic Square game, we work with the specific source obtained by applying
the depolarizing channel to two EPR pairs.

\begin{definition}[Depolarizing noise Model for Magic Square]
\label{def:magic_square}
We focus on the depolarized noise model acting on
two-EPR-pair state, defined by
\[
\Phi_\rho^{(2)} := \Delta_\rho^{(16)}(\Phi^{\otimes 2})
= \rho \Phi^{\otimes 2} + (1-\rho)\frac{I}{16},
\]
and the players are allowed to share arbitrarily many copies of \(\Phi_\rho^{(2)}\).
\end{definition}

We leave the study of more general noise models for future work.

Our answers to \cref{q1,q2} for the Magic Square game are summarized as follows:

\begin{theorem}[Magic Square, informal]
Let $\Phi_\rho^{(2)}$ be the noisy two-EPR-pair state obtained by applying depolarizing noise with parameter $\rho \in (0,1)$. Suppose the players share arbitrarily many i.i.d.\ copies of $\Phi_\rho^{(2)}$.

\begin{enumerate}
    \item The optimal quantum winning probability of the Magic Square game, under traceless observables, is
    \[
    \frac{1+\rho}{2}.
    \]

    \item Moreover, if a strategy with shared state $\br{\Phi_\rho^{(2)}}^{\otimes n}$ and traceless observables wins the Magic Square game with probability at least
    \[
    \frac{1+\rho}{2} - \epsilon,
    \]
    then the players’ observables are $O(\sqrt{\epsilon})$-close to the  canonical optimal observables defined in \cref{tab:mgs-OptimalMeasurementTable} on a single shared register, up to local unitaries.
\end{enumerate}
\end{theorem}

For a constant correlation parameter $\rho \in (0,1)$, the robustness parameter $O(\sqrt{\epsilon})$ is the same as that in the pure-state setting~\cite[Theorem 7.10]{JNVWY'20}. As in the CHSH case, the presence of noise enforces a more localized structure for observables, enabling the use of simpler unitaries to extract Pauli observables. The Magic Square game can tolerate depolarizing noise of rate $\sqrt{1- 17/18} \approx 0.23$ on each 2-qubit register. 

Similar results hold for the 2-out-of-$n$ Magic Square game:

\begin{theorem}[2-out-of-$n$ Magic Square, informal]\label{thm:2nMS-informal}
Let $\Phi_\rho^{(2)}$ be the noisy two-EPR-pair state obtained by applying depolarizing noise with parameter $\rho \in (0,1)$. Suppose the players share arbitrarily many i.i.d.\ copies of $\Phi_\rho^{(2)}$.

\begin{enumerate}
    \item The optimal quantum winning probability of the 2-out-of-$n$ Magic Square game, under traceless observables, is
    \[
    \frac{1+\rho}{2}.
    \]

    \item Moreover, if a strategy with shared state $\br{\Phi_\rho^{(2)}}^{\otimes n}$ and traceless observables wins the 2-out-of-$n$ Magic Square game with probability at least
    \[
    \frac{1+\rho}{2}-\epsilon,
    \]
    then there exist $n$ distinct four-dimensional registers such that for each index $i \in [n]$, the corresponding observables are $O(\sqrt{\epsilon})$-close to the canonical optimal observables acting on the $i$-th register, up to local unitaries.
\end{enumerate}
\end{theorem}

\paragraph{Applications.} Our results can be used to develop protocols with a trusted but noisy source.
 One application is in the context of quantum networks~\cite{kimble2008quantum,wehner2018quantum}. In such networks, multipartite entangled resource states are generated, distributed, and stored across spatially separated nodes to enable tasks such as distributed quantum computation, networking-assisted sensing, and quantum communication~\cite{wehner2018quantum,zhang2021distributed}. In many realistic scenarios, the quality of the shared entanglement is not known a priori and may vary across network links or over time~\cite{Mor2024Influence}. When measurement devices are trusted, our characterization of the optimal winning probabilities allows us to assess the quality of the shared entangled resources by relating the observed winning probability to bounds on the underlying noise~\cite{brunner2014bell}.
 The other application is measurement device certification. If the noise level of the source is fully characterized, our rigidity results can be used to certify the untrusted Pauli measurements for MDI randomness generation and MDI QKD.
 Furthermore, as the rigidity of the 2-out-of-$n$ Magic Square game is the essential component of the gap-shrinking question reduction technique of \cite{mousavi2022nonlocal}, \cref{thm:2nMS-informal} can be used to establish a similar technique in the noisy setting and paving the way towards bounding the computational power of noisy $\MIP^*_0$. 




\section{Preliminaries}
\subsection{Basic notations}
This paper takes the following notational conventions. For an integer $n\geq 1$, let $[n]$ and $[n]_{\geq 0}$ represent the sets $\set{1,2,...,n}$ and $\set{0,1,...,n-1}$, respectively. Given a set $S$ and a natural number $k$, let $S^k$ be the set $S\times\cdots\times S$, i.e., the $k$-times Cartesian product of $S$. For any $x\in\mathbb{Z}_{\geq 0}^k$, we define $\abs{x}=\abs{\set{i:x_i\neq 0}}$. 
For a complex matrix $M$, we denote its conjugate, transpose, and conjugate transpose by $\bar{M}$, $M^T$, $M^*$, respectively. We use $\| \cdot\|_F$ and $\| \cdot\|$ to represent the Frobenius norm and the spectral norm (or maximum singular value), respectively. We denote $\M_d$ to be the space of all $d\times d$ matrices, and $\H_d$ to be the space of all $d\times d$ Hermitian matrices. Given $M\in\M_d$, we call $M$ {\em traceless} if $\Tr\br{M}=0$. Given $A\in\M_d$, $\overline{\Tr}~A:= \frac{1}{d} \Tr~A$ denotes the normalized trace.
Given an observable $R\in\H_d$, we say $R$ is a \emph{binary observable} if $R^2 = I$.

For convenience, we use $\ket{\Phi}$ and $\Phi$ to denote the EPR state throughout this paper:
\begin{align*}
\ket{\Phi}=\frac{1}{\sqrt{2}}\br{\ket{00}+\ket{11}},\qquad
\Phi=\ketbra{\Phi}.
\end{align*}
For matrices $A, B \in \M_d$, we define the \emph{commutator} and \emph{anti-commutator} as
\[
[A, B] := AB - BA, \quad \{A, B\} := AB + BA.
\]
We say that $A$ and $B$ \emph{commute} if $[A, B] = 0$, and \emph{anti-commute} if $\{A, B\} = 0$. For matrices $A,B\in\H_d$, we write $A\succeq B$ or $B\preceq A$ if $A-B$ is positive semidefinite.

For parameter $\epsilon \to 0^{+}$, we use the following notations for approximate equivalence:
\begin{itemize}
    \item For scalars dependent on $\epsilon$: $a_\epsilon,b_\epsilon$, we write $a_{\epsilon}\approx_\epsilon b_{\epsilon}$ if $|a_{\epsilon}-b_{\epsilon}|=O\br{\epsilon}$.
    \item For vectors $\ket{u_\epsilon}$ and $\ket{v_\epsilon}$, we write $\ket{u_\epsilon}\approx_{\epsilon}\ket{v_{\epsilon}}$ if the Euclidean norm $\|\ket{u_\epsilon}-\ket{v_\epsilon}\|= O\br{\epsilon}$.
    \item For matrices $A_\epsilon,B_\epsilon\in\M_d$, we write $A_\epsilon\approx_{\epsilon} B_\epsilon$ if
    \begin{align}\label{def-approx-matrix}
            \frac{1}{\sqrt{d}}\|A_\epsilon-B_\epsilon\|_F=O\br{\epsilon},
    \end{align}
    \item We say $f(\epsilon)=O(\epsilon)$ for some real function $f$ if there exists an absolute constant $C>0$, such that for sufficiently small $\epsilon$, $\left|f\br{\epsilon}\right|\leq C\epsilon$.
\end{itemize}

The following basic properties follow directly from definitions and standard identities.
\begin{fact}\label{fact:partialtrace}
    For Hermitian matrices $P,Q$ acting on quantum systems $A,B$, respectively, and a bipartite state $\varphi_{AB}$,
    we have
    \[
    \Tr\br{\br{P\x I}\varphi_{AB}}=\Tr\br{P\varphi_A}.
    \]
\end{fact}

\begin{fact}\label{vec-epr}
 Let $P$ be a Hermitian operator on $(\mathbb{C}^2)^{\otimes n}$. Then
   $ P\otimes I\ket{\Phi^{\x n}}=I\otimes P^T\ket{\Phi^{\x n}}$.
\end{fact}
The vectorization map $\mathrm{vec}(\cdot)$ is defined as the linear map that transforms a matrix into a vector by stacking its columns on top of each other. Fact \ref{vec-epr} follows from the observation that $\ket{\Phi^{\otimes n}}$ is the vectorization of $\frac{1}{2^{n/2}} I_{2^n}$, and the identity:
\[
\mathrm{vec}(ABC) = (C^T \otimes A) \, \mathrm{vec}(B),
\]
for matrices $A, B, C$ of compatible dimensions.
\begin{fact}\label{approx-epr}
   Let $\epsilon > 0$. For Hermitian matrices $P, Q$ on $(\mathbb{C}^2)^{\otimes n}$, we have
    \begin{align*}
        P\approx_{\epsilon} Q \Leftrightarrow P\otimes I\ket{\Phi^{\x n}}\approx_{\epsilon}Q\otimes I\ket{\Phi^{\x n}}.
    \end{align*}
\end{fact}
 Fact \ref{approx-epr} follows from Fact \ref{fact:partialtrace} and the fact that the reduced state of $\ket{\Phi^{\x n}}$ is $I/2^n$.

\subsection{Pauli Analysis}
Let $m,n\in \posint$, and set $d = m^n$ as the dimension of the underlying Hilbert space. For matrices $P,Q\in\M_d$, we define the normalized Hilbert–Schmidt inner product as
	\begin{align*}	
	\innerproduct{P}{Q}=\frac{1}{d}\Tr~P^{*}Q.
	\end{align*}
It is straightforward to verify that $\br{\innerproduct{\cdot}{\cdot},\H_d}$ forms a real Hilbert space of Hermitian matrices.

For $m\in \posint$, we say that $\set{\B_i:i\in[m^2]_{\geq 0}}$ is a \emph{standard orthonormal basis} in $\H_m$ if it is an orthonormal basis of $\H_m$, and $\B_0$ is the $m\times m$ identity matrix. For example, the Pauli matrices form a standard orthonormal basis of $\H_2$:
\begin{align*}
&\sigma_0=I=\begin{bmatrix}1 & 0 \\0 & 1\end{bmatrix} ~~~\sigma_1=X=\begin{bmatrix}0 & 1 \\1 & 0\end{bmatrix}\\
&\sigma_2=Y=\begin{bmatrix}0 & -i \\i & 0\end{bmatrix}~~~\sigma_3=Z=\begin{bmatrix}1 & 0 \\0 & -1\end{bmatrix}.
\end{align*}

Given a standard orthonormal basis $\B=\set{\B_i}_{i=0}^{m^2-1}$ in $\H_m$, every matrix $P\in\H_{m}^{\x n}=\H_d$ has a {\em Pauli expansion} with respect to the basis $\B$ given by
\[P=\sum_{x\in[m^2]_{\geq 0}^{n}}\widehat{P}\br{x}\B_{x},\] where $\widehat{P}\br{x}\in \mathbb{R}$ are called \emph{Pauli coefficients}, and $\B_{x}=\bigotimes_{i=1}^n\B_{x_i}$. $\left|x\right|=\left|\set{i:x_i\neq 0}\right|$ is called the \emph{degree} of $\B_x$, and the degree of $P$ is
\begin{align*}
    \textrm{deg}\br{P}:=\max_x \set{\left|x\right|:\hat{P}\br{x}\neq 0}.
\end{align*}

\begin{fact}[Parseval's identity]\label{Parseval's identity}
    For any $P\in\H_d$,
    \[\widebar{\Tr}\br{P^2}=\sum_{x\in[m^2]_{\geq 0}^n}\abs{\widehat{P}\br{x}}^2.\]
\end{fact}

\subsection{Noise Models}

\paragraph{Noise model for CHSH.}

For the CHSH game, we consider a bipartite state $\Psi$ with $\Psi_A=\Psi_B=I/2$, and use its \emph{quantum maximal correlation} to quantify the noise level.

\begin{definition}[Quantum maximal correlation~\cite{Beigi:2013}]\label{def:maximalcorrelation}
	Given quantum systems $A, B$ of dimension $d$ and a bipartite state $\varphi_{AB}$ with $\varphi_A=\varphi_B=\frac{I}{d}$, the quantum maximal correlation of $\Phi_{AB}$ is defined to be
\[
r = \sup \left\{
\left| \Tr\!\left[(P \otimes Q)\varphi_{AB}\right] \right|
:\;
\begin{aligned}
& P, Q \in \H_d, \\
& \Tr P = \Tr Q = 0, \\
& \overline{\Tr}(P^2) = \overline{\Tr}(Q^2) = 1
\end{aligned}
\right\}.
\]
\end{definition}
To describe the noise model more conveniently, we introduce the correlation matrix, which provides a matrix representation of the bipartite correlations.
\begin{definition}[Correlation Matrix~\cite{qin2021nonlocal}]
    For two standard orthonormal bases $\A=\set{\A_i}_{i=0}^{m^2-1}$ and $\B=\set{\B_i}_{i=0}^{m^2-1}$ in $\H_d$, 
we define the correlation matrix $\Corr\br{\varphi,\A,\B}\in\M_{m^2}:$
\begin{align*}
    \Corr\br{\varphi,\A,\B}_{i,j}:=\Tr\br{\varphi \br{\A_i\x\B_j}}
\end{align*}
for $i,j\in [m^2]_{\geq 0}.$
For the $n$-fold tensor product state $\varphi^{\otimes n}$, we define $\Corr(\varphi^{\otimes n}, \A, \B) \in \M_{m^{2n}}$ by
\begin{align*}
    \Corr\br{\varphi^{\x n},\A,\B}_{x,y}=\Tr\br{\varphi^{\x n} \br{\A_x\x\B_y}}
\end{align*} for $x,y\in [m^2]_{\geq 0}^n,$ where $\A_{x}=\bigotimes_{i=1}^n\A_{x_i}$ and $\B_{y}=\bigotimes_{i=1}^n\B_{y_i}$.
\end{definition}

The following facts describe properties of the correlation matrix:

\begin{fact}\cite[Lemma 7.2]{qin2021nonlocal}\label{qylemma7.2}
    $$\Corr\br{\varphi^{\otimes n},\A,\B}=\Corr\br{\varphi,\A,\B}^{\otimes n}.$$
\end{fact}
\begin{fact}\cite[Lemma 7.3]{qin2021nonlocal}\label{qylemma7.3}
    Given a bipartite state $\varphi_{AB}$ with $\varphi_A=\varphi_B=\frac{I}{d}$,
    there exist two standard orthonormal bases $\A$ and $\B$ such that
    \begin{align}
        \Corr\br{\varphi^{\otimes n},\A,\B}_{i,j}=c_i\delta_{i,j},
    \end{align}
    where $c_0=1\geq c_1\geq ...\geq c_{m^{2n}-1}\geq 0$, and $c_1$ equals the maximal correlation of $\varphi_{AB}$.
\end{fact}

Given a bipartite entangled state $\Psi$ with $\Psi_A=\Psi_B=\frac{I}{2}$, by \cref{qylemma7.3} there exist two standard orthonormal bases $\A$ and $\B$ of $\H_2$, and constants $1=c_0\ge c_1\ge c_2\ge c_3\ge 0$, such that 
\[
\Psi
=
\frac14\sum_{i=0}^3 c_i\,\A_i\otimes \B_i.
\]

Here $c_1,c_2$ and $c_3$ quantify the noise in $\Psi$. When $c_1=c_2=c_3=1$, $\Psi$ is equivalent to EPR state, up to a local unitary (\cref{EquivtoEPR}). Otherwise, the coefficients attenuate the non-identity components and therefore capture the effect of noise.

To simplify the analysis, we equivalently push the noise from the state to the observables. Specifically, given $P,Q\in\H_2^{\otimes n}$, if we write
\[
P=\sum_{x\in [4]_{\ge 0}^n}\widehat{P}(x)\,\A_x
\quad\mbox{and}\quad
Q=\sum_{x\in [4]_{\ge 0}^n}\widehat{Q}(x)\,\B_x,
\]
then
\[\Tr\!\left((P\otimes Q)\Psi^{\otimes n}\right)
=
\sum_{x\in [4]_{\ge 0}^n}
\left(\prod_{i=0}^3 c_i^{\,w_i(x)}\right)\widehat{P}(x)\widehat{Q}(x),\]
where 
\[
w_i(x):=\bigl|\{j\in[n]:x_j=i\}\bigr|.
\]

Define the associated noisy observables
\begin{equation}\label{eqn:noisyObservable}
    P' := \sum_{x\in [4]_{\ge 0}^n}
\left(\prod_{i=0}^{3} c_i^{\,w_i(x)}\right)\widehat{P}(x)\,\A_x,
\end{equation}
\begin{equation*}
    Q' := \sum_{x\in [4]_{\ge 0}^n}
\left(\prod_{i=0}^{3} c_i^{\,w_i(x)}\right)\widehat{Q}(x)\,\B_x
\end{equation*}
and by \cref{EquivtoEPR}, we set
\[
\ketbra{\phi}:=\frac14\sum_{i=0}^3 \A_i\otimes \B_i,
\]
we have 
\begin{multline}\label{pre-noise-to-obs}
    \Tr\!\left((P\otimes Q)\Psi^{\otimes n}\right)
=
\bra{\phi^{\otimes n}}P\otimes Q'\ket{\phi^{\otimes n}}
\\=
\bra{\phi^{\otimes n}}P'\otimes Q\ket{\phi^{\otimes n}}.
\end{multline}

The usual depolarizing channel is the special case where all non-identity basis
elements are attenuated by the same factor.

\begin{definition}[Depolarizing channel]
Given $m\in \mathbb{Z}_{>0}$, $\rho\in [0,1]$, and $M\in \H_m$, the
depolarizing channel of dimension $m$ with correlation parameter $\rho$ is
defined by
\[
\Delta_\rho^{(m)}(M)=\rho M+(1-\rho)\Tr(M)\,I_m.
\]
For simplicity, we omit the superscript $(m)$ when the dimension is clear from
context.
\end{definition}

When $m=2$ and $c_1=c_2=c_3=\rho$, the $P'$ defined in \cref{eqn:noisyObservable} is exactly $\Delta_\rho^{(2)}(P).$

\paragraph{Depolarizing noise model for Magic Square.}
For the Magic Square game analysis in \cref{sec:NMS}, we only focus on the depolarizing noise on 4 qubits, i.e., $\Delta_\rho^{(16)}(\Phi^{\otimes 2})$. 
Let
\[
\A=\{I,X,Y,Z\},
\qquad
\B=\{I,X,\overline{Y},Z\},
\]
and define
\[
\A^{\otimes 2}:=\{M\otimes N:M,N\in \A\},
~
\B^{\otimes 2}:=\{M\otimes N:M,N\in \B\}.
\]

Given $P,Q\in\H_4^{\otimes n}$, if we write
\[
P=\sum_{x\in [16]_{\ge 0}^n}\widehat{P}(x)\,\A_x^{\otimes 2},
\qquad
Q=\sum_{x\in [16]_{\ge 0}^n}\widehat{Q}(x)\,\B_x^{\otimes 2},
\]
then
\[
    \Tr\!\left((P\otimes Q)(\Delta_\rho(\Phi^{\otimes 2}))^{\otimes n}\right)
=
\sum_{x\in [16]_{\ge 0}^n}\rho^{|x|}\widehat{P}(x)\widehat{Q}(x).
\]

In this case, we define 
\begin{equation}\label{eqn:noisyObservable_mgs}
    P' := \sum_{x\in [16]_{\ge 0}^n}
\rho^{|x|}\widehat{P}(x)\,\A_x
\end{equation}
and
\begin{equation*}
    Q' := \sum_{x\in [16]_{\ge 0}^n}
\rho^{|x|}\widehat{Q}(x)\,\B_x,
\end{equation*}
then we have
\begin{multline}\label{pre-noise-to-obs-mgs}
    \Tr\!\left((P\otimes Q)(\Delta_\rho(\Phi^{\otimes 2}))^{\otimes n}\right)
=
\Tr\!\left((P\otimes Q')\Phi^{\otimes 2n}\right)
\\=
\Tr\!\left((P'\otimes Q)\Phi^{\otimes 2n}\right).
\end{multline}

\subsection{Nonlocal Games and Rigidity}

We use the following notations in \cite{JNVWY'20} for nonlocal games. 
\begin{definition}[Nonlocal games]
    A nonlocal game $G$ is specified by a tuple $\br{\calX, \calY, \calA, \calB, \mu, V}$ where
    \begin{itemize}
        \item $\calX$ and $\calY$ are finite sets, called the \emph{question sets},
        \item $\calA$ and $\calB$ are finite sets, called the \emph{answer sets},
        \item $\mu$ is a probability distribution over $\calX \times \calY$, called the \emph{question distribution}, and
        \item $V: \calX \times \calY \times \calA \times \calB \to \set{0,1}$ is a function, called the \emph{decision predicate}.
    \end{itemize}
\end{definition}
\begin{definition}[Quantum strategies]
    A quantum strategy $S$ of a nonlocal game $G = \br{\calX, \calY, \calA, \calB, \mu, V}$ is a tuple 
    $\br{\Phi, A, B}$ where
    \begin{itemize}
        \item a bipartite quantum state $\Phi \in \H_A \x \H_B$ for finite dimensional complex Hilbert
spaces $\H_A$ and $\H_B$,
        \item $A$ is a set $\set{A^x}$ such that for every $x \in \calX$, $A^x = \set{A^x_a \mid a \in \calA}$ is 
        a POVM over $\H_A$, and 
        \item $B$ is a set $\set{B^y}$ such that for every $y \in \calY$, $B^y = \set{B^y_b \mid b \in \calB}$ is 
        a POVM over $\H_B$.
    \end{itemize}
\end{definition}
\begin{definition}[Quantum value]
    The quantum value of a quantum strategy $S = \br{\varphi, A ,B}$ for a nonlocal game
    $G = \br{\calX, \calY, \calA, \calB, \mu, V}$ is defined as
    \begin{align*}
        \val^\ast\br{G,S} = \sum_{x,y,a,b} \mu\br{x,y}V\br{x,y,a,b} \Tr\br{\br{A^x_a \x B^y_b}\varphi}.
    \end{align*}
    For $v \in [0,1]$ we say that the strategy passes or wins $G$ with probability $v$ if $\val^\ast\br{G,S}\geq v$.
    The optimal quantum value of $G$ is defined as
    \begin{align*}
        \val^*\br{G} = \sup_{S} \val^\ast\br{G,S}
    \end{align*}
    where the supremum is taken over all quantum strategies $S$ for $G$.
\end{definition}

One important and well-studied nonlocal game is the CHSH game introduced in \cite{clauser1969proposed}:
\begin{definition}(CHSH game)\label{def-chsh}
    The CHSH game $G$ is defined by the tuple $\br{\calX, \calY, \calA, \calB, \mu, V}$ where
    \begin{itemize}
        \item $\calX=\calY=\set{0,1}$, $\calA =\calB=\set{-1,1}$,
        \item $\mu\br{x,y}=\frac{1}{4},\forall x,y\in\set{0,1}$,
        \item $V\br{x,y,a,b}=1$ if and only if $\frac{1+a}{2}\oplus\frac{1+b}{2}=xy$.
    \end{itemize}
\end{definition}

When Alice and Bob share a quantum state $\varphi$ and receive questions $x$ and $y$, they measure with observables $P_x$ and $Q_y$, respectively. The violation of CHSH game is defined as:
\begin{align*}
    v=\Tr\br{\br{P_0\x Q_0+P_0\x Q_1+P_1\x Q_0-P_1\x Q_1}\varphi}.
\end{align*}
The violation $v$ and winning probability $\omega$ under the same strategy satisfy 
$\omega=\frac{1}{2}+\frac{v}{8}$.
The maximal violation of the CHSH game was proven to be $2\sqrt{2}$ \cite{tsirel1987quantum}. An optimal quantum strategy achieves this sharing EPR state $\Phi$ and using observables $\br{P_0,P_1,Q_0,Q_1}=\br{Z,X,\frac{Z+X}{\sqrt{2}},\frac{Z-X}{\sqrt{2}}}$.

Rigidity is a key property of certain nonlocal games. It certifies that when a strategy achieves a near-optimal violation, the shared state and measurements are close to the ideal strategy (up to local isometries). The rigidity of CHSH game is shown in \cite[Theorem 2]{mckague2012robust}. A standard proof of the rigidity of CHSH first establishes anti-commutation relations between observables, and then constructs isometries via a swap method.


Another well-studied nonlocal game is the Magic Square game \cite{mermin1990simple,peres1990incompatible,aravind2004quantum}.

\begin{definition}[Magic Square Game]\label{def-mgs-game}
    Consider a $3\times 3$ matrix of Boolean variables 
     \begin{center}
    \begin{tabular}{|c|c|c|}
            \hline
        $s_{1,1}$&$s_{1,2}$&$s_{1,3}$\\\hline
        $s_{2,1}$&$s_{2,2}$&$s_{2,3}$\\\hline
        $s_{3,1}$&$s_{3,2}$&$s_{3,3}$\\\hline
    \end{tabular}
    \end{center}
In the Magic Square game, the question sets are $\calX=\{r_1,r_2,r_3,c_1,c_2,c_3\}$ and $\calY=\set{s_{i,j}:i,j\in[3]}$, and the answer sets are $\calA=\{-1,1\}^3$ and $\calB=\{-1,1\}$. The referee uniformly samples $x\in \calX$ and sends $x$ to Alice, where $r_i$ (resp. $c_j$) stands for the $i$th row $\{s_{i,1},s_{i,2},s_{i,3}\}$ (resp. $j$th column $\{s_{1,j},s_{2,j},s_{3,j}\}$). Then the referee samples a variable $y\in x$ uniformly at random and sends $y$ to Bob. 
 
 Alice returns an assignment $a=(a_1,a_2,a_3)\in\calA$ for the three variables in $x$, while Bob returns an assignment $b\in\calB$ for the variable $y$.
 
    Players win the game if and only if both of the following conditions are satisfied:
    \begin{enumerate}
        \item (\textbf{Parity Test}) If $x\in \{r_1,r_2,r_3,c_1,c_2\}$, then $a_1a_2a_3=1$; otherwise (i.e., $x = c_3$), $a_1a_2a_3=-1$.
        \item (\textbf{Consistency Test}) Alice’s and Bob’s assignments for $y$ agree.
    \end{enumerate}
\end{definition}

An optimal quantum strategy that achieves a perfect winning probability of 1 uses two EPR pairs shared between the players, along with the measurement settings as in \cref{tab:mgs-OptimalMeasurementTable}.
The Magic Square game also exhibits rigidity; see \cite{wu2016device} for further discussion.
\begin{table}[ht]
    \centering
    \begin{tabular}{|c|c|c|}
        \hline
        $X\otimes I$ & $I\otimes X$ & $X\otimes X$ \\
        \hline
        $I\otimes Z$ & $Z\otimes I$ & $Z\otimes Z$ \\
        \hline
        $X\otimes Z$ & $Z\otimes X$ & $Y\otimes Y$ \\
        \hline
    \end{tabular}
    \caption{Optimal Measurement Table for the Magic Square Game}
    \label{tab:mgs-OptimalMeasurementTable}
\end{table}

The \(2\)-out-of-\(n\) CHSH game is defined in~\cite[Section 3.2]{chao2018test}. In each round, the referee chooses uniformly among three tests. In the consistency test, the referee selects one index from \(\set{1,\ldots,n}\) and sends the same indexed CHSH question to both players, requiring them to give the same answer. In the CHSH test, the referee selects two indices from \(\set{1,\ldots,n}\) at random and sends both to one player (together with two CHSH questions) and one of them to the other player (together with one CHSH question). The players return answers for each index received, and their answers are tested for the CHSH predicate on the shared index. The third test is defined symmetrically by exchanging the roles of Alice and Bob. To attain the optimal winning probability \(\frac{2+\sqrt{2}}{4}\), the players can share \(n\) EPR pairs and apply the optimal CHSH observables from the standard CHSH game on the \(i\)-th EPR pair whenever they receive index \(i\). 

The 2-out-of-$n$ Magic Square game is defined analogously as in \cite[Section 3.2]{mousavi2022nonlocal}.
\section{CHSH games under noise}
\subsection{CHSH}\label{self-test-depolarize}
This section focuses on the effect of noise on the CHSH game \cref{def-chsh}, where the players share $n$ copies of states $\Psi^{\x n}$. Here the bipartite state $\Psi$ satisfies the following assumption:
\begin{assumption}\label{share-state-assumption}
    The maximal correlation of $\Psi$ is upper bounded by a constant $\rho \in (0,1)$ and $\Psi_A=\Psi_B=I/2$.
\end{assumption}
By \cref{qylemma7.3}, there exists standard orthonormal basis $\set{\A_i:0\leq i\leq 3},\set{\B_i:0\leq i\leq 3}$ such that the correlation matrix of $\Psi$ is $\mathrm{diag}\set{c_0=1,c_1\leq \rho,c_2,c_3}$, where $ c_0\geq c_1\geq c_2\geq c_3\geq 0$, i.e.,
\begin{align}\label{Psi_cor}
\Psi=\sum_{i=0}^3 c_i \A_i\otimes \B_i.
\end{align}
The players' measurements are denoted by $P_0, P_1, Q_0, Q_1$. 
Define
\begin{align}\label{def-chsh-max-trace}
\Trerr=\max \set{
\left|\widebar{\Tr}(P_0)\right|,
\left|\widebar{\Tr}(P_1)\right|,
\left|\widebar{\Tr}(Q_0)\right|,
\left|\widebar{\Tr}(Q_1)\right|
}.
\end{align}
Since the local marginals of \(\Psi^{\otimes n}\) are maximally mixed, the normalized trace of an
observable is exactly the bias of its output distribution. In other words, $\Trerr$ can be estimated statistically by a simple \emph{trace test}. One may repeat the CHSH game under noise many times, and we assume that each repetitions are i.i.d. as follows.
\begin{assumption}\label{repetition-iid}
Alice and Bob use fixed observables and share fresh copies of \(\Psi^{\otimes n}\) in each repetition.
\end{assumption}
To estimate the traces, we record for each player and each question the empirical frequency of the answer \(1\), and reject if
this frequency deviates from \(1/2\) by more than a prescribed threshold. Therefore, concentration bounds imply that
passing this test with high probability forces \(\Trerr\) to be small. We defer the precise relation
between the number of repetitions, the threshold, and the failure probability to \cref{app:chsh-trace-test}.

Therefore, we will simply work under the assumption that the
observables have bounded normalized trace.
We first give an upper bound of the violation.
\begin{theorem}\label{traceless-chsh-vio-upperbound}
    Suppose the players share $\Psi^{\x n}$ satisfying \cref{share-state-assumption} then the violation of the CHSH game under noise is at most $2\sqrt{2}\rho+\frac{\sqrt{2}\Trerr^2}{\rho}$, where $\Trerr$ is defined in \cref{def-chsh-max-trace}. 
\end{theorem}

\begin{proof}
Define \begin{align*}
    C=&P_0\otimes \br{Q_0'+Q_1'}+P_1\otimes\br{Q_0'-Q_1'}\\
    =&-\frac{\rho}{\sqrt{2}}\br{ P_0\otimes I-I\otimes\frac{Q_0'+Q_1'}{\sqrt{2}\rho}}^2\\
    &-\frac{\rho}{\sqrt{2}}\br{ P_1\otimes I-I\otimes \frac{Q_0'-Q_1'}{\sqrt{2}\rho}}^2\\
    &+I\otimes\frac{\br{Q_0'}^2+\br{Q_1'}^2}{\sqrt{2}\rho}+\frac{\rho}{\sqrt{2}}P_0^2\otimes I+\frac{\rho}{\sqrt{2}}P_1^2\otimes I.
\end{align*}  
For each $j\in\set{0,1}$, the trace bound implies
\begin{multline}\label{almost-deg-one-1}
   \bra{\phi^{\otimes n}}I\otimes Q_j'^2 \ket{\phi^{\otimes n}}= \widebar{\Tr}\br{Q_j'^2}\\\leq \hat{Q}_j\br{0^n}^2+\sum_{x\neq 0^n}\rho^{2|x|}\hat{Q}_j\br{x}^2 \leq \Trerr^2+\rho^2.
\end{multline}
Also $P_0^2\preceq I,P_1^2\preceq I$. By \cref{pre-noise-to-obs}, the violation is $\bra{\phi^{\otimes n}}C\ket{\phi^{\otimes n}}$, upper bounded by $2\sqrt{2}\rho+\frac{\sqrt{2}\Trerr^2}{\rho}$.
\end{proof}

\begin{remark}
    In \cref{traceless-chsh-vio-upperbound}, if $c_1=c_2=\rho$ and $\Trerr=0$, then we can get a tight upper bound $2\sqrt{2}\rho$ for the violation, and the maximal violation is achieved if $P_0=\A_1,P_1=\A_2,Q_0=\br{\B_1+\B_2}/\sqrt{2},Q_1=\br{\B_1-\B_2}/\sqrt{2}$, for $\set{\A_i},\set{\B_j}$ defined in \cref{Psi_cor}.
\end{remark}

The following theorem demonstrates that a near-optimal violation, combined with trace-bounded observables implies strong structural properties, which can be used to certify observables and properties of $\Psi$.

\begin{theorem}\label{chsh-selftest-main-thm}
    Suppose the players share $\Psi^{\x n}$ for an arbitrary $n$, where $\Psi$ satisfies \cref{share-state-assumption}. The players' measurements attain a violation $2\sqrt{2}\rho-\epsilon_v> 2$ in the CHSH game under noise.
    Define parameters
    \begin{align*}
        \gamma =\frac{1}{1-\rho}> 1,\qquad\epsilon=\frac{\sqrt{2}}{\rho}\epsilon_v+\frac{2}{\rho^2}\Trerr^2 .
    \end{align*}
    Then
    \begin{itemize}
        \item $\rho\geq c_1\geq c_2\geq \rho-O(\sqrt{\epsilon}\gamma)$,
    \end{itemize} 
    and there exists a register index $k\in[n]$, satisfying:
    \begin{itemize}
        \item (Register-concentration) For $D\in\set{P,Q},j\in\set{0,1}$, there exist 2-dimensional traceless observables $\tilde{D}_j^{\br{k}}$, such that
    \begin{align}
        D_j\approx_{\sqrt{\epsilon}\gamma} \tilde{D}_j^{\br{k}} \otimes I^{[n]\backslash\set{k}}.
    \end{align}
    \item (Simple Unitaries) There exists unitaries $U_A,U_B$ acting non-trivially only on the $k$-th register, such that
\begin{align*}
    &U_AP_0U_A^*\approx_{\sqrt{\epsilon}\gamma}  Z^{\br{k}}\otimes I^{[n]\backslash\set{k}}, \\
    &U_AP_1U_A^*\approx_{\sqrt{\epsilon}\gamma}  X^{\br{k}}\otimes I^{[n]\backslash\set{k}},\\
     &U_BQ_0U_B^*\approx_{\sqrt{\epsilon}\gamma}  Z^{\br{k}}\otimes I^{[n]\backslash\set{k}}, \\
     &U_BQ_1U_B^*\approx_{\sqrt{\epsilon}\gamma}  X^{\br{k}}\otimes I^{[n]\backslash\set{k}}.
\end{align*}
    \end{itemize}
\end{theorem}
The first conclusion of \cref{chsh-selftest-main-thm} $c_1\approx c_2$ indicates the effect of noise on the shared state. 
The second conclusion implies that the observables are concentrated on a single register. In contrast, for the noiseless case, an optimal measurement may measure on multiple registers. That is, adding noise forces the players to use only one pair of states to play the game. The third conclusion shows that to extract Pauli $Z$ and $X$, the unitaries act only on a single qubit, as opposed to the whole system in noiseless cases.
Effectively, our rigidity result in \cref{chsh-selftest-main-thm} certifies both the state and the measurements used by the players as in the noiseless rigidity result.

The final rigidity error is $\sqrt{\epsilon}\gamma=O\br{\sqrt{\epsilon_v+\Trerr^2}/\br{\rho\br{1-\rho}}}$.
In both the noiseless case and noisy case, when the winning violation is $\epsilon$-close to optimal, the rigidity robustness parameter scales as $O(\sqrt{\epsilon})$ with a constant depolarizing noise level. 
However for noiseless case, the last two conclusions (register-concentration and simple unitaries) as listed in \cref{chsh-selftest-main-thm} are not satisfied as $\gamma\to\infty$. For example, the measurement using observables $\br{Z,X,\frac{Z+X}{\sqrt{2}},\frac{Z-X}{\sqrt{2}}}^{\x 3}$ acting on three EPR pairs is also optimal, but violates the last two conclusions.

To prove \cref{chsh-selftest-main-thm}, we will first prove these two propositions:
\begin{prop}[Observable Scaling]\label{thm-close-to-deg1}
     Under the assumptions in \cref{chsh-selftest-main-thm}, if Alice and Bob's measurements attain a violation $2\sqrt{2}\rho-\epsilon_v>2$ in CHSH game under noise, then for $D\in \set{P,Q},i\in\set{0,1}$,
    \begin{align}
        D_i'\approx_{\sqrt{\epsilon}} \rho D_i,
    \end{align}
    where $\epsilon=\frac{\sqrt{2}}{\rho}\epsilon_v+\frac{2}{\rho^2}\Trerr^2$ and $\Trerr$ is defined in \cref{def-chsh-max-trace}.
\end{prop}
\begin{proof}
 Since the violation is greater than two, $\Psi$ is not a separable state. By \cref{EquivtoEPR}, $\frac{1}{4}\sum_{i=0}^3\A_i\otimes  \B_i$ is a pure state, we set
    \begin{align*}
        \ketbra{\phi}=\frac{1}{4}\sum_{i=0}^3\A_i\otimes  \B_i.
    \end{align*}
    Since the partial traces of $\ketbra{\phi^{\x n}}$ on both sides are $I/2^n$, for any Hermitian matrices $P,Q\in \H_{2^n}$,
    \begin{align*}
        P\approx_{\epsilon} Q \Leftrightarrow P\otimes I\ket{\phi^{\x n}}\approx_{\epsilon}Q\otimes I\ket{\phi^{\x n}}.
    \end{align*}
By \cref{pre-noise-to-obs}, we have $\bra{\phi^{\otimes n}}C_1\ket{\phi^{\otimes n}}=\bra{\phi^{\otimes n}}C_2\ket{\phi^{\otimes n}}=\epsilon_v$, where
\begin{align*}
    C_1=&2\sqrt{2}\rho I\x I -\br{P_0\otimes \br{Q_0'+Q_1'}+P_1\otimes\br{Q_0'-Q_1'}}\\
    =&\frac{\rho}{\sqrt{2}}\br{ P_0\otimes I-I\otimes\frac{Q_0'+Q_1'}{\sqrt{2}\rho}}^2\\
    &+\frac{\rho}{\sqrt{2}}\br{ P_1\otimes I-I\otimes \frac{Q_0'-Q_1'}{\sqrt{2}\rho}}^2\\
    &+\sqrt{2}\rho I\otimes I-I\otimes\frac{\br{Q_0'}^2+\br{Q_1'}^2}{\sqrt{2}\rho}\\
    &+\frac{\rho}{\sqrt{2}}\br{I\otimes I-P_0^2\otimes I}+\frac{\rho}{\sqrt{2}}\br{I\otimes I-P_1^2\otimes I},\\
    C_2=&2\sqrt{2}\rho I\x I -\br{P_0'\otimes \br{Q_0+Q_1}+P_1'\otimes\br{Q_0-Q_1}}\\
    =&\frac{\rho}{\sqrt{2}}\br{ I\otimes Q_0-\frac{P_0'+P_1'}{\sqrt{2}\rho}\otimes I}^2\\
    &+\frac{\rho}{\sqrt{2}}\br{ I\otimes Q_1-\frac{P_0'-P_1'}{\sqrt{2}\rho}\otimes I}^2\\
    &+\sqrt{2}\rho I\otimes I-\frac{\br{P_0'}^2+\br{P_1'}^2}{\sqrt{2}\rho}\otimes I\\
    &+\frac{\rho}{\sqrt{2}}\br{I\otimes I-I\otimes Q_0^2}+\frac{\rho}{\sqrt{2}}\br{I\otimes I-I\otimes Q_1^2}.
\end{align*}
Since the square terms in $C_1$ and $I-P_i$ are positive semidefinite, their inner-products with $\ketbra{\phi}^{\otimes n}$ are non-negative. Therefore,
\begin{align}\label{almost-deg-one-0}
    \bra{\phi^{\otimes n}} \br{\sqrt{2}\rho I\otimes I-I\otimes\frac{\br{Q_0'}^2+\br{Q_1'}^2}{\sqrt{2}\rho}}\ket{\phi^{\otimes n}}\leq \epsilon_v.
\end{align}
Combining \cref{almost-deg-one-1} and \cref{almost-deg-one-0}, we have for $j\in\set{0,1}$,
\begin{align}\label{chsh-self-test-eq-1}
    \bra{\phi^{\otimes n}}I\otimes Q_j'^2 \ket{\phi^{\otimes n}} \geq \rho^2-\sqrt{2}\rho\epsilon_v-\Trerr^2.
\end{align}
Combining \cref{almost-deg-one-1} and \cref{chsh-self-test-eq-1}, we have
\begin{align}\label{almost-deg-one-2}
    \bra{\phi^{\otimes n}}I\otimes Q_j'^2 \ket{\phi^{\otimes n}}\approx_{\rho^2\epsilon} \rho^2
\end{align}
By \cref{almost-deg-one-1}, $$\bra{\phi^{\otimes n}}\br{\sqrt{2}\rho I\otimes I-I\otimes\frac{\br{Q_0'}^2+\br{Q_1'}^2}{\sqrt{2}\rho}}\ket{\phi^{\otimes n}}\geq -\frac{\sqrt{2}\Trerr^2}{\rho}.$$ Then, we have for $i\in\set{0,1},$
\begin{align*}
    \bra{\phi^{\otimes n}}\br{I\otimes I-P_i^2\otimes I}\ket{\phi^{\otimes n}} \leq \frac{\sqrt{2}}{\rho}\epsilon_v+\frac{2\Trerr^2}{\rho^2}=\epsilon,
\end{align*}
which gives
\begin{align*}
    \bra{\phi^{\otimes n}}P_i^2\otimes I\ket{\phi^{\otimes n}}\approx_\epsilon 1.
\end{align*} 
A symmetric analysis of $C_2$ gives
\begin{align}\label{almost-deg-one-3}
    \bra{\phi^{\otimes n}}I\otimes Q_j^2\ket{\phi^{\otimes n}}\approx_\epsilon 1.
\end{align}
From \cref{almost-deg-one-2}, \cref{almost-deg-one-3} and \cref{general-obs-scaling} (where we set $m=2,r=\rho $), we have 
\begin{align*}
    Q_j'\approx_{\sqrt{\epsilon}} \rho Q_j.
\end{align*}
Similarly, if we analyze $C_2$ we also have $P_i'\approx_{\sqrt{\epsilon}} \rho P_i$ for $i\in\set{0,1}$.
\end{proof}
Based on $C_1,C_2$ in \cref{thm-close-to-deg1}, we can also derive the anti-commutation relations of observables. 
\begin{prop}[Anti-commutation Relation]\label{cor-almost-anti-commute}

    Under the assumptions in \cref{chsh-selftest-main-thm}, if Alice and Bob's measurements attain a violation of $2\sqrt{2}\rho-\epsilon_v>2$ in the CHSH game under noise, for $D\in \set{P,Q}$
    \begin{align}
        D_0D_1\approx_{\sqrt{\epsilon}} -D_1D_0
    \end{align} where $\epsilon=\frac{\sqrt{2}}{\rho}\epsilon_v+\frac{2}{\rho^2}\Trerr^2$, and $\Trerr$ is defined in \cref{def-chsh-max-trace}.
\end{prop}
\begin{proof}
    For $C_1,C_2$ defined in \cref{thm-close-to-deg1}, we have
    \begin{align*}
        \bra{\phi^{\otimes n}}\br{ P_i\otimes I-I\otimes\frac{Q_0'+\br{-1}^iQ_1'}{\sqrt{2}\rho}}^2\ket{\phi^{\otimes n}}\approx_{\epsilon} 0,
    \end{align*}
    which is equivalent to
    \begin{align*}
        P_i\otimes I\ket{\phi^{\otimes n}}\approx_{\sqrt{\epsilon}} I\otimes\frac{Q_0'+\br{-1}^iQ_1'}{\sqrt{2}\rho}\ket{\phi^{\otimes n}}.
    \end{align*}
    By \cref{thm-close-to-deg1}, $Q_i'\approx_{\sqrt{\epsilon}}\rho Q_i$, and by \cref{EquivtoEPR}, there exists unitary $V$, such that $\ket{\phi}=(I \otimes V^T)\ket{\Phi}$. 
    Combining \cref{approx-epr}, we have
\begin{align}\label{player-strategy-relation}
    P_i&\approx_{\sqrt{\epsilon}} V^{\otimes n}\br{\frac{Q_0+\br{-1}^iQ_1}{\sqrt{2}}} \br{V^*}^{\otimes n}.
    \end{align}
    Therefore
    \begin{align*}       P_0P_1\approx_{\sqrt{\epsilon}}V^{\otimes n} \frac{Q_0^2-Q_1^2-Q_0Q_1+Q_1Q_0}{2} \br{V^*}^{\x n},\\
    P_1P_0\approx_{\sqrt{\epsilon}}V^{\otimes n}\frac{Q_0^2-Q_1^2+Q_0Q_1-Q_1Q_0}{2}\br{V^*}^{\x n}.
    \end{align*}
    By \cref{closetobservable}, $ Q_0^2\approx_{\sqrt{\epsilon}}  Q_1^2\approx_{\sqrt{\epsilon}} I$. 
    Thus
    \begin{align*}
        P_0P_1\approx_{\sqrt{\epsilon}}V^{\otimes n}\frac{-Q_0Q_1+Q_1Q_0}{2}\br{V^*}^{\x n}\approx_{\sqrt{\epsilon}} -P_1P_0.
    \end{align*}
    A symmetric analysis gives that $Q_0Q_1\approx_{\sqrt{\epsilon}}- Q_1Q_0$.
\end{proof}

\medskip
\noindent

We now state several auxiliary lemmas that together establish the structural conclusions of \cref{chsh-selftest-main-thm}. 
\cref{lem:single-observable-localization} shows that each observable is essentially supported on a single register. 
\cref{lem:same-register} ensures that for each player, the two observables act on the same register. \cref{lem:register-matching} further aligns Alice’s and Bob’s registers. 
The proofs are deferred to \cref{appendix-lem-supporting-chsh-self-test}.
\begin{lemma}\label{lem:single-observable-localization}
Assume the hypotheses of \cref{chsh-selftest-main-thm}. Then for every \(D\in\{P,Q\}\) and \(j\in\{0,1\}\), there exist an index
\[
k_{D,j}\in[n]
\]
and a traceless binary observable
\[
\widetilde{D}_j^{(k_{D,j})}\in \H_2
\]
such that
\[
D_j \approx_{\sqrt{\epsilon}\gamma}
\widetilde{D}_j^{(k_{D,j})}\otimes I^{[n]\setminus\{k_{D,j}\}}.
\]
\end{lemma}

\begin{lemma}\label{lem:same-register}
Let \(D\in\{P,Q\}\). Assume
\[
D_0 \approx_{\sqrt{\epsilon}\gamma}
\widetilde{D}_0^{(k_{D,0})}\otimes I^{[n]\setminus\{k_{D,0}\}},
\\
D_1 \approx_{\sqrt{\epsilon}\gamma}
\widetilde{D}_1^{(k_{D,1})}\otimes I^{[n]\setminus\{k_{D,1}\}},
\]
where \(\widetilde{D}_0^{(k_{D,0})},\widetilde{D}_1^{(k_{D,1})}\in\H_2\) are traceless binary observables, and suppose moreover that
\[
\{D_0,D_1\}\approx_{\sqrt{\epsilon}}0.
\]
Then
\[
k_{D,0}=k_{D,1}.
\]
\end{lemma}

\begin{lemma}\label{lem:register-matching}
Assume
\[
P_i \approx_{\sqrt{\epsilon}\gamma}
\widetilde{P}_i^{(k_1)}\otimes I^{[n]\setminus\{k_1\}},
\qquad
Q_j \approx_{\sqrt{\epsilon}\gamma}
\widetilde{Q}_j^{(k_2)}\otimes I^{[n]\setminus\{k_2\}},
\]
for \(i,j\in\{0,1\}\), where all local observables are traceless and binary. If there exists a two-dimensional unitary \(V\) such that
\[
P_i \approx_{\sqrt{\epsilon}\gamma}
V^{\otimes n}
\left(
\frac{Q_0+(-1)^iQ_1}{\sqrt{2}}
\right)
(V^*)^{\otimes n},
\qquad i\in\{0,1\},
\]
then
\[
k_1=k_2.
\]
\end{lemma}

It remains to extract information about the source. 
Once the observables are localized on a common register, we analyze their expansions in the basis that diagonalizes the correlation matrix. 
The near-optimality conditions then constrain the corresponding coefficients, which leads to the following lemma. 

\begin{lemma}[Source coefficient rigidity]\label{lem:source-coefficient-rigidity}
Assume the hypotheses of \cref{chsh-selftest-main-thm}. Then
\[
\rho \ge c_1 \ge c_2 \ge \rho - O(\sqrt{\epsilon}\gamma).
\]
\end{lemma}

The proof is deferred to \cref{appendix-lem-supporting-chsh-self-test}. Now we are ready to prove \cref{chsh-selftest-main-thm}.
\begin{proof}[Proof of \cref{chsh-selftest-main-thm}]

By \cref{thm-close-to-deg1,lem:single-observable-localization}, each of
$P_0,\ P_1,\ Q_0,\ Q_1$
is \(O(\sqrt{\epsilon}\gamma)\)-close to a traceless binary observable acting on a single register. Thus there exist indices
\[
k_{P,0},\ k_{P,1},\ k_{Q,0},\ k_{Q,1}\in[n]
\]
and traceless binary observables
\[
\widetilde{P}_0^{(k_{P,0})},\ \widetilde{P}_1^{(k_{P,1})},\
\widetilde{Q}_0^{(k_{Q,0})},\ \widetilde{Q}_1^{(k_{Q,1})}\in\H_2
\]
such that
\begin{align*}
P_0 &\approx_{\sqrt{\epsilon}\gamma}
\widetilde{P}_0^{(k_{P,0})}\otimes I^{[n]\setminus\{k_{P,0}\}},\\
P_1 &\approx_{\sqrt{\epsilon}\gamma}
\widetilde{P}_1^{(k_{P,1})}\otimes I^{[n]\setminus\{k_{P,1}\}},\\
Q_0 &\approx_{\sqrt{\epsilon}\gamma}
\widetilde{Q}_0^{(k_{Q,0})}\otimes I^{[n]\setminus\{k_{Q,0}\}},\\
Q_1 &\approx_{\sqrt{\epsilon}\gamma}
\widetilde{Q}_1^{(k_{Q,1})}\otimes I^{[n]\setminus\{k_{Q,1}\}}.
\end{align*}

By \cref{cor-almost-anti-commute,lem:same-register}, we obtain
\[
k_{P,0}=k_{P,1}=:k_1,
\qquad
k_{Q,0}=k_{Q,1}=:k_2.
\]

Moreover, by the argument leading to \cref{cor-almost-anti-commute}, there exists a two-dimensional unitary \(V\) such that
\[
P_i \approx_{\sqrt{\epsilon}\gamma}
V^{\otimes n}
\left(
\frac{Q_0+(-1)^iQ_1}{\sqrt{2}}
\right)
(V^*)^{\otimes n},
\qquad i\in\{0,1\}.
\]
Hence \cref{lem:register-matching} implies
\[
k_1=k_2=:k.
\]

Therefore, for each \(D\in\{P,Q\}\), the localized observables on the \(k\)-th register satisfy
\[
\widetilde{D}_0^{(k)}\widetilde{D}_1^{(k)}
+
\widetilde{D}_1^{(k)}\widetilde{D}_0^{(k)}
\approx_{\sqrt{\epsilon}\gamma}0,
\qquad
\bigl(\widetilde{D}_j^{(k)}\bigr)^2
\approx_{\sqrt{\epsilon}\gamma}I.
\]
By \cref{anti-commute-to-ZX}, there exist two-dimensional unitaries
\[
\widetilde{U}_A,\widetilde{U}_B:\H_2\to\H_2
\]
such that
\begin{align*}
\widetilde{U}_A\widetilde{P}_0^{(k)}\widetilde{U}_A^*
&\approx_{\sqrt{\epsilon}\gamma} Z,
&
\widetilde{U}_A\widetilde{P}_1^{(k)}\widetilde{U}_A^*
&\approx_{\sqrt{\epsilon}\gamma} X,\\
\widetilde{U}_B\widetilde{Q}_0^{(k)}\widetilde{U}_B^*
&\approx_{\sqrt{\epsilon}\gamma} Z,
&
\widetilde{U}_B\widetilde{Q}_1^{(k)}\widetilde{U}_B^*
&\approx_{\sqrt{\epsilon}\gamma} X.
\end{align*}
Define
\[
U_A:=\widetilde{U}_A^{(k)}\otimes I^{[n]\setminus\{k\}},
\qquad
U_B:=\widetilde{U}_B^{(k)}\otimes I^{[n]\setminus\{k\}}.
\]
Then \(U_A\) and \(U_B\) satisfy the measurement conclusion of the theorem.

Finally, \cref{lem:source-coefficient-rigidity} gives
\[
\rho \ge c_1 \ge c_2 \ge \rho - O(\sqrt{\epsilon}\gamma).
\]
This completes the proof.

\end{proof}

\subsection{Implication on Correlation Set}


Our analysis of the maximal CHSH value extends to other Bell directions as well, and hence
gives a description of the projected correlation set \(\mathcal{Q}_\rho\) in certain
two-dimensional planes. Here we consider the \((S,\tilde S)\)-plane, where \(S\) is the
standard CHSH direction and \(\tilde S\) is obtained by exchanging Bob's two inputs.
For traceless observables \(P_0,P_1,Q_0,Q_1\), define
\begin{align*}
S
&=
P_0' \otimes (Q_0+Q_1)
+
P_1' \otimes (Q_0-Q_1),\\
\tilde S
&=
P_0' \otimes (Q_0+Q_1)
+
P_1' \otimes (Q_1-Q_0).
\end{align*}

If we write the correlators as
\[
 E_{xy}:=\Tr\bigl((P_x\otimes Q_y)\Psi^{\otimes n}\bigr),
 \qquad x,y\in\{0,1\},
 \]
define
\[
s:=E_{00}+E_{01}+E_{10}-E_{11},
\qquad
\tilde s:=E_{00}+E_{01}-E_{10}+E_{11}.
\]
By \cref{pre-noise-to-obs}, this is equivalent to
\[
s=\bra{\phi^{\otimes n}}S\ket{\phi^{\otimes n}},
\qquad
\tilde s=\bra{\phi^{\otimes n}}\tilde S\ket{\phi^{\otimes n}},
\]
where $\ket{\phi}$ is equivalent to $\ket{\Phi}$ up to local unitaries (\cref{EquivtoEPR}). 
Thus, in the correlator space with coordinates \((E_{00},E_{01},E_{10},E_{11})\),
the two Bell directions \(S\) and \(\tilde S\) correspond to the coefficient vectors
\[
(1,1,1,-1),
\qquad
(1,1,-1,1),
\]
whose Euclidean inner product is \(0\). Hence \(S\) and \(\tilde S\) are orthogonal directions in the correlator space.

Now set
\[
a:=\bra{\phi^{\otimes n}}P_0'\otimes(Q_0+Q_1)\ket{\phi^{\otimes n}},\]
\[
b:=\bra{\phi^{\otimes n}}P_1'\otimes(Q_0-Q_1)\ket{\phi^{\otimes n}}.
\]
Then
\[
s=a+b,
\qquad
\tilde s=a-b,
\]
and therefore
\[
s^2+\tilde s^2=2(a^2+b^2).
\]
Since \(P_0\) and \(P_1\) are traceless, their degree-zero Pauli coefficients vanish. Hence, by \cref{almost-deg-one-1}
\begin{align*}
\widebar{\Tr}\bigl((P_x')^2\bigr)
\le \rho^2,
\qquad x\in\{0,1\}.
\end{align*}
For any Hermitian operators \(A,B\), suppose that $\ket{\phi^{\otimes n}}=(U\otimes V)\ket{\Phi^{\otimes n}}$, we have
\begin{multline*}
    \Bigl|\bra{\phi^{\otimes n}}A\otimes B\ket{\phi^{\otimes n}}\Bigr|
=\Bigl|\bra{\Phi^{\otimes n}}UAU^{\dagger}\otimes VBV^\dagger\ket{\Phi^{\otimes n}}\Bigr|\\=
\bigl|\widebar{\Tr}(\bar{U}A^TU^T VBV^\dagger)\bigr|
\le
\sqrt{\widebar{\Tr}(A^2)\,\widebar{\Tr}(B^2)}.
\end{multline*}

Applying this to \(a\) and \(b\), we obtain
\begin{align*}
a^2
&\le
\widebar{\Tr}\bigl((P_0')^2\bigr)\,
\widebar{\Tr}\bigl((Q_0+Q_1)^2\bigr)
\le
\rho^2\,\widebar{\Tr}\bigl((Q_0+Q_1)^2\bigr),\\
b^2
&\le
\widebar{\Tr}\bigl((P_1')^2\bigr)\,
\widebar{\Tr}\bigl((Q_0-Q_1)^2\bigr)
\le
\rho^2\,\widebar{\Tr}\bigl((Q_0-Q_1)^2\bigr).
\end{align*}
Therefore,
\begin{align*}
s^2+\tilde s^2
&=2(a^2+b^2)\\
&\le
2\rho^2\Bigl(
\widebar{\Tr}\bigl((Q_0+Q_1)^2\bigr)
+
\widebar{\Tr}\bigl((Q_0-Q_1)^2\bigr)
\Bigr)\\
&=
4\rho^2\,\widebar{\Tr}(Q_0^2+Q_1^2)\\
&\le
8\rho^2,
\end{align*}
where in the last step we used \(Q_0^2,Q_1^2\preceq I\).

Therefore, for every point \((s,\tilde s)\) in the projection of \(\mathcal{Q}_\rho\) onto the \((S,\tilde S)\)-plane, we have
\[
s^2+\tilde s^2\le 8\rho^2.
\]
Equivalently, this projection is contained in the disk of radius \(2\sqrt{2}\rho\).




We now show that the boundary circle is attainable when Alice and Bob share depolarized EPR pairs $\Delta_\rho(\Phi)^{\otimes n}$. Fix one register and let the observables
act nontrivially only on that register. Take
\begin{align*}
P_0 &= Z,
&
P_1 &= X,\\
Q_0(\theta) &= \cos\theta\, Z+\sin\theta\, X,
&
Q_1(\theta) &= \cos\theta\, Z-\sin\theta\, X,
\end{align*}
where $\theta\in[0,2\pi).$
These are traceless observables, and for the depolarized EPR state one obtains
\begin{align*}
s(\theta)
&:=
\Tr\bigl(S\,\Delta_\rho(\Phi)^{\otimes n}\bigr)
=
2\rho(\cos\theta+\sin\theta),\\
\tilde s(\theta)
&:=
\Tr\bigl(\tilde S\,\Delta_\rho(\Phi)^{\otimes n}\bigr)
=
2\rho(\cos\theta-\sin\theta).
\end{align*}
Consequently,
\[
s(\theta)^2+\tilde s(\theta)^2=8\rho^2
\qquad
\text{for all }\theta\in[0,2\pi).
\]
Hence every point on the boundary circle of radius \(2\sqrt{2}\rho\) is achieved by a
valid measurement. Since the projected correlation set is convex, it follows that the
projection of \(\mathcal{Q}_\rho\) onto the \((S,\tilde S)\)-plane is exactly the disk
of radius \(2\sqrt{2}\rho\), as illustrated in \cref{fig correlation set}.

\subsection{2-out-of-\texorpdfstring{$n$}{n} CHSH Game}
\label{Noisy 2-out-of-$n$ CHSH Game}

In this section we study the 2-out-of-$n$ CHSH game where Alice and Bob share $n'$ noisy states: $\Psi^{\otimes n'}$, where the maximal correlation of $\Psi$ is at most $\rho$ and the marginals of $\Psi$ are $I/2$, the same as \cref{self-test-depolarize} . Our goal is to prove a rigidity result on $n\leq n'$ pairs of anti-commuting observables on different registers. 
Our protocol resembles that of \cite{chao2018test}, except that we omit the consistency test. In \cite[Theorem 2.1]{chao2018test}, the consistency test is used to certify that the shared state is stabilized by pairs of anti-commuting observables, and therefore is isometric to EPR pairs. In our setting, however, the source is already assumed to prepare copies of shared states of a fixed form, so the consistency test is not needed. The referee repeats \cref{multiple-single-test} for enough times to guarantee that the players' observables have bounded trace.

\begin{framed}
\begin{protocol}[2-out-of-$n$ CHSH Game]\label{multiple-single-test}
With probability $\frac{1}{2}$ each, the verifier performs:
\begin{enumerate}       
\item The verifier randomly chooses indices $i, j \in [n]$, $i \neq j$ and questions of CHSH $x, y, z \in \set{0, 1}$. The verifier sends Alice $\set{\br{i, x}}$ and Bob $\set{\br{i, y},\br{j, z}}$. Alice returns $a \in \set{-1, 1}$. Bob returns $b,c\in\set{-1,1}$, ordered corresponding to $\br{i, y}$ , $\br{j, z}$, respectively. The verifier accepts if and only if $\frac{1+a}{2}\oplus \frac{1+b}{2} = xy$.
 
\item The same as the above sub-protocol, but with the roles of Alice and Bob exchanged.
\end{enumerate}
\end{protocol}
\end{framed}

We formalize the players' measurements:
\begin{enumerate}
    \item Alice uses observable $P_{i,0}$, when receiving question $\set{\br{i,0}}$ and observable $P_{i,1}$ when receiving question $\set{\br{i,1}}$. Similarly Bob uses observables $Q_{j,0}$ or $Q_{j,1}$ when receiving question $\set{\br{j,0}}$ or $\set{\br{j,1}}$.
    \item Upon receiving $\set{\br{i,y},\br{j,z}}$, Bob measures with POVM $\set{E_{a,b}^{\br{i,y},\br{j,z}}:a,b\in\set{-1, 1}}$ where $a,b$ are answers for $\br{i,y}$ and $\br{j,z}$, respectively. Define Bob's observable on index $i$ by marginalizing over index $j$:
    \begin{align*}
        &R^{iy|\br{iy,jz}}\\
        =~& \sum_{a\in \set{-1, 1}} a\br{E_{a,1}^{\br{i,y},\br{j,z}}+E_{a,-1}^{\br{i,y},\br{j,z}}}\\
        =~&E_{1,1}^{\br{i,y},\br{j,z}}+E_{1,-1}^{\br{i,y},\br{j,z}}-E_{-1,1}^{\br{i,y},\br{j,z}}-E_{-1,-1}^{\br{i,y},\br{j,z}}.
    \end{align*} 
    
    Similarly, define Bob's observable on index $j$ by marginalizing over index $i$:
    \begin{align*}
        &R^{jz|\br{iy,jz}}\\
        =~& \sum_{b\in \set{-1, 1}} b\br{E_{1,b}^{\br{i,y},\br{j,z}}+E_{-1,b}^{\br{i,y},\br{j,z}}}\\
        =~&E_{1,1}^{\br{i,y},\br{j,z}}+E_{-1,1}^{\br{i,y},\br{j,z}}-E_{1,-1}^{\br{i,y},\br{j,z}}-E_{-1,-1}^{\br{i,y},\br{j,z}}.
    \end{align*}    
    Alice's observables $T^{iy|\br{iy,jz}},T^{jz|\br{iy,jz}}$ are defined analogously.
\end{enumerate}

Define the trace bound

\begin{align}\label{2-out-of-n-def-max-trace}
\Trerr
&= \max_{\substack{i,j \in [n],\, i \neq j \\ x,y,z \in \{0,1\}}}
\Bigl\{
    \left|\overline{\Tr} P_{i,x}\right|,
    \left|\overline{\Tr} Q_{i,x}\right|, \notag\\
&\qquad
    \left|\overline{\Tr} R^{iy \mid (iy,jz)}\right|,
    \left|\overline{\Tr} T^{iy \mid (iy,jz)}\right|
\Bigr\}.
\end{align}

The same as in \cref{app:chsh-trace-test}, high probability passing the trace test implies bounded traces.

Although the following proof is similar to that in \cite{chao2018test}, we provide full details here for completeness.

\begin{theorem}\label{thm-multiple-self-test}
    Suppose the players share $n'\geq n$ copies of 
    $\Psi$, where $\rho\in\br{0,1}$ and that the players can pass \cref{multiple-single-test} with probability at least $\frac{1}{2}+\frac{\sqrt{2}}{4}\rho-\epsilon_v> 0.75$.

    Define parameters:
    \begin{align*}
        \gamma=\frac{1}{1-\rho}>1,\qquad \epsilon=\epsilon_v/\rho+\Trerr^2/\rho^2=o(1/n^2),
    \end{align*}
    then the following hold:
    \begin{enumerate}
        \item  For each $i\in[n]$, there exists $s_i\in[n']$ and a 2-dimensional traceless observable $\tilde{P}_{i,x}^{\br{s_i}}$ acting on the $s_i$-th register, such that
    \begin{align}\label{multiple-one-qubit}
        \forall x\in\set{0,1},P_{i,x}&\approx_{n\sqrt{\epsilon}\gamma} \tilde{P}_{i,x}^{\br{s_i}} \otimes I^{[n']\backslash\set{s_i}},
    \end{align} 
where $\gamma=\frac{1}{1-\rho}$. Moreover, $s_i\neq s_j$ if $i\neq j$.
        \item For each $i\in [n]$,
        there exists a two-dimensional unitary $\tilde{U}_i$, such that $\tilde{U}_i\tilde{P}_{i,0}\tilde{U}_i^*\approx_{n\sqrt{\epsilon}\gamma} Z,\tilde{U}_i\tilde{P}_{i,1}\tilde{U}_i^*\approx_{n\sqrt{\epsilon}\gamma} X$.
          
    \end{enumerate}

\end{theorem}
The second result directly implies that we can extract $n$ pairs of Pauli observables from $\br{P_{i,0},P_{i,1}}$ for $i\in [n]$, where for each extraction we only need a 2-dimensional unitary acting on the $s_i$-th qubit.
For constant noise levels, the robustness parameter is $O(n\sqrt{\epsilon})$, scaling linear on $n$. And for noiseless cases $\rho \to 1$ and $\gamma\to \infty$, the conclusions in \cref{thm-multiple-self-test} do not necessarily hold. 
\begin{proof}[Proof of \cref{thm-multiple-self-test}]

    By \cref{traceless-chsh-vio-upperbound}, the winning rate for a single CHSH game is at most $\frac{1}{2}+\frac{\sqrt{2}}{4}\rho+O\br{\Trerr^2}$. Therefore, conditioning on a specific index pair $\br{i,j}$, the winning rate is at least $\frac{1}{2}+\frac{\sqrt{2}}{4}\rho- O\br{n^2\br{\Trerr^2+\epsilon_v}}$, which is larger than $0.75$ for $\epsilon=o(1/n^2)$ (or equivalently attaining violation $2\sqrt{2}\rho-O\br{n^2\epsilon}>2$).
    
    Suppose that Alice receives question $\set{\br{i,x}}$ and Bob receives questions $\set{\br{i,y},\br{j,z}}$, then the measurement $\br{P_{i,x}, R^{iy|\br{iy,jz}}}$ attains violation $2\sqrt{2}\rho-O\br{n^2\epsilon}$ in the single CHSH game under noise.
    By \cref{cor-almost-anti-commute}, $\set{ P_{i,0},P_{i,1}}\approx_{n\sqrt{\epsilon}}0.$ Similarly, $\set{ Q_{i,0},Q_{i,1}}\approx_{n\sqrt{\epsilon}}0.$

    \cref{thm-close-to-deg1} implies that
    \begin{align*}
         R^{iy|\br{iy,jz}}\approx_{n\sqrt{\epsilon}} \frac{1}{\sqrt{2}}\br{P_{i,0}+\br{-1}^yP_{i,1}}^T.
     \end{align*} 
     Moreover, if Alice receives question $\set{\br{j,x}}$, Bob receives questions $\set{\br{i,y},\br{j,z}}$, we have
     \begin{align*}
          R^{jz|\br{iy,jz}}\approx_{n\sqrt{\epsilon}} \frac{1}{\sqrt{2}}\br{P_{j,0}+\br{-1}^zP_{j,1}}^T.
    \end{align*}
    Since $\set{ P_{i,0},P_{i,1}},\set{ P_{j,0},P_{j,1}}\approx_{n\sqrt{\epsilon}}0$, and by \cref{thm-close-to-deg1} and \cref{closetobservable}, for $x\in\set{0,1}$, $\br{P_{i,x}}^2,\br{P_{j,x}}^2\approx_{n\sqrt{\epsilon}} I$, we have $\br{ R^{iy|\br{iy,jz}}}^2\approx_{n\sqrt{\epsilon}} I,\br{ R^{jz|\br{iy,jz}}}^2\approx_{n\sqrt{\epsilon}} I.$ By \cref{sqrt-approx-project}, $R^{iy|\br{iy,jz}}$ and $R^{jz|\br{iy,jz}}$ are $n\sqrt{\epsilon}$-close to binary observables. Following the similar proof in \cref{mag-commute-same-row} (where the approximation parameter changes from $\sqrt{\epsilon}$ to $n\sqrt{\epsilon}$) we have for different POVMs $E_{a,b}^{\br{i,y},\br{j,z}}$ and $E_{a',b'}^{\br{i,y},\br{j,z}}$,
    \begin{align*}
        E_{a,b}^{\br{i,y},\br{j,z}}E_{a',b'}^{\br{i,y},\br{j,z}}\approx_{n\sqrt{\epsilon}} 0.        
    \end{align*} Therefore 
    \begin{align*}
    \left[R^{iy|\br{iy,jz}}, R^{jz|\br{iy,jz}}\right]\approx_{n\sqrt{\epsilon}}0,
    \end{align*} and consequently, $\forall y,z\in\set{0,1} $,
    \begin{align*}
        \left[P_{i,0}+\br{-1}^yP_{i,1},P_{j,0}+\br{-1}^zP_{j,1}\right]\approx_{n\sqrt{\epsilon}}0.
    \end{align*}
    This implies $\forall u,v\in\set{0,1}, \left[P_{iu},P_{jv}\right]\approx_{n\sqrt{\epsilon}}0.$
    
    Following the proof in \cref{thm-close-to-deg1}, we obtain for $x\in\set{0,1},P_{i,x}'\approx_{n\sqrt{\epsilon}}\rho P_{i,x}$. Therefore \cref{multiple-one-qubit} can be similarly proved as in \cref{chsh-selftest-main-thm} by truncating Pauli terms with degree $\neq 1$. 
    Suppose that there exists $i\neq j,s_i=s_j,$ we then have
    \begin{align*}
        \left[\tilde{P}_{i,0}^{\br{s_i}},\tilde{P}_{j,0}^{\br{s_j}}\right]&\approx_{n\sqrt{\epsilon}\gamma} 0, &\left[\tilde{P}_{i,1}^{\br{s_i}},\tilde{P}_{j,0}^{\br{s_j}}\right]\approx_{n\sqrt{\epsilon}\gamma} 0,\\
        \set{\tilde{P}_{i,0}^{\br{s_i}},\tilde{P}_{i,1}^{\br{s_i}}}&\approx_{n\sqrt{\epsilon}\gamma} 0,
    \end{align*}
    which is a contradiction by \cref{commute-anticommute}. That proves the first item. By \cref{anti-commute-to-ZX} we can prove the second item.
\end{proof}

\section{Magic Square Games Under Noise}
\subsection{Magic Square Game}
\label{sec:NMS}
In this section we study the effect of noise on the Magic Square game\cref{def-mgs-game}, where the players share $n$ copies of the state 
\[\Delta_\rho^{(16)}\br{\Phi^{\otimes 2}}=\rho \Phi^{\otimes 2}+(1-\rho)\frac{I}{16},\]
for a constant correlation parameter $\rho\in\br{0,1}$. Then we prove the rigidity result of this game.

We now formalize the players' measurements. Suppose the players share $n$ copies of $\Delta_\rho\br{\Phi^{\otimes 2}}$. For $i,j\in[3]$, let $Q_{i,j}$ denote Bob's observable corresponding to the variable $s_{i,j}$. When Alice receives a question $x\in\{r_1,r_2,r_3,c_1,c_2,c_3\}$, she performs a POVM $\set{E_{a_1,a_2,a_3}^x:a_1,a_2,a_3\in\{-1,1\}}$, where $E_{a_1,a_2,a_3}^x$ corresponds to the outcome $a=(a_1,a_2,a_3).$ 

Fix a variable $s_{i,j}$. Define Alice's observable for $s_{i,j}$ in the row and column settings as
\[
    P_{i,j}^\mathrm{row}=\sum_{a\in\{-1,1\}^3}a_jE_{a_1,a_2,a_3}^{r_i},\quad
    P_{i,j}^\mathrm{column}=\sum_{a\in\{-1,1\}^3}a_iE_{a_1,a_2,a_3}^{c_j}.
\]


The trace error $\Trerr$ is defined as:
\begin{equation}\label{eqn:traceerror}
    \Trerr=\max_{i,j\in [3]}\set{\widebar{\Tr}\br{P_{i,j}^{\mathrm{row}}},\widebar{\Tr}\br{P_{i,j}^{\mathrm{column}}},\widebar{\Tr}\br{Q_{i,j}}}.
\end{equation}
The methoc to estimate the trace error is the same as in the previous section: The referee repeats independent instances of the game sufficiently many times to ensure that all observables have bounded trace as in \cref{app:chsh-trace-test}, where we make the same i.i.d. repetition assumption as in \cref{repetition-iid}, that is, the players follow fixed measurements and share $n$ new copies of $\Delta_{\rho}^{\br{16}}\br{\Phi^{\x 2}}$ in each repetition.

       


We then give a tight upper bound of the winning probability:
\begin{theorem}\label{mgs-upper-bound-winning-rate}
    Suppose the players share $n$ copies of $\Delta_\rho\br{\Phi^{\otimes 2}}$ in the Magic Square game under noise. If the players use observables with trace error at most $\Trerr$, then the probability of passing the consistency test, and therefore also the overall game, is at most
\[
\frac{1 + \rho}{2} + \frac{\Trerr^2}{4\rho}.
\]
\end{theorem}
\begin{proof}
     Fix a variable $s_{i,j}$. The probability that Alice and Bob's assignments agree on $s_{i,j}$ is 
    \begin{align*}
        &\frac{1}{2}+\frac{1}{2}\Tr\br{P_{i,j}^x\x Q_{i,j} \br{\Delta_{\rho}\br{\Phi^{\x 2}} }^{\x n}}\\
        =~&\frac{1}{2}+\frac{1}{2}\Tr\br{P_{i,j}^{\mathrm{row}}\x Q_{i,j}' \br{\Phi^{\x 2n}}},
    \end{align*} where the equality follows from \cref{pre-noise-to-obs-mgs}.
    Still using  SoS-like decomposition, we define
    \begin{align*}
         C=~&P_{i,j}^{\mathrm{row}}\otimes Q_{i,j}'=-\frac{\rho}{2}\br{ P_{i,j}^{\mathrm{row}}\otimes I-I\otimes \frac{Q_{i,j}'}{\rho}}^2\\
         &+\frac{\rho}{2}\br{P_{i,j}^{\mathrm{row}}}^2\otimes I+\frac{1}{2\rho} I\otimes \br{Q_{i,j}'}^2.
    \end{align*}
    By a similar argument to that for \cref{almost-deg-one-1}, we have
\begin{align}
    \bra{\Phi^{\x 2n}}I\otimes \br{Q_{i,j}'}^2\ket{\Phi^{\x 2n}}\leq  \rho^2+\Trerr^2.
\end{align}
Combined with $\br{P_{i,j}^{\mathrm{row}}}^2\preceq I$,
we have
\begin{align*}
    \Tr\br{C\Phi^{\x 2n}}\leq \frac{\rho}{2}+\frac{\rho^2+\Trerr^2}{2\rho}=\rho+\frac{\Trerr^2}{2\rho}.
\end{align*}
 Thus the probability of passing the consistency test is at most $\frac{1+\rho}{2}+\frac{\Trerr^2}{4\rho}$. This is also an upper bound of the overall winning probability of the game.
 Moreover, this upper bound is obtainable by applying the measurements in \cref{tab:mgs-OptimalMeasurementTable} to
 $\Delta_{\rho}^{\br{16}}\br{\Phi^{\x 2}}$.
\end{proof}
By \cref{mgs-upper-bound-winning-rate}, if the players restrict themselves to traceless observables, the maximum winning probability is $\frac{1+\rho}{2}$. 
We show that if the players win the Magic Square game under noise with trace-bounded observables and near-optimal success probability, then their observables, up to local unitaries, behave like those in the optimal strategy that uses one copy of $\Phi^{\otimes 2}$ in the noiseless setting.

\begin{theorem}\label{mgs-self-testing-main-thm}
Suppose the players share $n$ copies of the state $\Delta_\rho(\Phi^{\otimes 2})$, where $\rho \in (0,1)$ and that the players win with probability at least $\frac{1+\rho}{2} - \winerr$.
Define the parameters:
\[
\gamma := \frac{1}{1 - \rho}, \qquad
\epsilon := \winerr/\rho + \Trerr^2/\rho^2.
\]

Then there exists $\ell \in [n]$ such that the following statements hold:
\begin{enumerate}
    \item For all $i,j \in [3]$, 
    \[
    P_{i,j}^{\mathrm{row}} \approx_{\sqrt{\epsilon}} P_{i,j}^{\mathrm{column}}\approx_{\sqrt{\epsilon}}Q_{i,j}^T.
    \]
    We then omit superscripts of $P_{i,j}$ in the following items.

    \item For $D \in \{P, Q\}$ and $i,j \in [3]$, there exist traceless binary observables $\tilde{D}_{i,j}^{(\ell)} \in \H_4$ acting on the $\ell$-th register such that
    \[
    D_{i,j} \approx_{\gamma \sqrt{\epsilon}} \tilde{D}_{i,j}^{(\ell)} \otimes I_4^{[n] \setminus \{\ell\}}.
    \]

    \item There exists a 4-dimensional unitary $U$ acting on the $\ell$-th register, such that for $i,j\in[3]$,
    \begin{align*}
 UP_{i,j}U^* &\approx_{\gamma \sqrt{\epsilon}}   \br{P_{i,j}^{\star}}^{\br{\ell}} \otimes I_{4}^{[n]\backslash\set{\ell}},\\
  \bar{U}Q_{i,j}U^T &\approx_{\gamma \sqrt{\epsilon}}   \br{Q_{i,j}^{\star}}^{\br{\ell}} \otimes I_{4}^{[n]\backslash\set{\ell}},
    \end{align*}
    where $P_{i,j}^{\star}$ is the $\br{i,j}$-th entry of \cref{tab:mgs-OptimalMeasurementTable} and $Q_{i,j}^{\star}=\br{P_{i,j}^{\star}}^T$.
\end{enumerate}
\end{theorem}


In both the noiseless setting~\cite[Theorem 7.10]{JNVWY'20} and the noisy case with a constant correlation parameter, the robustness parameter scales as $O(\sqrt{\epsilon})$ when the players' winning probability is within $\epsilon$ of optimal.

In the noiseless case, however, the structural conclusions stated in item 2 and item 3 of \cref{mgs-self-testing-main-thm} may fail to hold. For example,  one can verify that the measurement using six EPR pairs and observables $\{(P_{i,j}^\star)^{\otimes 3}\}$ and $\{(Q_{i,j}^\star)^{\otimes 3}\}$ also achieves the optimal winning probability, but does not satisfy item 2 or item 3 of \cref{mgs-self-testing-main-thm}.

Prior to proving \cref{mgs-self-testing-main-thm}, we establish several preparatory lemmas. These results not only facilitate the proof of \cref{mgs-self-testing-main-thm} but also reveal key properties of Alice's and Bob's observables. Proofs are deferred to \cref{app:mgs}.



We first show that success of the consistency test implies alignment between Alice’s and Bob’s observables:
\begin{lemma}\label{mgs-consist-row-col}
Suppose the players pass the consistency test in the Magic Square game under noise (\cref{def-mgs-game}) with probability at least $\frac{1+\rho}{2}-\winerr$. Let $\epsilon=\winerr/\rho+\Trerr^2/\rho^2$, where $\Trerr$ is defined in \cref{eqn:traceerror}. Then for all $i,j\in[3]$, the following statements hold:
\begin{enumerate}
    \item $P_{i,j}^{\mathrm{row}}\approx_{\sqrt{\epsilon}} P_{i,j}^{\mathrm{column}}\approx_{\sqrt{\epsilon}}Q_{i,j}^T.$ We then omit superscripts of $P_{i,j}$.
    \item $P_{i,j}^2\approx_{\sqrt{\epsilon}}I$.
\end{enumerate}
\end{lemma}
The next lemma states that observables in the same row or column approximately commute.
\begin{lemma}\label{mag-commute-same-row}
    Under the same assumptions as in \cref{mgs-consist-row-col}, 
    for any $i,i',j,j'\in[3]$ such that $i\ne i'$ and $j\ne j'$, we have 
    \[\left[P_{i,j},P_{i,j'}\right]\approx_{\sqrt{\epsilon}}0\quad\text{and}\quad \left[P_{i,j},P_{i',j}\right]\approx_{\sqrt{\epsilon}}0.\]
\end{lemma}
We then recover the parity constraints of the Magic Square game in terms of the observables:
\begin{lemma}\label{mgs-obs-constraint}
    Under the same assumptions as in \cref{mgs-consist-row-col}, the following observable constraints hold:
    \begin{align*}
        &P_{i,1}P_{i,2}P_{i,3} \approx_{\sqrt{\epsilon}} I \quad \text{for all } i\in[3], \\
        &P_{1,j}P_{2,j}P_{3,j} \approx_{\sqrt{\epsilon}} I \quad \text{for } j=1,2, \\
        &P_{1,3}P_{2,3}P_{3,3} \approx_{\sqrt{\epsilon}} -I.
    \end{align*}
\end{lemma}
Then, we show that observables located in different rows and different columns are approximately anti-commute. 
\begin{lemma}\label{claim-mgs-anti-commute}
Under the same assumptions as in \cref{mgs-consist-row-col}, for all $i \neq i'$ and $j \neq j'$, we have
\[
\set{P_{i,j}, P_{i',j'}} \approx_{\sqrt{\epsilon}} 0, \qquad \set{Q_{i,j}, Q_{i',j'}} \approx_{\sqrt{\epsilon}} 0.
\]
\end{lemma}
We also need the following lemmas to prove that all players measurements act on the same single register.
\begin{lemma}\label{lem:ms-single-observable-localization}
Under the assumptions of \cref{mgs-consist-row-col}, for every \(i,j\in[3]\), there exist indices
\[
\ell_{i,j},m_{i,j}\in[n]
\]
and traceless binary observables
\[
\widetilde{P}_{i,j}^{(\ell_{i,j})},\ \widetilde{Q}_{i,j}^{(m_{i,j})}\in \H_4
\]
such that
\begin{align*}
P_{i,j}
&\approx_{\gamma\sqrt{\epsilon}}
\widetilde{P}_{i,j}^{(\ell_{i,j})}\otimes I^{[n]\setminus\{\ell_{i,j}\}}_{4},\\
Q_{i,j}
&\approx_{\gamma\sqrt{\epsilon}}
\widetilde{Q}_{i,j}^{(m_{i,j})}\otimes I^{[n]\setminus\{m_{i,j}\}}_{4}.
\end{align*}
\end{lemma}

\begin{lemma}\label{lem:ms-common-register}
Under the assumptions of \cref{mgs-consist-row-col}, there exists a single register index
\[
\ell\in[n]
\]
such that for every \(i,j\in[3]\), there exist traceless binary observables
\[
\widetilde{P}_{i,j}^{(\ell)},\ \widetilde{Q}_{i,j}^{(\ell)}\in\H_4
\]
satisfying
\begin{align*}
P_{i,j}
&\approx_{\gamma\sqrt{\epsilon}}
\widetilde{P}_{i,j}^{(\ell)}\otimes I^{[n]\setminus\{\ell\}}_{4},\\
Q_{i,j}
&\approx_{\gamma\sqrt{\epsilon}}
\widetilde{Q}_{i,j}^{(\ell)}\otimes I^{[n]\setminus\{\ell\}}_{4}.
\end{align*}
\end{lemma}

\begin{lemma}\label{lem:ms-local-canonical-form}
Assume the hypotheses of \cref{mgs-consist-row-col} and let \(\ell\in[n]\) be the common register given by
\cref{lem:ms-common-register}. Then there exists a \(4\)-dimensional unitary
\[
\widetilde{U}:\H_4\to\H_4
\]
such that for all \(i,j\in[3]\),
\begin{align*}
\widetilde{U}\widetilde{P}_{i,j}^{(\ell)}\widetilde{U}^*
&\approx_{\gamma\sqrt{\epsilon}} P_{i,j}^{\star},\\
\widetilde{U}\,\overline{\widetilde{Q}_{i,j}^{(\ell)}}\,\widetilde{U}^{T}
&\approx_{\gamma\sqrt{\epsilon}} Q_{i,j}^{\star},
\end{align*}
where \(P_{i,j}^{\star}\) is the \((i,j)\)-th entry of Table~I and
\[
Q_{i,j}^{\star}:=\left(P_{i,j}^{\star}\right)^T.
\]
\end{lemma}

We are now ready to prove \cref{mgs-self-testing-main-thm}.

\begin{proof}[Proof of \cref{mgs-self-testing-main-thm}]
Item \(1\) is exactly \cref{mgs-consist-row-col}. Hence we may omit the superscripts
\(\mathrm{row}\) and \(\mathrm{column}\) and simply write \(P_{i,j}\).

By \cref{lem:ms-single-observable-localization,lem:ms-common-register}, there exists a register
\[
\ell\in[n]
\]
such that for every \(i,j\in[3]\), there exist traceless binary observables
\[
\widetilde{P}_{i,j}^{(\ell)},\ \widetilde{Q}_{i,j}^{(\ell)}\in\H_4
\]
satisfying
\begin{align*}
P_{i,j}
&\approx_{\gamma\sqrt{\epsilon}}
\widetilde{P}_{i,j}^{(\ell)}\otimes I^{[n]\setminus\{\ell\}}_{4},\\
Q_{i,j}
&\approx_{\gamma\sqrt{\epsilon}}
\widetilde{Q}_{i,j}^{(\ell)}\otimes I^{[n]\setminus\{\ell\}}_{4}.
\end{align*}
This proves item \(2\).

By \cref{lem:ms-local-canonical-form}, there exists a \(4\)-dimensional unitary
\[
\widetilde{U}:\H_4\to\H_4
\]
such that for all \(i,j\in[3]\),
\begin{align*}
\widetilde{U}\widetilde{P}_{i,j}^{(\ell)}\widetilde{U}^*
&\approx_{\gamma\sqrt{\epsilon}} P_{i,j}^{\star},\\
\widetilde{U}\,\overline{\widetilde{Q}_{i,j}^{(\ell)}}\,\widetilde{U}^{T}
&\approx_{\gamma\sqrt{\epsilon}} Q_{i,j}^{\star}.
\end{align*}
Define the global unitary
\[
U:=\widetilde{U}^{(\ell)}\otimes I^{[n]\setminus\{\ell\}}_{4}.
\]
Combining this with item \(2\), we obtain
\begin{align*}
UP_{i,j}U^*
&\approx_{\gamma\sqrt{\epsilon}}
\left(P_{i,j}^{\star}\right)^{(\ell)}\otimes I^{[n]\setminus\{\ell\}}_{4},\\
U\overline{Q_{i,j}}\,U^{T}
&\approx_{\gamma\sqrt{\epsilon}}
\left(Q_{i,j}^{\star}\right)^{(\ell)}\otimes I^{[n]\setminus\{\ell\}}_{4},
\qquad \forall i,j\in[3].
\end{align*}
This proves item \(3\), and hence completes the proof.
\end{proof}

\subsection{2-out-of-\texorpdfstring{$n$}{n} Magic Square Game}

In this section, we study the 2-out-of-$n$  Magics Square game, where Alice and Bob share
\begin{align*}
    \bigl(\Delta_{\rho}^{(16)}(\Phi^{\otimes 2})\bigr)^{\otimes n'}
\end{align*}
for some $n' \ge n$ and a fixed constant $\rho \in (0,1)$. Our goal is to extract $n$ copies of the canonical
Magic Square observables, each acting on a distinct $4$-dimensional register.

\begin{framed}
    \begin{protocol}[2-out-of-$n$ Magic Square Game]\label{2-of-n-mgs-single}
With probability $1/2$, the verifier performs the following.
\begin{enumerate}
    \item Randomly choose two distinct indices $i,j \in [n]$, and two Magic Square questions
    $x,y \in \{r_1,r_2,r_3,c_1,c_2,c_3\}$.
    The verifier sends Alice the pair $\{(i,x),(j,y)\}$.
    Then the verifier samples a cell $z \in x$ uniformly at random and sends Bob the question $\{(i,z)\}$.
    Alice returns two assignments in $\{-1,1\}^3$, one for $(i,x)$ and one for $(j,y)$.
    Bob returns one bit in $\{-1,1\}$ for the cell $(i,z)$.
    The verifier accepts if and only if:
    \begin{enumerate}
        \item Alice's two assignments satisfy the Magic Square parity constraints for $x$ and $y$; and
        \item Alice's assignment on the cell $z$ in question $(i,x)$ is consistent with Bob's answer.
    \end{enumerate}

    \item This test is identical to item 1, except that the roles of Alice and Bob are exchanged.
\end{enumerate}
\end{protocol}
\end{framed}

We now formalize the players' strategies. For each $i\in [n]$ and each cell $z\in [3]\times [3]$, let
$Q_{i,z}$ denote the single-cell observable used by the player who receives the question $\{(i,z)\}$.

Next, consider a player who receives the pair of indexed Magic Square questions $\{(i,x),(j,y)\}$ with
$i\neq j$. For each cell $z_1 \in x$, define the marginal observable on the $i$-th index by
\begin{align*}
    P_{i,z_1 \mid (ix,jy)}
    :=
    \sum_{a,b}
    a(z_1)\, E^{(ix,jy)}_{a,b},
\end{align*}
where $a \in \{-1,1\}^3$ is the assignment returned on question $(i,x)$, $b \in \{-1,1\}^3$ is the assignment
returned on question $(j,y)$, and $E^{(ix,jy)}_{a,b}$ is the POVM element corresponding to the pair of
assignments $(a,b)$. Similarly, for each cell $z_2 \in y$, define
\begin{align*}
    P_{j,z_2 \mid (ix,jy)}
    :=
    \sum_{a,b}
    b(z_2)\, E^{(ix,jy)}_{a,b}.
\end{align*}
When the roles of Alice and Bob are exchanged, we define the analogous marginal observables in the same
way.

Define the trace error parameter

\begin{align}
\epsilon_{\mathrm{tr}}
&:=
\max_{\substack{
i,j \in [n],\, i \neq j \\
x,y \in \{r_1,r_2,r_3,c_1,c_2,c_3\}
}}
\;\;
\max_{\substack{
z_1 \in x \\
z_2 \in y
}}
\Bigl\{
    \bigl|\overline{\Tr}(Q_{i,z_1})\bigr|,
    \bigl|\overline{\Tr}(P_{i,z_1 \mid (ix,jy)})\bigr|, \notag\\
&\qquad\qquad
    \bigl|\overline{\Tr}(P_{j,z_2 \mid (ix,jy)})\bigr|
\Bigr\}.
\label{eq:2n-ms-trace-bound}
\end{align}
Similarly as in \cref{app:chsh-trace-test}, passing the trace test with high probability implies bounded trace.


We now state the rigidity theorem for the $2$-out-of-$n$ Magic Square game under noise.

\begin{theorem}\label{2-of-n-thm-rigidity}
Suppose the players share $n' \ge n$ copies of $\Delta_{\rho}^{(16)}(\Phi^{\otimes 2})$, where $\rho \in (0,1)$ and that in \cref{2-of-n-mgs-single}, the players win with probability at least
    \begin{align*}
        \frac{1+\rho}{2} - \epsilon_{\mathrm{win}} .
    \end{align*}
Define
\begin{align*}
    \gamma := \frac{1}{1-\rho},
    \qquad
    \epsilon := \epsilon_{\mathrm{win}}/\rho + \Trerr^2/\rho^2 =o(1/n^2).
\end{align*}
Then, for every $i\in [n]$, there exists a register index $s_i \in [n']$ such that the following statements hold.
\begin{enumerate}
    \item The extracted registers are pairwise distinct: if $i\neq j$, then
    \begin{align}
        s_i \neq s_j .
        \label{eq:2n-ms-distinct}
    \end{align}
    \item For every cell $z\in [3]\times [3]$, there exist traceless binary observables
    $\widetilde{Q}_{i,z}\in \H_4$ acting on the $s_i$-th register such that
    \begin{align}
        Q_{i,z}
        \approx_{\,n\gamma\sqrt{\epsilon}}
        \widetilde{Q}_{i,z}^{(s_i)} \otimes I_4^{[n']\setminus\{s_i\}} .
        \label{eq:2n-ms-single-reg}
    \end{align}

    \item There exists a $4$-dimensional unitary $U_i$ such that for all $z\in [3]\times [3]$,
    \begin{align}
        \overline{U_i}\, Q_{i,z}\, U_i^T
        \approx_{\,n\gamma\sqrt{\epsilon}}
        \bigl(Q_z^\star\bigr)^{(s_i)} \otimes I_4^{[n']\setminus\{s_i\}} .
        \label{eq:2n-ms-canonical}
    \end{align}
\end{enumerate}
\end{theorem}

The proof is parallel in spirit to that of \cref{thm-multiple-self-test}. For each fixed index, conditioning on a particular
pair of indices yields an induced noisy Magic Square measurement, to which \cref{mgs-self-testing-main-thm} applies. To separate
different registers, we prove that observables associated with different indices approximately commute.

\begin{lemma}
For all distinct indices $i,j \in [n]$ and all cells $z_1,z_2 \in [3]\times [3]$,
\begin{align}
    [Q_{i,z_1},Q_{j,z_2}]
    \approx_{\,n\sqrt{\epsilon}}
    0 .
    \label{eq:2n-ms-cross-comm}
\end{align}
\end{lemma}

\begin{proof}
Fix distinct indices $i,j\in [n]$, questions $x,y\in \{r_1,r_2,r_3,c_1,c_2,c_3\}$, and cells
$z_1\in x$, $z_2\in y$. Consider the measurement performed by a player receiving the pair
$\{(i,x),(j,y)\}$. This measurement is a POVM
\begin{align*}
    \bigl\{ E^{(ix,jy)}_{a,b} : a,b \in \{-1,1\}^3 \bigr\}.
\end{align*}
Define the two marginal observables
\begin{align*}
    P_{i,z_1\mid (ix,jy)} &= \sum_{a,b} a(z_1)\, E^{(ix,jy)}_{a,b}, \\
    P_{j,z_2\mid (ix,jy)} &= \sum_{a,b} b(z_2)\, E^{(ix,jy)}_{a,b}.
\end{align*}
Conditioning on the ordered pair $(i,j)$ and the fact that the tested index is $i$, the induced measurement on
index $i$ is a noisy Magic Square measurement whose winning probability is still
\begin{align*}
    \frac{1+\rho}{2} - O(n^2\epsilon) .
\end{align*}
Therefore, by \cref{mgs-self-testing-main-thm} applied to this induced measurement, we have
\begin{align}
    P_{i,z_1\mid (ix,jy)}
    \approx_{\,n\sqrt{\epsilon}}
    Q_{i,z_1}^T ,
    \label{eq:2n-ms-align-i}
\end{align}
and similarly,
\begin{align}
    P_{j,z_2\mid (ix,jy)}
    \approx_{\,n\sqrt{\epsilon}}
    Q_{j,z_2}^T .
    \label{eq:2n-ms-align-j}
\end{align}
It therefore suffices to prove that
\begin{align}
    [P_{i,z_1\mid (ix,jy)}, P_{j,z_2\mid (ix,jy)}]
    \approx_{\,n\sqrt{\epsilon}}
    0 .
    \label{eq:2n-ms-marg-comm}
\end{align}

We prove \eqref{eq:2n-ms-marg-comm} exactly as in the proof of \cref{mag-commute-same-row}. Since the induced single-index
measurement is $n\sqrt{\epsilon}$-close to an ideal noisy Magic Square measurement, \cref{mgs-self-testing-main-thm} implies that both
$P_{i,z_1\mid (ix,jy)}$ and $P_{j,z_2\mid (ix,jy)}$ are $n\sqrt{\epsilon}$-close to binary observables. Hence, by
\cref{sqrt-approx-project}, there exist projectors $\Pi_0,\Pi_1,\Pi_2,\Pi_3$ such that
\begin{align*}
    P_{i,z_1\mid (ix,jy)} &\approx_{\,n\sqrt{\epsilon}} \Pi_0-\Pi_1, \\
    P_{j,z_2\mid (ix,jy)} &\approx_{\,n\sqrt{\epsilon}} \Pi_2-\Pi_3,
\end{align*}
with
\begin{align*}
    \Pi_0+\Pi_1 = \Pi_2+\Pi_3 = I .
\end{align*}

Now partition the $36$ POVM outcomes according to the signs assigned to the two cells $z_1$ and $z_2$:
\begin{align*}
    S_{++} &:= \{(a,b): a(z_1)=1,\ b(z_2)=1\}, \\
    S_{+-} &:= \{(a,b): a(z_1)=1,\ b(z_2)=-1\}, \\
    S_{-+} &:= \{(a,b): a(z_1)=-1,\ b(z_2)=1\}, \\
    S_{--} &:= \{(a,b): a(z_1)=-1,\ b(z_2)=-1\}.
\end{align*}
Define
\begin{align*}
    A &:= \sum_{(a,b)\in S_{++}} E^{(ix,jy)}_{a,b}, &
    B &:= \sum_{(a,b)\in S_{+-}} E^{(ix,jy)}_{a,b}, \\
    C &:= \sum_{(a,b)\in S_{-+}} E^{(ix,jy)}_{a,b}, &
    D &:= \sum_{(a,b)\in S_{--}} E^{(ix,jy)}_{a,b}.
\end{align*}
Then
\begin{align*}
    A+B &\approx_{\,n\sqrt{\epsilon}} \Pi_0, &
    C+D &\approx_{\,n\sqrt{\epsilon}} \Pi_1, \\
    A+C &\approx_{\,n\sqrt{\epsilon}} \Pi_2, &
    B+D &\approx_{\,n\sqrt{\epsilon}} \Pi_3, \\
    A+B+C+D &= I .
\end{align*}
Applying \cref{lem:approx0}, we obtain
\begin{align*}
    AB \approx_{\,n\sqrt{\epsilon}} 0,
    \qquad
    AC \approx_{\,n\sqrt{\epsilon}} 0,
    \qquad
    AD \approx_{\,n\sqrt{\epsilon}} 0,
\end{align*}
and similarly for all other pairwise products among $A,B,C,D$. Consequently,
\begin{align*}
    [P_{i,z_1\mid (ix,jy)}, P_{j,z_2\mid (ix,jy)}]
    \approx_{\,n\sqrt{\epsilon}}
    0.
\end{align*}
Combining this with \eqref{eq:2n-ms-align-i} and \eqref{eq:2n-ms-align-j}, and using \cref{approx_multiply}, gives
\eqref{eq:2n-ms-cross-comm}.
\end{proof}

\begin{proof}[Proof of \cref{2-of-n-thm-rigidity}]
Fix an index $i\in [n]$. Conditioning on any ordered pair $(i,j)$ with $j\neq i$, and on the event that the
tested index is $i$, we obtain an induced noisy Magic Square measurement on index $i$. Since the overall
winning probability of the $2$-out-of-$n$ game is at least
\begin{align*}
    \frac{1+\rho}{2} - \epsilon_{\mathrm{win}},
\end{align*}
the induced measurement wins the Magic Square game under noise with probability at least
\begin{align*}
    \frac{1+\rho}{2} - O(n^2\epsilon) .
\end{align*}
Applying \cref{mgs-self-testing-main-thm} to this induced measurement, we obtain a register index $s_i \in [n']$, traceless binary
observables $\widetilde{Q}_{i,z}\in \H_4$, and a $4$-dimensional unitary $U_i$ such that
\begin{align*}
    Q_{i,z}
    &\approx_{\,n\gamma\sqrt{\epsilon}}
    \widetilde{Q}_{i,z}^{(s_i)} \otimes I_4^{[n']\setminus\{s_i\}},
    \\
    \overline{U_i}\, Q_{i,z}\, U_i^T
    &\approx_{\,n\gamma\sqrt{\epsilon}}
    \bigl(Q_z^\star\bigr)^{(s_i)} \otimes I_4^{[n']\setminus\{s_i\}}
\end{align*}
for every cell $z\in [3]\times [3]$. This proves \eqref{eq:2n-ms-single-reg} and
\eqref{eq:2n-ms-canonical}.

It remains to prove that the registers \(s_i\) are distinct. Suppose, toward a contradiction, that for some \(i\neq j\),
\[
s_i=s_j=:s.
\]
By \eqref{eq:2n-ms-cross-comm}, for all cells \(z_1,z_2\),
\begin{align*}
    [Q_{i,z_1},Q_{j,z_2}]
    \approx_{\,n\sqrt{\epsilon}}
    0 .
\end{align*}
Since \(s_i=s_j=s\), for every \(z_1,z_2\),
\begin{align*}
    Q_{i,z_1}
    &\approx_{\,n\gamma\sqrt{\epsilon}}
    \widetilde Q_{i,z_1}^{(s)}\otimes I,\\
    Q_{j,z_2}
    &\approx_{\,n\gamma\sqrt{\epsilon}}
    \widetilde Q_{j,z_2}^{(s)}\otimes I.
\end{align*}
Hence
\begin{align*}
    [\widetilde Q_{i,z_1}^{(s)},\widetilde Q_{j,z_2}^{(s)}]
    \approx_{\,n\gamma\sqrt{\epsilon}}
    0
    \qquad
    \forall\, z_1,z_2\in [3]\times [3].
\end{align*}

Fix \(z_1\in [3]\times [3]\), and define
\[
A:=U_j\,\widetilde Q_{i,z_1}^{(s)}\,U_j^* .
\]
Then by \eqref{eq:2n-ms-canonical}, for every \(z_2\in [3]\times [3]\),
\[
U_j\,\widetilde Q_{j,z_2}^{(s)}\,U_j^*
\approx_{\,n\gamma\sqrt{\epsilon}}
Q_{z_2}^\star .
\]
Therefore,
\begin{align*}
    [A,Q_{z_2}^\star]
    &\approx_{\,n\gamma\sqrt{\epsilon}}
    \Bigl[
        U_j\,\widetilde Q_{i,z_1}^{(s)}\,U_j^*,
        U_j\,\widetilde Q_{j,z_2}^{(s)}\,U_j^*
    \Bigr] \\
    &=
    U_j
    [\widetilde Q_{i,z_1}^{(s)},\widetilde Q_{j,z_2}^{(s)}]
    U_j^*
    \approx_{\,n\gamma\sqrt{\epsilon}}
    0 .
\end{align*}
In particular, the observable \(A\) approximately commutes with all canonical Magic Square observables on the same register.

We now recall why this is impossible. Let \(A_0\in \mathcal{H}_4\) be a traceless Hermitian matrix commuting with all canonical Magic Square observables.
Since the canonical table contains
\begin{align*}
    X\otimes I,
    \qquad
    Z\otimes I,
    \qquad
    I\otimes X,
    \qquad
    I\otimes Z,
\end{align*}
the conditions
\begin{align*}
    [A_0,X\otimes I]=[A_0,Z\otimes I]=0
\end{align*}
imply that
\begin{align*}
    A_0 = I\otimes B
\end{align*}
for some \(B\in \mathcal{H}_2\). Then
\begin{align*}
    [A_0,I\otimes X]=[A_0,I\otimes Z]=0
\end{align*}
imply that \(B\) commutes with both \(X\) and \(Z\), hence \(B\) is a scalar multiple of \(I\). Since \(A_0\) is traceless,
that scalar must be zero, so
\begin{align*}
    A_0=0 .
\end{align*}
Therefore, there is no nonzero traceless observable on a single \(4\)-dimensional register that commutes with
all canonical Magic Square observables.

Applying the same reasoning to observable \(A\), the approximate commutation relations
\begin{align*}
[A,X\otimes I]\approx_{\,n\gamma\sqrt{\epsilon}}0,
\qquad
[A,Z\otimes I]\approx_{\,n\gamma\sqrt{\epsilon}}0,
\\
[A,I\otimes X]\approx_{\,n\gamma\sqrt{\epsilon}}0,
\qquad
[A,I\otimes Z]\approx_{\,n\gamma\sqrt{\epsilon}}0
\end{align*}
imply that
\[
A\approx_{\,n\gamma\sqrt{\epsilon}}0.
\]
On the other hand, \(A\) is unitarily equivalent to \(\widetilde Q_{i,z_1}^{(s)}\), hence \(A\) is also a traceless binary observable and therefore nonzero. This contradiction shows that \(s_i\neq s_j\) whenever \(i\neq j\). Thus \eqref{eq:2n-ms-distinct} holds.
\end{proof}

\section{Discussion and Future Directions}
Our characterization of the maximal winning probabilities in nonlocal games under noise makes a first step toward understanding the noise robustness of nonlocal games. Furthermore, our rigidity results provide a precise description of the players' binary observables, making them particularly useful for a verifier in an $\MIP^*$ proof system, just like previous rigidity theorems in the noiseless setting. 
Hence they are relevant for analyzing quantum proof systems in the presence of noise. Finally, our work highlights fundamental differences between the rigidity theorems in the noisy versus noiseless settings, although in both cases the players can have quantum states of arbitrary dimensions.
Hence, our results reveal new structural and technical challenges unique to the high-noise regime.

\begin{description}
    \item[Register Concentration.] In the noiseless setting, local isometries are necessary to extract the ideal observables onto a specific register. In contrast, in the presence of noise, we can \emph{directly} prove that the players' observables act non-trivially on only one register without extra operation.
    \item[Simple Unitary.] As a result of the first item, to extract a Pauli pair $\br{Z, X}$, we only need to apply a unitary on one qubit. However, in the noiseless case, we use a swap isometry, which introduces an extra qubit and operates on the whole system. This difference partly stems from assuming a trusted but noisy source, where the form the source is assumed to produce i.i.d. copies, each one with bounded maximal correlation. That is a feature different from the noiseless setting, but this difference is still not obvious since the provers can share arbitrarily many copies of noisy states, i.e., the dimension can be arbitrarily large.
    \item[Measurement Projectivity.] Noiseless rigidity results can focus on projective measurements because Naimark's theorem implies that for any quantum strategy with POVMs, there is another quantum strategy
    with projective measurements that achieves the same winning probability by introducing ancilla qubits. However, in our trusted-but-noisy setting, the source can only prepare states with one fixed form, and we cannot apply Naimark's theorem directly. Thus, we use new insights and techniques to establish the projectivity of the measurements.
    \item[Noise Tolerance.] Previous work \cite{kaniewski2016analytic} shows that for the CHSH game, if the violation exceeds $\frac{16+14\sqrt{2}}{17}\approx 2.11$, a non-trivial fidelity between the shared state and the EPR state can be inferred. In comparison, in the CHSH game under noise, if we compare to all classical strategies, our rigidity results for observables hold for any correlation parameter $\rho \in\br{1/\sqrt{2},1}$.
    

\end{description}

Our result also raises some natural but intriguing questions. We list some of them below.

\begin{enumerate}
    \item What are the maximal violations for CHSH and Magic Square game under noise without the traceless condition? In our paper we use a trace test to ensure that the observables have bounded trace. 
    It will complete the picture if the maximum winning probability can be derived and 
    similar rigidity results can be proved without the traceless condition.
    \item What are the maximum winning probabilities when players share other types of noisy quantum states?
    In our paper we primarily focus on the depolarizing noise and extend our analysis to more general types of noise. However, we still need the condition that the marginal state is maximally mixed. 
    Moreover, for the Magic Square game under noise, we analyze the winning probability when a depolarizing channel is applied to a 4-qubit system. Determining the optimal winning probability under qubit depolarizing noise and more general noise remains an open problem.
\end{enumerate}

\section*{Acknowledgment}
M.Q. is supported by the National Research Foundation, Singapore through the National Quantum Office, hosted in A*STAR, under its Centre for Quantum Technologies Funding Initiative (S24Q2d0009). H.X. was supported by Beijing Nova Program (Grant No. 20220484128 and 20240484652).
P.Y. was supported by National Natural Science Foundation of China (Grant No. 62332009, 12347104), Innovation Program for Quantum Science and Technology (Grant No. 2021ZD0302901), NSFC/RGC Joint Research Scheme (Grant no. 12461160276), Fundamental and Interdisciplinary Disciplines Breakthrough Plan of the Ministry of Education of China (No. JYB2025XDXM118), Natural Science Foundation of Jiangsu Province (No. BK20243060).

\appendix

\section{Deferred Proofs of Lemmas}\label{Deferred0-Proofs-of-Lemmas}
\begin{fact}\label{approx_multiply}
    For $\epsilon>0$ and Hermitian matrices $A,B,C,D\in\H_d$, with $\max\set{\|A\|,\|B\|,\|C\|,\|D\|}=O\br{1}$, if $A\approx_\epsilon C$ and $B\approx_\epsilon D$, then $AB\approx_{\epsilon} CD$.
\end{fact}
\begin{proof}
    Note that 
   \[
    \begin{aligned}
        \frac{1}{\sqrt{d}}\|AB-CD\|_F
        &\leq \frac{1}{\sqrt{d}}\br{\|A\br{B-D}\|_F+\|\br{A-C}D\|_F} \\
        &\leq \frac{1}{\sqrt{d}}\br{\|A\|\|B-D\|_F+\|D\|\|A-C\|_F} \\
        &=O\br{\epsilon}, 
    \end{aligned} 
     \] 
     thus $AB\approx_{\epsilon} CD$.
\end{proof}

\begin{lemma}\label{closetobservable}
    For a Hermitian matrix $Q\in\H_d$ with $\opnorm{Q}\leq1$, if $\widebar{\Tr}\br{Q^2}\approx_\epsilon 1$, then $Q^2\approx_{\sqrt{\epsilon}} I$.
\end{lemma}

\begin{proof}
    Since
    \begin{align*}
    \widebar{\Tr}\br{Q^2}&=1-O\br{\epsilon},\\
        \widebar{\Tr}\br{\br{Q^2-I}^2}&=\widebar{\Tr} \br{Q^4}+1-2\widebar{\Tr}\br{Q^2}\\
        &=\widebar{\Tr} \br{Q^4}-1+2O\br{\epsilon}=O\br{\epsilon}
    \end{align*}
    using the fact that $\widebar{\Tr} \br{Q^4} \leq 1$ and $\widebar{\Tr}\br{\br{Q^2-I}^2}\geq 0$, we conclude the result.
\end{proof}

\begin{lemma}\label{general-obs-scaling}
Let \(n\in \posint\), \(m\ge 2\), and \(d=m^n\). Let
\[
1=c_0 \ge c_1 \ge \cdots \ge c_{m^2-1}\ge 0,
\qquad r:=c_1,
\]
and let \(\{\mathcal{B}_i\}_{i=0}^{m^2-1}\) be a standard orthonormal basis of \(\H_m\).
For a Hermitian matrix \(Q\in \H_d\), write its decomposition under \(\mathcal{B}\) as
\[
Q=\sum_{x\in [m^2]_{\geq 0}^n}\hat{Q}(x)\,\mathcal{B}_x,
\qquad
\mathcal{B}_x:=\bigotimes_{t=1}^n \mathcal{B}_{x_t}.
\]
Define
\[
Q'
:=
\sum_{x\in [m^2]_{\geq 0}^n}
\left(\prod_{j=0}^{m^2-1} c_j^{\,w_j(x)}\right)
\hat{Q}(x)\,\mathcal{B}_x,
\]
where, for each \(j\in [m^2]_{\geq 0}\),
\[
w_j(x):=\bigl|\{\,i\in [n]:x_i=j\,\}\bigr|.
\]
If
\[
\hat{Q}(0^n)^2\approx_{\epsilon_1} 0,
\qquad
\widebar{\Tr}(Q^2)\approx_{\epsilon_2} 1,
\qquad
\widebar{\Tr}\bigl((Q')^2\bigr)\approx_{\epsilon_3} r^2,
\]
then
\[
Q' \approx_{\sqrt{\epsilon_1+\epsilon_2+\epsilon_3}} rQ.
\]
\end{lemma}

\begin{proof}
Let
\[
Q_0:=\hat{Q}(0^n)\,\mathcal{B}_{0^n},
\qquad
\widetilde{Q}:=Q-Q_0,
\qquad
\widetilde{Q}':=Q'-Q_0.
\]
Since \(c_0=1\), the coefficient of \(\mathcal{B}_{0^n}\) in \(Q'\) is also \(\hat{Q}(0^n)\). Hence
\[
\widebar{\Tr}(Q_0^2)=\hat{Q}(0^n)^2=O(\epsilon_1),
\]
and therefore
\[
Q\approx_{\sqrt{\epsilon_1}} \widetilde{Q},
\qquad
Q'\approx_{\sqrt{\epsilon_1}} \widetilde{Q}'.
\]
Thus it suffices to prove
\[
\widetilde{Q}'\approx_{\sqrt{\epsilon_2+\epsilon_3}} r\widetilde{Q}.
\]
From now on, we assume that \(Q\) is traceless, i.e. \(\hat{Q}(0^n)=0\).

For each \(x\neq 0^n\), define
\[
\lambda_x:=\prod_{j=1}^{m^2-1} c_j^{\,w_j(x)}.
\]
Then
\[
Q'=\sum_{x\neq 0^n}\lambda_x \hat{Q}(x)\,\mathcal{B}_x.
\]
Since \(x\neq 0^n\), at least one coordinate of \(x\) is nonzero, so
\[
0\le \lambda_x\le c_1=r.
\]

By Parseval's identity,
\[
t_1:=\sum_{x\neq 0^n}\hat{Q}(x)^2
=
\widebar{\Tr}(Q^2)
\approx_{\epsilon_2} 1,
\]
and
\[
t_2:=\sum_{x\neq 0^n}\lambda_x^2\hat{Q}(x)^2
=
\widebar{\Tr}\bigl((Q')^2\bigr)
\approx_{\epsilon_3} r^2.
\]

Now compute
\begin{align*}
&\widebar{\Tr}\bigl((Q'-rQ)^2\bigr)\\
=&
\sum_{x\neq 0^n}(\lambda_x-r)^2\hat{Q}(x)^2\\
=&
\sum_{x\neq 0^n}\lambda_x^2\hat{Q}(x)^2
+
r^2\sum_{x\neq 0^n}\hat{Q}(x)^2
-
2r\sum_{x\neq 0^n}\lambda_x\hat{Q}(x)^2\\
=&
t_2+r^2t_1-2r\sum_{x\neq 0^n}\lambda_x\hat{Q}(x)^2.
\end{align*}

Since \(0\le \lambda_x\le r\), we have \(\lambda_x^2\le r\lambda_x\). Multiplying by
\(\hat{Q}(x)^2\) and summing over \(x\neq 0^n\), we obtain
\[
t_2
=
\sum_{x\neq 0^n}\lambda_x^2\hat{Q}(x)^2
\le
r\sum_{x\neq 0^n}\lambda_x\hat{Q}(x)^2.
\]
Therefore,
\[
\widebar{\Tr}\bigl((Q'-rQ)^2\bigr)
\le
r^2 t_1 - t_2
=
O(\epsilon_2+\epsilon_3).
\]
Combining this with the degree-zero truncation error \(O(\epsilon_1)\), we conclude that
\[
\widebar{\Tr}\bigl((Q'-rQ)^2\bigr)
=
O(\epsilon_1+\epsilon_2+\epsilon_3).
\]
Equivalently,
\[
Q' \approx_{\sqrt{\epsilon_1+\epsilon_2+\epsilon_3}} rQ.
\]
\end{proof}

\begin{fact}\label{anti-commute-to-ZX}
    For traceless Hermitian matrices $A,B\in\H_2$, if $A^2= I,B^2=I,\set{A,B}\approx_\epsilon 0$, then there exists a unitary $U$ such that $UAU^*\approx_\epsilon Z,UBU^*\approx_\epsilon X$.
\end{fact}
\begin{proof}
    Set $A=\sum_{i=1}^3 a_i \sigma_i,B=\sum_{i=1}^3b_i\sigma_i$. By \cref{commute-anticommute}, we have $\set{A,B}=0$, which gives
    \begin{align*}
        \left|\sum_{i=1}^3 a_ib_i\right|\approx_\epsilon 0. 
    \end{align*}
    It is easy to see that there exists unitary $U_1$, such that $U_1AU_1^*= Z$. Set $C=U_1BU_1^*=\sum_{i=1}^3c_i\sigma_i$. We have $\set{C,Z}\approx_{\epsilon} 0$. By \cref{commute-anticommute}, $|c_3|\approx_{\epsilon}0$. For $\theta=\arccos\frac{c_1}{\sqrt{c_1^2+c_2^2}}$, define
    \begin{align*}
        U_2=\cos{\frac{\theta}{2}}I-i\sin{\frac{\theta}{2}}Z.
    \end{align*}
    We have $U_2CU_2^*=\sqrt{1-c_3^2}X+c_3Z\approx_{\epsilon} X$ and $U_2 Z U_2^*=Z$. Thus $U=U_2U_1$ satisfies the requirements.
\end{proof}


\begin{fact}\label{sqrt-approx-project}
    For a Hermitian matrix $A\in\H_d$, if $A^2\approx_\epsilon I$, then there exist two projectors $\Pi_0,\Pi_1$ with $\Pi_0+\Pi_1=I$, such that 
    \begin{align*}
        A\approx_\epsilon \Pi_0-\Pi_1.
    \end{align*}
\end{fact}
\begin{proof}
    Let $A=\sum_i\lambda_i\ketbra{\lambda_i}$ be the spectral decomposition. From $A^2\approx_\epsilon I$, we have
    \begin{align*}
        \frac{1}{d}\sum_i \br{\lambda_i^2-1}^2 = O\br{\epsilon^2}.
    \end{align*}
     Define $S_0=\set{i:\lambda_i\geq 0}, S_1=\set{i:\lambda_i<0}$. For $j\in \set{0,1},$ define $\Pi_j=\sum_{i\in S_j}\ketbra{\lambda_i}$, then
     \begin{align*}
         &\frac{1}{d}\|A-\Pi_0+\Pi_1\|_F^2\\
         =~&\frac{1}{d}\br{\sum_{i\in S_0}\br{\lambda_i-1}^2+\sum_{i\in S_1}\br{\lambda_i+1}^2}\\
         \leq~&\frac{1}{d}\br{\sum_{i\in S_0}\br{\lambda_i-1}^2\br{\lambda_i+1}^2+\sum_{i\in S_1}\br{\lambda_i+1}^2\br{\lambda_i-1}^2}
         \\
         =~&\frac{1}{d}\sum_i \br{\lambda_i^2-1}^2 =O\br{\epsilon^2}.
     \end{align*}
     Thus $A\approx_\epsilon \Pi_0-\Pi_1.$
\end{proof}
\begin{fact}\label{commute-anticommute}
    Let $A,B$ be two $2$-dimensional non-zero, traceless Hermitian matrices expressed as $A=\sum_{i=1}^3 a_i\sigma_i,B=\sum_{i=1}^3 b_i\sigma_i,$ then
    \begin{align*}
        \set{A,B}=0&\Leftrightarrow \sum_{i=1}^3a_ib_i=0,\\
        \left[A,B\right]=0 &\Leftrightarrow \exists k\in\reals,\forall i\in\left[3\right], a_i=kb_i.
    \end{align*}
\end{fact}
\begin{proof}
    By direct calculation,
    \begin{align*}
        AB=\sum_{i=1}^3a_ib_i I+\sum_{1\leq i<j\leq 3}\br{a_ib_j-a_jb_i}\sigma_i\sigma_j,\\
    BA=\sum_{i=1}^3a_ib_i I-\sum_{1\leq i<j\leq 3}\br{a_ib_j-a_jb_i}\sigma_i\sigma_j. \end{align*}
    Let $\vec{a}=\br{a_1,a_2,a_3}^T,\vec{b}=\br{b_1,b_2,b_3}^T$.
    If $\set{A,B}=0$, comparing the items gives $\langle \vec{a},\vec{b}\rangle=\sum_{i=1}^3a_ib_i=0$. If $\left[A,B\right]=0$, we have $\forall 1\leq i<j\leq  3,a_ib_j=a_jb_i$, which gives $\vec{a}\times\vec{b}=\vec{0}$, and $\exists k\in\reals,\forall i\in\left[3\right], a_i=kb_i$.
\end{proof}
\begin{fact}\label{ONB-anitcommute}
    For any standard orthonormal basis $\set{\A_i:i\in\set{0,1,2,3}}$ of $\H_2$, we have $\set{\A_i,\A_j}=0$ for $1\leq i<j\leq 3$.
\end{fact}
\begin{proof}
    
    By \cite[Fact 2.8]{qin2021nonlocal}, there exists a $4\times 4$ orthogonal matrix $O=\br{o_{ij}}$, such that $\A_i=\sum_{j=0}^3 o_{ij}\sigma_j$. Therefore for $1\leq i<j\leq 3,
        \set{\A_i,\A_j}=2\sum_{ijk}o_{ik}o_{jk}I=0.$
\end{proof}

\begin{fact}\label{ONBtoPauli}
    For any standard orthonormal basis $\set{\A_i:i\in\set{0,1,2,3}}$ of $\H_2$, there exists a 2-dimensional unitary $U$, such that for $i\in\set{0,1,3}$, $ U \A_i U^* =\sigma_i$; and $U \A_2 U^* =\pm Y$.    
\end{fact}
\begin{proof}
    By \cref{ONB-anitcommute}, $\set{\A_1,\A_3}=0$. Standard matrix techniques (\cite[Lemma 8.2]{cleve2020qic}) ensure the existence of a 2-dimensional unitary $U$, such that $U\A_1U^*=X, U\A_3U^*=Z$. Since $U\A_2U^*$ anti-commutes with both $U\A_1U^*$ and $U\A_3U^*$, by \cref{commute-anticommute} $U\A_2U^*$ is either $Y$ or $-Y$.
\end{proof}
\begin{fact}\label{pre-not-separable}
    Let $\varphi=\frac{1}{4}\br{I\x I+a X\x X+b Y\x Y+c Z\x Z}$, where $a,b,c\in[0,1]$. If $\varphi$ is a valid quantum state, then $\varphi$ is separable.
\end{fact}

\begin{proof}
    By direct calculation, the eigenvalues of $\varphi$ are $$\frac{1}{4}\{1-a-b-c,1+a+b-c,1+a-b+c,1-a+b+c\}.$$ If $\varphi$ is a valid quantum state, then $a+b+c\leq1.$ Consider the partial transpose of $\varphi$:
    $$\varphi^{\text{PT}}=\frac{1}{4}\br{I\x I+a X\x X-b Y\x Y+c Z\x Z}.$$
    The eigenvalues of $\varphi^{\text{PT}}$ are $$\frac{1}{4}\{1-a+b-c,1+a-b-c,1+a+b+c,1-a-b+c\}.$$
    It is easy to verify that, if $a+b+c\leq1$, then $\varphi^{\text{PT}}$ is positive semidefinite. Thus $\varphi$ is a PPT state. By the well-known fact that any PPT state with local dimension 2 is separable, $\varphi$ must be separable.
\end{proof}

\begin{fact}\label{EquivtoEPR}
For any two standard orthonormal bases $\set{\A_i:i\in\set{0,1,2,3}},\set{\B_i:i\in\set{0,1,2,3}}$, 
let $\varphi=\frac{1}{4}\br{I\x I+a \A_1\x \B_1+b \A_2\x \B_2+c \A_3\x \B_3}$, where $a,b,c\in[0,1]$. If $\varphi$ is a valid quantum state, and is not separable, then the matrix $\phi =\frac{1}{4}\sum_{i=0}^3\A_i\x\B_i$ is equivalent to EPR state, up to local unitaries. 
\end{fact}
\begin{proof}
    By \cref{ONBtoPauli}, there exist unitaries $U,V$, such that $U\A_iU^*=\sigma_i,V\B_iV^*= \sigma_i$ for $i\in \set{0,1,3}$ and $U\A_2U^*=\pm Y,V\B_2V^*= \pm Y$. If $\varphi$ is not separable, then $\br{U\otimes V}\varphi\br{U^*\otimes V^*}$ is also not separable. By \cref{pre-not-separable}, we must have $\br{U\otimes V}\A_2\x \B_2\br{U^*\otimes V^*}=-Y\x Y$, thus
        \begin{multline*}
        \br{U\otimes V}\phi\br{U^*\otimes V^*} =\frac{1}{4}\br{I\otimes I+X\otimes X-Y\otimes Y+Z\otimes Z}\\=\ketbra{\Phi}.
    \end{multline*}

\end{proof}

\section{Trace test for the CHSH game under noise}
\label{app:chsh-trace-test}

For completeness, we record a simple statistical test that enforces the trace bound used in
\cref{self-test-depolarize}.


\begin{fact}[Hoeffding's inequality]\label{fact:Hoeffding}
Let \(\randX_1,\dots,\randX_t\) be independent random variables such that
\[
a_i\le \randX_i \le b_i.
\]
Let \(\randS=\sum_{i=1}^t \randX_i\). Then
\[
\Pr\!\left[\left|\randS-\bigE[\randS]\right|\ge s\right]
\le
2\exp\!\left(-\frac{2s^2}{\sum_{i=1}^t (b_i-a_i)^2}\right).
\]
\end{fact}

\begin{prop}[Trace test for the CHSH game under noise]\label{prop:chsh-trace-test-appendix}
Assume \cref{repetition-iid}. Fix parameters \(t\in \mathbb{Z}_{>0}\), \(\delta>0\), and
\(p\in(0,1)\) satisfying
\[
\delta=\left(\frac{2\ln(2/p)}{t}\right)^{1/2}.
\]
Suppose that, for each player \(D\in\{A,B\}\) and each question \(x\in\{0,1\}\), we repeat the
CHSH game under noise until the pair \((D,x)\) appears at least \(t\) times, and we reject whenever the empirical
frequency \(\widehat{p}_D(1|x)\) deviates from \(1/2\) by more than \(\delta\). If the players pass this
test with probability at least \(p\), then
\[
\Trerr\le 3\delta.
\]
\end{prop}

\begin{proof}
Since the local marginals of \(\Psi^{\otimes n}\) are maximally mixed, if a player measures an
observable \(D_x\), then the probability of answering \(1\) is
\[
\frac{1}{2}+\frac{\widebar{\Tr}(D_x)}{2}.
\]
Without loss of generality, assume
\[
\Trerr=\widebar{\Tr}(P_0).
\]
Suppose toward a contradiction that
\[
\Trerr>3\delta.
\]
Let \(\randX_1,\dots,\randX_t\in\{-1,1\}\) be Alice's first \(t\) answers to question \(0\). Then
\[
\bigE[\randX_i]=\Trerr.
\]
Passing the trace test for this question implies
\[
\left|\frac{1}{t}\sum_{i=1}^t \randX_i\right|\le 2\delta.
\]
Hence
\begin{align*}
\Pr[\text{pass the trace test}]
&\le
\Pr\!\left[\left|\sum_{i=1}^t \randX_i\right|\le 2\delta t\right]\\
&\le
\Pr\!\left[\left|\sum_{i=1}^t \randX_i-t\bigE[\randX_1]\right|
\ge (\Trerr-2\delta)t\right].
\end{align*}
Applying \cref{fact:Hoeffding} gives
\begin{multline*}
    \Pr[\text{pass the trace test}]
\le
2\exp\!\left(-\frac{(\Trerr-2\delta)^2 t}{2}\right)
\\<
2\exp\!\left(-\frac{\delta^2 t}{2}\right)
=
p,
\end{multline*}
a contradiction. Therefore
\[
\Trerr\le 3\delta.
\]
\end{proof}

\section{Proof of Lemmas Supporting \cref{chsh-selftest-main-thm}}\label{appendix-lem-supporting-chsh-self-test}

In this subsection we prove the auxiliary lemmas used in the shortened proof of \cref{chsh-selftest-main-thm}.

\begin{proof}[Proof of \cref{lem:single-observable-localization}]
It suffices to prove the claim for \(Q_0\); the other three cases are identical.

By \cref{thm-close-to-deg1},
\[
Q_0' \approx_{\sqrt{\epsilon}} \rho Q_0.
\]
Since \(\widebar{\Tr}(Q_0)^2=O(\epsilon)\), the degree-zero coefficient of \(Q_0\) is \(O(\epsilon)\). Hence
\begin{align*}
O(\epsilon)
&=
\widebar{\Tr}\!\left((Q_0'-\rho Q_0)^2\right)\\
&\approx_{\epsilon}
\sum_{|x|\ge 2}
\bigl(\rho-\rho^{|x|}\bigr)^2 \widehat{Q}_0(x)^2\\
&\ge
(\rho-\rho^2)^2
\sum_{|x|\ge 2}\widehat{Q}_0(x)^2.
\end{align*}
Therefore,
\[
\sum_{|x|\ge 2}\widehat{Q}_0(x)^2
=
O(\epsilon\gamma^2).
\]

Let \(\overline{Q}_0\) be the degree-one truncation of \(Q_0\). Then
\[
Q_0 \approx_{\sqrt{\epsilon}\gamma} \overline{Q}_0,
\]
and \(\overline{Q}_0\) has the form
\[
\overline{Q}_0
=
\sum_{r=1}^n a_r O^{(r)}\otimes I^{[n]\setminus\{r\}},
\]
where \(a_r\ge 0\), each \(O^{(r)}\in\H_2\) is traceless, and
\[
\bigl(O^{(r)}\bigr)^2=I.
\]

Since \(Q_0^2\approx_{\sqrt{\epsilon}}I\), we also have
\[
\overline{Q}_0^{\,2}\approx_{\sqrt{\epsilon}\gamma} I.
\]
Expanding \(\overline{Q}_0^{\,2}\) and applying Parseval's identity gives
\[
\sum_{r=1}^n a_r^2 \approx_{\sqrt{\epsilon}\gamma} 1,
\qquad
\sum_{r\neq s} a_r^2 a_s^2 = O(\epsilon\gamma^2).
\]
Hence
\[
\sum_{r=1}^n a_r^2(1-a_r^2)=O(\epsilon\gamma^2).
\]
Since every term on the left is nonnegative, there exists an index \(k_{Q,0}\in[n]\) such that
\[
a_{k_{Q,0}}^2 = 1-O(\epsilon\gamma^2),
\qquad
\sum_{r\neq k_{Q,0}} a_r^2 = O(\epsilon\gamma^2).
\]
Set
\[
\widetilde{Q}_0^{(k_{Q,0})}:=O^{(k_{Q,0})}.
\]
Then
\[
Q_0 \approx_{\sqrt{\epsilon}\gamma}
\widetilde{Q}_0^{(k_{Q,0})}\otimes I^{[n]\setminus\{k_{Q,0}\}}.
\]

The proofs for \(P_0,P_1,Q_1\) are identical.
\end{proof}

\begin{proof}[Proof of \cref{lem:same-register}]
Assume first that \(D=P\); the argument for \(D=Q\) is identical.

Suppose for contradiction that \(k_{P,0}\neq k_{P,1}\). Since \(P_0^2\approx_{\sqrt{\epsilon}}I\) and \(P_1^2\approx_{\sqrt{\epsilon}}I\), the single-register approximations imply
\[
\bigl(\widetilde{P}_0^{(k_{P,0})}\bigr)^2
\approx_{\sqrt{\epsilon}\gamma} I,
\qquad
\bigl(\widetilde{P}_1^{(k_{P,1})}\bigr)^2
\approx_{\sqrt{\epsilon}\gamma} I.
\]
Therefore, by \cref{approx_multiply},
\begin{align*}
\{P_0,P_1\}
&\approx_{\sqrt{\epsilon}\gamma}
\Bigl\{
\widetilde{P}_0^{(k_{P,0})}\otimes I^{[n]\setminus\{k_{P,0}\}},
\widetilde{P}_1^{(k_{P,1})}\otimes I^{[n]\setminus\{k_{P,1}\}}
\Bigr\}\\
&=
2\,
\widetilde{P}_0^{(k_{P,0})}
\otimes
\widetilde{P}_1^{(k_{P,1})}
\otimes
I^{[n]\setminus\{k_{P,0},k_{P,1}\}}.
\end{align*}
The right-hand side has normalized Frobenius norm bounded below by a positive constant, because each factor squares to \(I\) up to \(O(\sqrt{\epsilon}\gamma)\) error. This contradicts
\[
\{P_0,P_1\}\approx_{\sqrt{\epsilon}}0.
\]
Hence
\[
k_{P,0}=k_{P,1}.
\]
The same argument gives
\[
k_{Q,0}=k_{Q,1}.
\]
\end{proof}

\begin{proof}[Proof of \cref{lem:register-matching}]
Suppose for contradiction that \(k_1\neq k_2\). Then for each \(i\in\{0,1\}\), the left-hand side
\[
\widetilde{P}_i^{(k_1)}\otimes I^{[n]\setminus\{k_1\}}
\]
is a degree-one operator on the \(k_1\)-th register, whereas the right-hand side
\[
V^{\otimes n}
\left(
\frac{\widetilde{Q}_0^{(k_2)}+(-1)^i\widetilde{Q}_1^{(k_2)}}{\sqrt{2}}
\right)
(V^*)^{\otimes n}
\otimes I^{[n]\setminus\{k_2\}}
\]
is a degree-one operator on the \(k_2\)-th register. These two operators are orthogonal with respect to the normalized Hilbert--Schmidt inner product.

Moreover, each side has normalized Frobenius norm
\[
1+O(\sqrt{\epsilon}\gamma),
\]
since
\[
\bigl(\widetilde{P}_i^{(k_1)}\bigr)^2\approx_{\sqrt{\epsilon}\gamma}I
\]
and
\[
\left(
\frac{\widetilde{Q}_0^{(k_2)}+(-1)^i\widetilde{Q}_1^{(k_2)}}{\sqrt{2}}
\right)^2
\approx_{\sqrt{\epsilon}\gamma}I.
\]
Hence the two sides cannot be \(O(\sqrt{\epsilon}\gamma)\)-close, a contradiction. Therefore
\[
k_1=k_2.
\]
\end{proof}

\begin{proof}[Proof of \cref{lem:source-coefficient-rigidity}]
Let \(\{\mathcal{B}_i\}_{i=0}^3\) be the standard orthonormal basis fixed in the diagonalization of the correlation matrix, and write
\begin{equation}
\widetilde{Q}_0^{(k)}=\sum_{r=1}^3 \alpha_r \mathcal{B}_r,
\qquad
\widetilde{Q}_1^{(k)}=\sum_{r=1}^3 \beta_r \mathcal{B}_r,
\label{eq:appendix-thm37-basis-expansion}
\end{equation}
where
\[
\sum_{r=1}^3 \alpha_r^2=1,
\qquad
\sum_{r=1}^3 \beta_r^2=1.
\]
By the definition of the noisy observables and the single-register localization,
\begin{align}
Q_0'
&\approx_{\sqrt{\epsilon}\gamma}
\left(\sum_{r=1}^3 \alpha_r c_r \mathcal{B}_r^{(k)}\right)
\otimes I^{[n]\setminus\{k\}}, \label{eq:appendix-thm37-Q0prime-local}\\
Q_1'
&\approx_{\sqrt{\epsilon}\gamma}
\left(\sum_{r=1}^3 \beta_r c_r \mathcal{B}_r^{(k)}\right)
\otimes I^{[n]\setminus\{k\}}. \label{eq:appendix-thm37-Q1prime-local}
\end{align}
On the other hand, by \cref{almost-deg-one-2} and the fact that \(Q_0^2,Q_1^2\approx_{\sqrt{\epsilon}}I\), we have
\[
\widebar{\Tr}\!\left((Q_0')^2\right)\approx_{\epsilon}\rho^2,
\qquad
\widebar{\Tr}\!\left((Q_1')^2\right)\approx_{\epsilon}\rho^2.
\]
Combining this with \cref{eq:appendix-thm37-Q0prime-local,eq:appendix-thm37-Q1prime-local}, we obtain
\begin{equation}
\sum_{r=1}^3 \alpha_r^2 c_r^2
\approx_{\epsilon\gamma^2}
\rho^2,
\qquad
\sum_{r=1}^3 \beta_r^2 c_r^2
\approx_{\epsilon\gamma^2}
\rho^2.
\label{eq:appendix-thm37-weighted-rho}
\end{equation}

We first prove that \(c_1\) is close to \(\rho\). Since
\[
c_r\le c_1\le \rho
\qquad
\text{for all }r\in\{1,2,3\},
\]
and \(\sum_{r=1}^3\alpha_r^2=1\), \cref{eq:appendix-thm37-weighted-rho} implies
\[
\rho^2-O(\epsilon\gamma^2)
\le
\sum_{r=1}^3 \alpha_r^2 c_r^2
\le
c_1^2
\le
\rho^2.
\]
Hence
\[
0\le \rho^2-c_1^2=O(\epsilon\gamma^2).
\]
Since \(c_1\le \rho\), we get
\[
0\le \rho-c_1
=
\frac{\rho^2-c_1^2}{\rho+c_1}
=
O(\epsilon\gamma^2).
\]
In particular,
\[
c_1\approx_{\sqrt{\epsilon}\gamma}\rho.
\]

Substituting this back into \cref{eq:appendix-thm37-weighted-rho}, we also obtain
\begin{equation}
\sum_{r=1}^3 \alpha_r^2 c_r^2
\approx_{\epsilon\gamma^2}
c_1^2,
\qquad
\sum_{r=1}^3 \beta_r^2 c_r^2
\approx_{\epsilon\gamma^2}
c_1^2.
\label{eq:appendix-thm37-weighted-c}
\end{equation}

Next, since
\[
Q_0Q_1+Q_1Q_0\approx_{\sqrt{\epsilon}\gamma}0,
\]
the single-register reduction and \cref{commute-anticommute} imply
\[
\left|\sum_{r=1}^3 \alpha_r\beta_r\right|
\le \eta,
\qquad
\eta:=O(\sqrt{\epsilon}\gamma).
\]
Set
\[
\Delta:=c_1^2-c_2^2.
\]
From \cref{eq:appendix-thm37-weighted-c} and the ordering \(c_1\ge c_2\ge c_3\ge 0\), we have
\begin{align*}
c_1^2-\sum_{r=1}^3 \alpha_r^2 c_r^2
&=
\alpha_2^2(c_1^2-c_2^2)+\alpha_3^2(c_1^2-c_3^2)\\
&\ge
(\alpha_2^2+\alpha_3^2)(c_1^2-c_2^2)
=
(1-\alpha_1^2)\Delta.
\end{align*}
Hence
\[
1-\alpha_1^2\le \frac{C\eta^2}{\Delta}
\]
for some absolute constant \(C>0\). The same argument gives
\[
1-\beta_1^2\le \frac{C\eta^2}{\Delta}.
\]
Therefore,
\begin{align*}
\eta
&\ge
\left|\sum_{r=1}^3 \alpha_r\beta_r\right|\\
&\ge
|\alpha_1\beta_1|-|\alpha_2\beta_2+\alpha_3\beta_3|\\
&\ge
|\alpha_1\beta_1|
-
\sqrt{(\alpha_2^2+\alpha_3^2)(\beta_2^2+\beta_3^2)}\\
&\ge
\sqrt{1-\frac{C\eta^2}{\Delta}}\,
\sqrt{1-\frac{C\eta^2}{\Delta}}
-
\frac{C\eta^2}{\Delta}.
\end{align*}
For sufficiently small \(\epsilon\), we have \(\eta<1/2\), and the previous inequality implies
\[
\eta \ge 1-\frac{C'\eta^2}{\Delta}
\]
for another absolute constant \(C'>0\). Rearranging,
\[
\Delta\le \frac{C'\eta^2}{1-\eta}.
\]
Since
\[
(c_1-c_2)^2\le c_1^2-c_2^2=\Delta,
\]
we conclude that
\[
c_1-c_2
=
O\!\left(\frac{\eta}{\sqrt{1-\eta}}\right)
=
O(\sqrt{\epsilon}\gamma).
\]
Consequently,
\[
\rho-c_2
=
(\rho-c_1)+(c_1-c_2)
=
O(\sqrt{\epsilon}\gamma),
\]
and therefore
\[
\rho \ge c_1 \ge c_2 \ge \rho - O(\sqrt{\epsilon}\gamma).
\]
\end{proof}

\section{Proof of lemmas supporting \cref{mgs-self-testing-main-thm}}\label{app:mgs}
In this section, we provide detailed proofs of the preparatory lemmas used in the proof of \cref{mgs-self-testing-main-thm}. For the reader's convenience, we restate each lemma before its proof.

\begin{replemma}{mgs-consist-row-col}
Suppose the players pass the consistency test in the Magic Square game under noise (\cref{def-mgs-game}) with probability at least $\frac{1+\rho}{2}-\winerr$. Let $\epsilon=\winerr/\rho+\Trerr^2/\rho^2$, where $\Trerr$ is defined in \cref{eqn:traceerror}. Then for all $i,j\in[3]$, the following statements hold:
\begin{enumerate}
    \item $P_{i,j}^{\mathrm{row}}\approx_{\sqrt{\epsilon}} P_{i,j}^{\mathrm{column}}\approx_{\sqrt{\epsilon}}Q_{i,j}^T.$ We then omit superscripts of $P_{i,j}$.
    \item $P_{i,j}^2\approx_{\sqrt{\epsilon}}I$.
\end{enumerate}
\end{replemma}

To prove \cref{mgs-consist-row-col}, we need the following lemma, which shows that if the players pass the parity test with high probability, then Alice's POVM elements corresponding to incorrect parity can be neglected at the cost of a bounded error.
\begin{lemma}\label{prop:faultyboundederror}
Suppose the players pass the parity test in the Magic Square game under noise (\cref{def-mgs-game}) with probability at least $1-\epsilon$. Then 
for each $x\in\{r_1,r_2,r_3,c_1,c_2\}$,
\begin{equation}\label{eqn:correctPOVM1}
   \sum_{a_1,a_2,a_3:a_1a_2a_3=1}E_{a_1,a_2,a_3}^{x}\approx_{\sqrt{\epsilon} }I,
\end{equation}
and for $x = c_3$,
\begin{equation}\label{eqn:correctPOVM2}
\sum_{a_1,a_2,a_3:a_1a_2a_3=-1}E_{a_1,a_2,a_3}^{x}\approx_{\sqrt{\epsilon} }I.\end{equation}
\end{lemma}
\begin{proof}
    Since the probability that Alice answers incorrectly in the parity test is at most $\epsilon$, we have for any $x\in\{r_1,r_2,r_3,c_1,c_2\}$ and any $(a_1, a_2, a_3)$ such that $a_1a_2a_3=-1$,
    \begin{align*}
        \Tr\br{\br{E_{a_1,a_2,a_3}^{x} \otimes I }\br{\Delta_{\rho}\br{\Phi^{\x 2}} }^{\x n} }=\widebar{\Tr}~E_{a_1,a_2,a_3}^{x}=O\br{\epsilon}.
    \end{align*}
    Since  $E_{a_1,a_2,a_3}^{x}\preceq I$, we have 
    \[\widebar{\Tr}\br{E_{a_1,a_2,a_3}^{x}}^2\leq \widebar{\Tr}\br{E_{a_1,a_2,a_3}^{x}}= O\br{\epsilon},\] and hence $E_{a_1,a_2,a_3}^{x}\approx_{\sqrt{\epsilon}} 0$. 
    
    Since $\set{E_{a_1,a_2,a_3}^x}$ is a POVM, we have
    \[\sum_{a_1,a_2,a_3}E_{a_1,a_2,a_3}^{x}=I.\]
    Therefore \cref{eqn:correctPOVM1} holds. Using a similar argument, we can also prove \cref{eqn:correctPOVM2}.
\end{proof}

We are now ready to prove \cref{mgs-consist-row-col}.
\begin{proof}[Proof of \cref{mgs-consist-row-col}]
    Suppose Alice receives a row question containing the variable $s_{i,j}$ and Bob receives the corresponding question $s_{i,j}$. The probability that they pass the consistency test is
\begin{align*}
    &\frac{1}{2}+\frac{1}{2}\Tr\br{P_{i,j}^{\mathrm{row}}\x Q_{i,j} \br{\Delta_{\rho}\br{\Phi^{\x 2}} }^{\x n}}\\
    =~&\frac{1}{2}+\frac{1}{2}\Tr\br{P_{i,j}^{\mathrm{row}}\x Q_{i,j}' \br{\Phi^{\x 2n}}}\\
    =~&\frac{1}{2}+\frac{1}{2}\Tr\br{\br{P_{i,j}^{\mathrm{row}}}'\x Q_{i,j} \br{\Phi^{\x 2n}}},
\end{align*}
where the equalities follow from \cref{pre-noise-to-obs-mgs}.

Since the players pass the consistency test with probability $\frac{1+\rho}{2}-\winerr$, then 
\begin{multline}\label{eqn:Xsmall}
   \bra{\Phi^{\x 2n}}\rho I- P_{i,j}^{\mathrm{row}}\x Q_{i,j}'\ket{\Phi^{\x 2n}}\\=\bra{\Phi^{\x 2n}}\rho I- \br{P_{i,j}^{\mathrm{row}}}'\x Q_{i,j}\ket{\Phi^{\x 2n}}=2\winerr. 
\end{multline}

Define the following operators:
\begin{align*}
    X_1&:=\br{ P_{i,j}^{\mathrm{row}}\otimes I-I\otimes \frac{Q_{i,j}'}{\rho}}^2,\quad X_2:=I\otimes I-\br{P_{i,j}^{\mathrm{row}}}^2\otimes I,\\ X_3&:=I\otimes I-\frac{1}{\rho^2}I\otimes \br{Q_{i,j}'}^2.
\end{align*}

A direct calculation yields:
\begin{equation*}%
    \rho I- P_{i,j}^{\mathrm{row}}\x Q_{i,j}'=\frac{\rho}{2}\br{X_1+X_2+X_3}.
\end{equation*}

Thus
\begin{equation}\label{eqn:Xsum}
    \bra{\Phi^{\x 2n}}X_1+X_2+X_3\ket{\Phi^{\x 2n}}=\frac{4}{\rho}\winerr. 
\end{equation}

We now analyze each of the $X_k$ terms individually, showing that $\bra{\Phi^{\x 2n}} X_k\ket{\Phi^{\x 2n}}\approx_\epsilon 0$ for all $k\in[3]$. We first focus on $\bra{\Phi^{\x 2n}} X_3\ket{\Phi^{\x 2n}}.$

Since $X_1,X_2\succeq0$, by \cref{eqn:Xsum} we have
\begin{equation}\label{eqn:X3upper}
    \bra{\Phi^{\x 2n}} X_3\ket{\Phi^{\x 2n}}\leq\frac{4}{\rho}\winerr.
\end{equation}
On the other hand, by a similar argument as in \cref{almost-deg-one-1}, we know
\begin{equation*}
   \bra{\Phi^{\x 2n}}I\otimes \br{Q_{i,j}'}^2\ket{\Phi^{\x 2n}}\leq  \rho^2+\Trerr^2. 
\end{equation*}
This implies:
    \[\bra{\Phi^{\x 2n}} X_3\ket{\Phi^{\x 2n}}=1-\frac{1}{\rho^2}\bra{\Phi^{\x 2n}}I\otimes \br{Q_{i,j}'}^2\ket{\Phi^{\x 2n}}\geq-\frac{\Trerr^2}{\rho^2}.\]
Combining upper and lower bounds, and recall that $\epsilon=\winerr/\rho+\Trerr^2/\rho^2$, we have
    \begin{equation}\label{eqn:X3bound}\bra{\Phi^{\x 2n}} X_3\ket{\Phi^{\x 2n}}\approx_\epsilon0.\end{equation}

    Since $X_1,X_2\succeq0$, by \cref{eqn:Xsum}, we also have
    \begin{align}
        \bra{\Phi^{\x 2n}} X_1\ket{\Phi^{\x 2n}}&\approx_\epsilon0,\label{eqn:X1bound}\\
        \bra{\Phi^{\x 2n}} X_2\ket{\Phi^{\x 2n}}&\approx_\epsilon0.\label{eqn:X2bound}
    \end{align}
    From \cref{eqn:X2bound}, we obtain
    \begin{equation}\label{eqn:trPij2}
        \overline{\Tr}\br{P_{i,j}^{\mathrm{row}}}^2\approx_\epsilon1.
    \end{equation}
   Then by \cref{closetobservable}, we conclude
    \begin{equation}\label{mgs-alice-obs}
        \br{P_{i,j}^{\mathrm{row}}}^2\approx_{\sqrt{\epsilon}} I.
    \end{equation}

Using a symmetric argument applied to Bob's observable (via \cref{eqn:Xsmall}), we similarly obtain:
\begin{align}
    \overline{\Tr}~Q_{i,j}^2 &\approx_\epsilon 1,\label{mgs_obs_Q}\\
    Q_{i,j}^2&\approx_{\sqrt{\epsilon}} I.\label{mgs_obs_Q_1}
\end{align}
By \cref{eqn:X3bound}, we get
\begin{equation}\label{mgs_noisy_Q}
    \overline{\Tr}~\br{Q_{i,j}'}^2\approx_{\epsilon} \rho^2.
\end{equation}
Combining \cref{mgs_noisy_Q}, \cref{mgs_obs_Q} and \cref{general-obs-scaling} (with parameters $m=4,r=\rho$), we get
\begin{align}\label{mgs-bob-noise}
    I\otimes Q_{i,j}' \ket{\Phi^{\x 2n}}\approx_{\sqrt{\epsilon}}\rho I\otimes Q_{i,j}\ket{\Phi^{\x 2n}}.
\end{align}
On the other hand, from \cref{eqn:X1bound}, we have
\[I\otimes Q_{i,j}'\ket{\Phi^{\x 2n}}\approx_{\sqrt{\epsilon}}\rho P_{i,j}^{\mathrm{row}}\otimes  I\ket{\Phi^{\x 2n}}.\]
Applying \cref{vec-epr} and \cref{approx-epr} to the above two approximations, we conclude
\begin{equation*}
    P_{i,j}^{\mathrm{row}}\approx_{\sqrt{\epsilon}}Q_{i,j}^T.
\end{equation*}

A symmetric analysis of the case when Alice receives a column including $s_{i,j}$ yields the same bound:
    \begin{align*}
        P_{i,j}^{\mathrm{column}}&\approx_{\sqrt{\epsilon}}Q_{i,j}^T.
    \end{align*} 
    Hence we have $P_{i,j}^{\mathrm{row}}\approx_{\sqrt{\epsilon}} P_{i,j}^{\mathrm{column}}$, proving item 1 of the lemma.
Finally, combining this with \cref{mgs-alice-obs}, we get item 2.
\end{proof}

\begin{replemma}{mag-commute-same-row}
    Under the same assumptions as in \cref{mgs-consist-row-col}, 
    for any $i,i',j,j'\in[3]$ such that $i\ne i'$ and $j\ne j'$, we have 
    \[\left[P_{i,j},P_{i,j'}\right]\approx_{\sqrt{\epsilon}}0\quad\text{and}\quad \left[P_{i,j},P_{i',j}\right]\approx_{\sqrt{\epsilon}}0.\]
\end{replemma}

We rely on the following technical lemma to prove \cref{mag-commute-same-row}:
\begin{lemma}\label{lem:approx0}
    Given four positive semidefinite matrices $A,B,C,D\in\H_d$ with 
  \(\|A\|,\|B\|,\|C\|,\|D\|\leq1\), and two projectors $P$ and $Q$. If \[
    A+B \approx_\epsilon P,\quad
    A+C \approx_\epsilon Q,\quad
    A+B+C+D \approx_\epsilon I,
  \]
  then
  \[
    AB\approx_\epsilon0.
  \]
\end{lemma}
\begin{proof}
    Let $r=\mathrm{rank}(P)$. Without loss of generality, we write $P$ in block-diagonal form 
    \begin{equation}\label{eqn:defQ}
        P=
\begin{bmatrix}
I_r & 0 \\
0 & 0
\end{bmatrix},
    \end{equation}
and in the same $r\oplus(d-r)$ decomposition we set
\[
A = \begin{bmatrix}
A_{11} & A_{12} \\
A_{12}^* & A_{22}
\end{bmatrix},\qquad
B = \begin{bmatrix}
B_{11} & B_{12} \\
B_{12}^* & B_{22}
\end{bmatrix},\qquad
C = \begin{bmatrix}
C_{11} & C_{12} \\
C_{12}^* & C_{22}
\end{bmatrix},
\]
\[
D = \begin{bmatrix}
D_{11} & D_{12} \\
D_{12}^* & D_{22}
\end{bmatrix},\qquad
Q = \begin{bmatrix}
Q_{11} & Q_{12} \\
Q_{12}^* & Q_{22}
\end{bmatrix},
\]
where $A_{11},B_{11},C_{11},D_{11},Q_{11}\in\H_r$ and $A_{22},B_{22},C_{22},D_{22},Q_{22}\in\H_{d-r}\,$.
Then
\begin{align*}
 \|AB\|_F =~& \Fnorm{\begin{bmatrix}
A_{11}B_{11}+A_{12}B_{12}^* & A_{11}B_{12}+A_{12}B_{22} \\
A_{12}^*B_{11} +A_{22}B_{12}^*&A_{12}^*B_{12}+A_{22}B_{22}
\end{bmatrix}}\\
\leq~&\Fnorm{A_{11}B_{11}}+\Fnorm{A_{12}B_{12}^*}+\Fnorm{A_{11}B_{12}}+\Fnorm{A_{12}B_{22}}\\
+&\Fnorm{A_{12}^*B_{11}}+\Fnorm{A_{22}B_{12}^*}+\Fnorm{A_{12}^*B_{12}}+\Fnorm{A_{22}B_{22}}.   
\end{align*}
It suffices to prove that all these terms are $\sqrt{d}\cdot O(\epsilon)$. 

We first bound $\Fnorm{A_{12}B_{22}},$ $\Fnorm{A_{22}B_{12}^*}$ and $\Fnorm{A_{22}B_{22}}$. Since $A+B\approx_\epsilon P$, we have
\begin{equation}\label{eqn:B22}
    \frac{1}{\sqrt{d}}\Fnorm{A_{22}}=O(\epsilon)\quad\text{and}\quad \frac{1}{\sqrt{d}}\Fnorm{B_{22}}=O(\epsilon).
\end{equation}
 Then 

\begin{align}
    \frac{1}{\sqrt{d}}\Fnorm{A_{12}B_{22}}&\leq\frac{1}{\sqrt{d}}\Fnorm{B_{22}}\opnorm{A_{12}}=O(\epsilon),\label{eqn:A12B22}\\
    \frac{1}{\sqrt{d}}\Fnorm{A_{22}B_{12}^*}&\leq\frac{1}{\sqrt{d}}\Fnorm{A_{22}}\opnorm{B_{12}}=O(\epsilon),\label{eqn:A22B12}\\
    \frac{1}{\sqrt{d}}\Fnorm{A_{22}B_{22}}&\leq\frac{1}{\sqrt{d}}\Fnorm{A_{22}}\opnorm{B_{22}}=O(\epsilon).\label{eqn:A22B22}
\end{align}
    
We then bound $\Fnorm{A_{12}B_{12}^*}$ and $\Fnorm{A_{12}^*B_{12}}$.
Since $A+B\approx_\epsilon P$, we have $\frac{1}{\sqrt{d}}\Fnorm{A_{12}+B_{12}}=O(\epsilon).$ Thus,
\begin{align*}
\frac{1}{\sqrt{d}}\Fnorm{A_{12}B_{12}^*}&\leq\frac{1}{\sqrt{d}}\br{\Fnorm{A_{12}A_{12}^*}+\Fnorm{A_{12}\br{B_{12}^*+A_{12}^*}}}\\&\leq\frac{1}{\sqrt{d}}\br{\Fnorm{A_{12}A_{12}^*}+\opnorm{A_{12}}\Fnorm{B_{12}+A_{12}}}\\
&=\frac{1}{\sqrt{d}}\Fnorm{A_{12}A_{12}^*}+O(\epsilon).
\end{align*}
Similarly, we have
\begin{multline*}
    \frac{1}{\sqrt{d}}\Fnorm{A_{12}^*B_{12}}=\frac{1}{\sqrt{d}}\Fnorm{A_{12}^*A_{12}}+O(\epsilon)\\=\frac{1}{\sqrt{d}}\Fnorm{A_{12}A_{12}^*}+O(\epsilon).
\end{multline*}

Since $A,B,C\succeq0$, by \cite[ Proposition 1.3.2]{Rajendra2007Positive} there exist $r\times (d-r)$ matrices $K_A,K_B$ and $K_C$ with $\opnorm{K_A},\opnorm{K_B},\opnorm{K_C}\leq1$ such that 
\begin{multline*}
    A_{12}=\sqrt{A_{11}}K_A\sqrt{A_{22}},\quad B_{12}=\sqrt{B_{11}}K_B\sqrt{B_{22}},\\ C_{12}=\sqrt{C_{11}}K_C\sqrt{C_{22}}.
\end{multline*}

Then
\begin{multline}\label{eqn:A12A12}
\frac{1}{\sqrt{d}}\Fnorm{A_{12}A_{12}^*}=\frac{1}{\sqrt{d}}\Fnorm{\sqrt{A_{11}}K_AA_{22}K_A^*\sqrt{A_{11}}}\\\leq\frac{1}{\sqrt{d}}\opnorm{A_{11}}\opnorm{K_A}^2\Fnorm{A_{22}}\overset{(*)}{=}O(\epsilon),
\end{multline}
where $(*)$ uses \cref{eqn:B22}. Therefore 
\begin{equation}\label{eqn:A12B12}
    \frac{1}{\sqrt{d}}\Fnorm{A_{12}B_{12}^*}=O(\epsilon)\quad\text{and}\quad\frac{1}{\sqrt{d}}\Fnorm{A_{12}^*B_{12}}=O(\epsilon).
\end{equation}

We then bound $\Fnorm{A_{11}B_{11}}$.
Since $A+B\approx_\epsilon P$ and $A+B+C+D \approx_\epsilon I$, we have $C+D\approx_\epsilon I-P.$ Because $C,D\succeq0$, we have $\frac{1}{\sqrt{d}}\Fnorm{C_{11}}=O(\epsilon)$. Since $A+C\approx_\epsilon Q$, we have 
\begin{equation}\label{eqn:A11Q11}
    \frac{1}{\sqrt{d}}\Fnorm{A_{11}-Q_{11}}=O(\epsilon).
\end{equation} Since $A+B\approx_\epsilon P$, by Eq.~\eqref{eqn:defQ} we have $\frac{1}{\sqrt{d}}\Fnorm{A_{11}+B_{11}-I_r}=O(\epsilon)$. By Eq.~\eqref{eqn:A11Q11}, we have 
\begin{equation}\label{eqn:B11Q11}
    \frac{1}{\sqrt{d}}\Fnorm{B_{11}-(I_r-Q_{11})}=O(\epsilon). 
\end{equation}

Therefore,
\begin{align*}
    &\frac{1}{\sqrt{d}}\Fnorm{A_{11}B_{11}}\\
    \leq~&\frac{1}{\sqrt{d}}\left(\Fnorm{Q_{11}(I_r-Q_{11})}+\Fnorm{(A_{11}-Q_{11})(I_r-Q_{11})}\right.\\
    &+\left.\Fnorm{A_{11}(B_{11}-(I_r-Q_{11}))}\right)\\
    \leq~&\frac{1}{\sqrt{d}}\left(\Fnorm{Q_{11}(I_r-Q_{11})}+\Fnorm{A_{11}-Q_{11}}\opnorm{I_r-Q_{11}}\right.\\
    &+\left.\Fnorm{B_{11}-(I_r-Q_{11})}\opnorm{A_{11}}\right)\\
    =~&\frac{1}{\sqrt{d}}\Fnorm{Q_{11}(I_r-Q_{11})}+O(\epsilon).
\end{align*}

Since $Q^2=Q$, we have $Q_{11}^2+Q_{12}Q_{12}^*=Q_{11}$. Then 
\[\frac{1}{\sqrt{d}}\Fnorm{A_{11}B_{11}}=\frac{1}{\sqrt{d}}\Fnorm{Q_{12}Q_{12}^*}+O(\epsilon).\]
Since $A+C\approx_\epsilon Q$, we have $\frac{1}{\sqrt{d}}\Fnorm{A_{12}+C_{12}-Q_{12}}=O(\epsilon)$. Therefore, 
\[\frac{1}{\sqrt{d}}\Fnorm{A_{11}B_{11}}=\frac{1}{\sqrt{d}}\Fnorm{(A_{12}+C_{12})(A_{12}+C_{12})^*}+O(\epsilon).\]
Then
\begin{multline}\label{eqn:C12C12}
\frac{1}{\sqrt{d}}\Fnorm{C_{12}C_{12}^*}=\frac{1}{\sqrt{d}}\Fnorm{\sqrt{C_{22}}K_C^*C_{11}K_C\sqrt{C_{22}}}\\\leq\frac{1}{\sqrt{d}}\Fnorm{C_{11}}\opnorm{K_C}^2\opnorm{C_{22}}=O(\epsilon).
\end{multline}
And
\begin{multline*}
    \frac{1}{\sqrt{d}}\Fnorm{A_{12}C_{12}^*}=\frac{1}{\sqrt{d}}\sqrt{\Tr~A_{12}C_{12}^*C_{12}A_{12}^*}
    \\\overset{(*)}{\leq}\frac{1}{\sqrt{d}}\sqrt{\Fnorm{A_{12}^*A_{12}}\Fnorm{C_{12}^*C_{12}}}\overset{(**)}{=}O(\epsilon),
\end{multline*}
where $(*)$ follows from Cauchy-Schwartz inequality and $(**)$ uses \cref{eqn:A12A12,eqn:C12C12}. Similarly, we have $\frac{1}{\sqrt{d}}\Fnorm{C_{12}A_{12}^*}=O(\epsilon).$ Thus 
\begin{equation}\label{eqn:A11B11}
    \frac{1}{\sqrt{d}}\Fnorm{A_{11}B_{11}}=O(\epsilon).
\end{equation}

At last we bound $\Fnorm{A_{11}B_{12}}$ and $\Fnorm{A_{12}^*B_{11}}$. 
\begin{align}
    &\frac{1}{\sqrt{d}}\Fnorm{A_{11}B_{12}}\nonumber\\
    =~&\frac{1}{\sqrt{d}}\Fnorm{A_{11}\sqrt{B_{11}}K_B\sqrt{B_{22}}}\nonumber\\
    \overset{(*)}{\leq}~&\frac{1}{\sqrt{d}}\sqrt{\Fnorm{A_{11}\sqrt{B_{11}}\br{A_{11}\sqrt{B_{11}}}^*}\Fnorm{K_B\sqrt{B_{22}}\br{K_B\sqrt{B_{22}}}^*}}\nonumber\\
    =~&\frac{1}{\sqrt{d}}\sqrt{\Fnorm{A_{11}B_{11}A_{11}}\Fnorm{K_BB_{22}K_B^*}}\nonumber\\
    \leq~&\frac{1}{\sqrt{d}}\sqrt{\Fnorm{A_{11}B_{11}}\opnorm{A_{11}}\Fnorm{B_{22}}\opnorm{K_B}^2}\quad \overset{(**)}{=}O(\epsilon),\label{eqn:A11B12}
\end{align}
where $(*)$ follows from Cauchy-Schwartz inequality and $(**)$ follows from \cref{eqn:B22,eqn:A11B11}. Similarly, 
\begin{align}
    \frac{1}{\sqrt{d}}\Fnorm{A_{12}^*B_{11}}=~&\frac{1}{\sqrt{d}}\Fnorm{\sqrt{A_{22}}K_A\sqrt{A_{11}}B_{11}}\nonumber\\
    \leq~&\frac{1}{\sqrt{d}}\sqrt{\Fnorm{K_A^*A_{22}K_A}\Fnorm{B_{11}A_{11}B_{11}}}\nonumber\\
    \leq~&\frac{1}{\sqrt{d}}\sqrt{\Fnorm{A_{22}}\opnorm{K_A}^2\opnorm{B_{11}}\Fnorm{A_{11}B_{11}}}\nonumber\\ =~&O(\epsilon).\label{eqn:A12B11}
\end{align}
Combining \cref{eqn:A22B12,eqn:A12B22,eqn:A22B22,eqn:A12B12,eqn:A11B11,eqn:A11B12,eqn:A12B11}, we finally conclude our result.
\end{proof}

Now we prove \cref{mag-commute-same-row}.
\begin{proof}[Proof of \cref{mag-commute-same-row}]
    Without loss of generality, we analyze $\left[P_{1,1},P_{1,2}\right]$. 
    
    By \cref{mgs-consist-row-col} item 2 and \cref{sqrt-approx-project}, $P_{1,1}$ and $P_{1,2}$ are each $\sqrt{\epsilon}$-close to binary observables, that is,
    \begin{align*}
        P_{1,1}\approx_{\sqrt{\epsilon}}E_{1,1,1}^{r_1}-E_{-1,-1,1}^{r_1}+E_{1,-1,-1}^{r_1}-E_{-1,1,-1}^{r_1}\approx_{\sqrt{\epsilon}}\Pi_{u^+}-\Pi_{u^-},\\
        P_{1,2}\approx_{\sqrt{\epsilon}}E_{1,1,1}^{r_1}-E_{-1,-1,1}^{r_1}-E_{1,-1,-1}^{r_1}+E_{-1,1,-1}^{r_1}\approx_{\sqrt{\epsilon}}\Pi_{v^+}-\Pi_{v^-},
    \end{align*}
    where $\Pi_{u^+},\Pi_{u^-},\Pi_{v^+}$ and $\Pi_{v^-}$ are projectors satisfying $\Pi_{u^+}+\Pi_{u^-}=\Pi_{v^+}+\Pi_{v^-}=I$. 
    
    Furthermore, by \cref{prop:faultyboundederror}, the sum of the four involved POVM elements satisfies \[E_{1,1,1}^{r_1}+E_{-1,-1,1}^{r_1}+E_{1,-1,-1}^{r_1}+E_{-1,1,-1}^{r_1}\approx_{\sqrt{\epsilon}}I.\] Thus we have
    \begin{align*}
         E_{1,1,1}^{r_1}+E_{1,-1,-1}^{r_1}\approx_{\sqrt{\epsilon}}\Pi_{u^+},&&
         E_{-1,-1,1}^{r_1}+E_{-1,1,-1}^{r_1}\approx_{\sqrt{\epsilon}}\Pi_{u^-},\\
         E_{1,1,1}^{r_1}+E_{-1,1,-1}^{r_1}\approx_{\sqrt{\epsilon}}\Pi_{v^+},&&
         E_{-1,-1,1}^{r_1}+E_{1,-1,-1}^{r_1}\approx_{\sqrt{\epsilon}}\Pi_{v^-}.
    \end{align*}

    By \cref{lem:approx0}, for any two different POVM elements $E_{a_1,a_2,a_3}^{r_1}$ and $E_{a_1',a_2',a_3'}^{r_1}$ where $a_1a_2a_3=a_1'a_2'a_3'=1$, they approximately annihilate each other:
    \begin{align*}
        E_{a_1,a_2,a_3}^{r_1}E_{a_1',a_2',a_3'}^{r_1}\approx_{\sqrt{\epsilon}} 0.
    \end{align*} Therefore we have
    \begin{multline*}
P_{1,1}P_{1,2}\approx_{\sqrt{\epsilon}}\br{E_{1,1,1}^{r_1}}^2+\br{E_{-1,-1,1}^{r_1}}^2\\-\br{E_{1,-1,-1}^{r_1}}^2-\br{E_{-1,1,-1}^{r_1}}^2\approx_{\sqrt{\epsilon}}P_{1,2}P_{1,1}.
    \end{multline*}

    Thus
    \begin{align*}\left[P_{1,1},P_{1,2}\right]\approx_{\sqrt{\epsilon}}0.
    \end{align*}
    The same argument applies symmetrically to any pair of observables $P_{i,j}, P_{i,j'}$ with $j \ne j'$, or $P_{i,j}, P_{i',j}$ with $i \ne i'$, completing the proof.
\end{proof}

\begin{replemma}{mgs-obs-constraint}
    Under the same assumptions as in \cref{mgs-consist-row-col}, the following observable constraints hold:
    \begin{align*}
        &P_{i,1}P_{i,2}P_{i,3} \approx_{\sqrt{\epsilon}} I \quad \text{for all } i\in[3], \\
        &P_{1,j}P_{2,j}P_{3,j} \approx_{\sqrt{\epsilon}} I \quad \text{for } j=1,2, \\
        &P_{1,3}P_{2,3}P_{3,3} \approx_{\sqrt{\epsilon}} -I.
    \end{align*}
\end{replemma}
\begin{proof}
    Without loss of generality, consider the first row. Let $\set{E_{a_1,a_2,a_3}^{r_1} : a_1a_2a_3=1}$ be relabeled as $\set{E_0, E_1, E_2, E_3}$ for convenience. 
    Given that
    \begin{align*}
E_0^2\approx_{\sqrt{\epsilon}}E_0\br{I-E_1-E_2-E_3}\approx_{\sqrt{\epsilon}} E_0,
\end{align*}
where the first approximation follows from \cref{prop:faultyboundederror} and the second approximation is by \cref{lem:approx0}. We can prove analogously that $E_1,E_2,E_3$ are each $\sqrt{\epsilon}$-close to a projector. Thus
\begin{align*}
    &P_{1,1}P_{1,2}P_{1,3}\\
    &= \br{E_0-E_1+E_2-E_3}\br{E_0-E_1-E_2+E_3}\br{E_0+E_1-E_2-E_3}\\
    &\approx_{\sqrt{\epsilon}} E_0^3+E_1^3+E_2^3+E_3^3\approx_{\sqrt{\epsilon}} E_0+E_1+E_2+E_3\approx_{\sqrt{\epsilon}}I.
\end{align*}
The remaining constraints follow analogously for the other rows and columns.
\end{proof}

\begin{replemma}{claim-mgs-anti-commute}
Under the same assumptions as in \cref{mgs-consist-row-col}, for all $i \neq i'$ and $j \neq j'$, we have
\[
\set{P_{i,j}, P_{i',j'}} \approx_{\sqrt{\epsilon}} 0, \qquad \set{Q_{i,j}, Q_{i',j'}} \approx_{\sqrt{\epsilon}} 0.
\]
\end{replemma}
\begin{proof}
    Without loss of generality, consider $P_{1,1}$ and $P_{2,2}$. By \cref{mgs-obs-constraint} and \cref{mgs-consist-row-col} item 2, we have 
    \begin{align*}
        P_{2,2}\approx_{\sqrt{\epsilon}} P_{2,3}P_{2,1},\quad P_{1,1}
    \approx_{\sqrt{\epsilon}}P_{1,2}P_{1,3},\quad P_{1,3}P_{2,3}
    \approx_{\sqrt{\epsilon}}-P_{3,3},\\
     P_{2,2}\approx_{\sqrt{\epsilon}} P_{1,2}P_{3,1},\quad P_{1,1}
    \approx_{\sqrt{\epsilon}}P_{3,1}P_{2,1},\quad P_{3,2}P_{3,1}
    \approx_{\sqrt{\epsilon}}P_{3,3}.
    \end{align*} Combining these, we get 
    \begin{align*}
    P_{1,1}P_{2,2}\approx_{\sqrt{\epsilon}} P_{1,2}P_{1,3}P_{2,3}P_{2,1}
    \approx_{\sqrt{\epsilon}}-P_{1,2}P_{3,3}P_{2,1}.\\
P_{2,2}P_{1,1}\approx_{\sqrt{\epsilon}}P_{1,2}P_{3,2}P_{3,1}P_{2,1}\approx_{\sqrt{\epsilon}}P_{1,2}P_{3,3}P_{2,1}.
    \end{align*}
    Thus $\set{P_{1,1},P_{2,2}}\approx_{\sqrt{\epsilon}}0$.  
    By \cref{mgs-consist-row-col} item 1 we also have $\set{Q_{1,1},Q_{2,2}}\approx_{\sqrt{\epsilon}} 0$.  
\end{proof}
\begin{replemma}{lem:ms-single-observable-localization}
Under the assumptions of \cref{mgs-consist-row-col}, for every \(i,j\in[3]\), there exist indices
\[
\ell_{i,j},m_{i,j}\in[n]
\]
and traceless binary observables
\[
\widetilde{P}_{i,j}^{(\ell_{i,j})},\ \widetilde{Q}_{i,j}^{(m_{i,j})}\in \H_4
\]
such that
\begin{align*}
P_{i,j}
&\approx_{\gamma\sqrt{\epsilon}}
\widetilde{P}_{i,j}^{(\ell_{i,j})}\otimes I^{[n]\setminus\{\ell_{i,j}\}}_{4},\\
Q_{i,j}
&\approx_{\gamma\sqrt{\epsilon}}
\widetilde{Q}_{i,j}^{(m_{i,j})}\otimes I^{[n]\setminus\{m_{i,j}\}}_{4}.
\end{align*}
\end{replemma}

\begin{proof}[Proof of \cref{lem:ms-single-observable-localization}]
We prove the claim for \(P_{i,j}\); the argument for \(Q_{i,j}\) is identical.

By item \(1\) of \cref{mgs-consist-row-col} and Eq.~\ref{mgs-bob-noise}, we have
\[
P_{i,j}' \approx_{\sqrt{\epsilon}} \rho P_{i,j}.
\]
Moreover, by item \(2\) of \cref{mgs-consist-row-col},
\[
P_{i,j}^2 \approx_{\sqrt{\epsilon}} I.
\]

Write the Pauli expansion of \(P_{i,j}\) under the basis \(\mathcal{A}^{\otimes 2}\) as
\[
P_{i,j}=\sum_{x\in[16]_{\ge 0}^{n}}\widehat{P}_{i,j}(x)\,\mathcal{A}_{x}^{\otimes 2}.
\]
Then
\[
\widehat{P}_{i,j}(0^n)^2=\widebar{\Tr}(P_{i,j})^2=O(\epsilon).
\]
Using \(P_{i,j}' \approx_{\sqrt{\epsilon}} \rho P_{i,j}\), we obtain
\begin{align*}
O(\epsilon)
&=
\widebar{\Tr}\!\left((P_{i,j}'-\rho P_{i,j})^2\right)\\
&=
\sum_{x\in[16]_{\ge 0}^{n}}
(\rho^{|x|}-\rho)^2 \widehat{P}_{i,j}(x)^2\\
&\ge
(\rho-\rho^2)^2
\sum_{|x|\ge 2}\widehat{P}_{i,j}(x)^2.
\end{align*}
Hence
\[
\sum_{|x|\ge 2}\widehat{P}_{i,j}(x)^2=O(\epsilon\gamma^2).
\]

Let \(\overline{P}_{i,j}\) be the degree-one truncation of \(P_{i,j}\). Then
\[
P_{i,j}\approx_{\gamma\sqrt{\epsilon}}\overline{P}_{i,j},
\]
and \(\overline{P}_{i,j}\) has the form
\[
\overline{P}_{i,j}
=
\sum_{r=1}^{n} a_{r}^{\,i,j}\,O_{r}^{\,i,j}\otimes I^{[n]\setminus\{r\}}_{4},
\]
where \(a_{r}^{\,i,j}\ge 0\), each \(O_{r}^{\,i,j}\in\H_4\) is traceless, and
\[
\widebar{\Tr}\!\left((O_{r}^{\,i,j})^2\right)=1.
\]
Since \(P_{i,j}^2\approx_{\sqrt{\epsilon}}I\), we also have
\begin{equation}\label{eqn:PbarclosetoI}
    \overline{P}_{i,j}^{\,2}\approx_{\gamma\sqrt{\epsilon}}I.
\end{equation}

Applying \cref{mgs-one-reg} to \(\overline{P}_{i,j}\), there exists an index
\[
\ell_{i,j}\in[n]
\]
and a traceless $M\in\H_4$ such that
\[
\overline{P}_{i,j}
\approx_{\gamma\sqrt{\epsilon}}
M\otimes I^{[n]\setminus\{\ell_{i,j}\}}_{4}.
\]

By \cref{eqn:PbarclosetoI} and \cref{approx_multiply}, we have
\[M\otimes I^{[n]\setminus\{\ell_{i,j}\}}_{4}\approx_{\gamma\sqrt{\epsilon}}I.\]

By \cref{sqrt-approx-project}, there exists a traceless binary observable
$
\widetilde{P}_{i,j}^{(\ell_{i,j})}\in\H_4
$
such that
\[
P_{i,j}
\approx_{\gamma\sqrt{\epsilon}}
\widetilde{P}_{i,j}^{(\ell_{i,j})}\otimes I^{[n]\setminus\{\ell_{i,j}\}}_{4}.
\]

Applying the same argument to \(Q_{i,j}\) gives an index \(m_{i,j}\in[n]\) and a traceless binary observable
\[
\widetilde{Q}_{i,j}^{(m_{i,j})}\in\H_4
\]
such that
\[
Q_{i,j}
\approx_{\gamma\sqrt{\epsilon}}
\widetilde{Q}_{i,j}^{(m_{i,j})}\otimes I^{[n]\setminus\{m_{i,j}\}}_{4}.
\]
This proves the lemma.
\end{proof}
\begin{replemma}{lem:ms-common-register}
Under the assumptions of \cref{mgs-consist-row-col}, there exists a single register index
\[
\ell\in[n]
\]
such that for every \(i,j\in[3]\), there exist traceless binary observables
\[
\widetilde{P}_{i,j}^{(\ell)},\ \widetilde{Q}_{i,j}^{(\ell)}\in\H_4
\]
satisfying
\begin{align*}
P_{i,j}
&\approx_{\gamma\sqrt{\epsilon}}
\widetilde{P}_{i,j}^{(\ell)}\otimes I^{[n]\setminus\{\ell\}}_{4},\\
Q_{i,j}
&\approx_{\gamma\sqrt{\epsilon}}
\widetilde{Q}_{i,j}^{(\ell)}\otimes I^{[n]\setminus\{\ell\}}_{4}.
\end{align*}
\end{replemma}
\begin{proof}[Proof of \cref{lem:ms-common-register}]
By \cref{lem:ms-single-observable-localization}, for each \(i,j\in[3]\) there exist indices
\[
\ell_{i,j},m_{i,j}\in[n]
\]
and traceless binary observables
\[
\widetilde{P}_{i,j}^{(\ell_{i,j})},\ \widetilde{Q}_{i,j}^{(m_{i,j})}\in\H_4
\]
such that
\begin{align}
P_{i,j}
&\approx_{\gamma\sqrt{\epsilon}}
\widetilde{P}_{i,j}^{(\ell_{i,j})}\otimes I^{[n]\setminus\{\ell_{i,j}\}}_{4}, \label{eq:ms-common-register-P}\\
Q_{i,j}
&\approx_{\gamma\sqrt{\epsilon}}
\widetilde{Q}_{i,j}^{(m_{i,j})}\otimes I^{[n]\setminus\{m_{i,j}\}}_{4}. \label{eq:ms-common-register-Q}
\end{align}

We first show that all \(\ell_{i,j}\) are equal. By \cref{claim-mgs-anti-commute},
\[
\{P_{1,1},P_{2,2}\}\approx_{\sqrt{\epsilon}}0.
\]
If \(\ell_{1,1}\neq \ell_{2,2}\), then by \cref{approx_multiply} and \cref{eq:ms-common-register-P},
\[
\{P_{1,1},P_{2,2}\}
\approx_{\gamma\sqrt{\epsilon}}
2\,\widetilde{P}_{1,1}^{(\ell_{1,1})}\widetilde{P}_{2,2}^{(\ell_{2,2})}
\otimes I^{[n]\setminus\{\ell_{1,1},\ell_{2,2}\}}_{4},
\]
whose normalized Frobenius norm is bounded below by a positive constant since both local factors are binary. This is a contradiction. Hence
\[
\ell_{1,1}=\ell_{2,2}.
\]

Next, suppose \(\ell_{1,1}\neq \ell_{1,2}\). By \cref{mag-commute-same-row,mgs-obs-constraint},
\[
P_{1,3}\approx_{\sqrt{\epsilon}}P_{1,1}P_{1,2}.
\]
Using \cref{eq:ms-common-register-P}, this gives
\[
P_{1,3}
\approx_{\gamma\sqrt{\epsilon}}
\widetilde{P}_{1,1}^{(\ell_{1,1})}\widetilde{P}_{1,2}^{(\ell_{1,2})}
\otimes I^{[n]\setminus\{\ell_{1,1},\ell_{1,2}\}}_{4}.
\]
The right-hand side is a degree-two operator supported on two distinct registers, whereas by
\cref{lem:ms-single-observable-localization}, \(P_{1,3}\) is \(O(\gamma\sqrt{\epsilon})\)-close to a single-register observable. Since a degree-two traceless operator is orthogonal in normalized Hilbert--Schmidt inner product to every single-register traceless operator, this is impossible. Hence
\[
\ell_{1,2}=\ell_{1,1}.
\]
Then the row parity relation again implies
\[
\ell_{1,3}=\ell_{1,1}.
\]

Applying the same argument to the other rows and columns, we conclude that all indices \(\ell_{i,j}\) are equal. Denote the common value by
\[
\ell\in[n].
\]
Thus for every \(i,j\in[3]\),
\[
P_{i,j}\approx_{\gamma\sqrt{\epsilon}}
\widetilde{P}_{i,j}^{(\ell)}\otimes I^{[n]\setminus\{\ell\}}_{4}.
\]

Finally, by item \(1\) of \cref{mgs-consist-row-col},
\[
Q_{i,j}^{T}\approx_{\sqrt{\epsilon}}P_{i,j}.
\]
Combining this with the previous display and \cref{eq:ms-common-register-Q}, we conclude that the same register \(\ell\) also works for Bob. This proves the lemma.
\end{proof}

\begin{replemma}{lem:ms-local-canonical-form}
Assume the hypotheses of \cref{mgs-consist-row-col} and let \(\ell\in[n]\) be the common register given by
\cref{lem:ms-common-register}. Then there exists a \(4\)-dimensional unitary
\[
\widetilde{U}:\H_4\to\H_4
\]
such that for all \(i,j\in[3]\),
\begin{align*}
\widetilde{U}\widetilde{P}_{i,j}^{(\ell)}\widetilde{U}^*
&\approx_{\gamma\sqrt{\epsilon}} P_{i,j}^{\star},\\
\widetilde{U}\,\overline{\widetilde{Q}_{i,j}^{(\ell)}}\,\widetilde{U}^{T}
&\approx_{\gamma\sqrt{\epsilon}} Q_{i,j}^{\star},
\end{align*}
where \(P_{i,j}^{\star}\) is the \((i,j)\)-th entry of Table~I and
\[
Q_{i,j}^{\star}:=\left(P_{i,j}^{\star}\right)^T.
\]
\end{replemma}

\begin{proof}[Proof of \cref{lem:ms-local-canonical-form}]
By \cref{lem:ms-common-register}, all localized observables act on the same register \(\ell\). For simplicity, suppress the superscript \((\ell)\) in this proof.

Set
\[
A:=\widetilde{P}_{2,2},\qquad
B:=\widetilde{P}_{1,1},\qquad
C:=\widetilde{P}_{1,2}.
\]
By \cref{mag-commute-same-row,claim-mgs-anti-commute} and the localization,
$A^2=B^2=C^2=I,\set{A,B}\approx_{\gamma\sqrt{\epsilon}}0,[A,C]\approx_{\gamma\sqrt{\epsilon}}0, [B,C]\approx_{\gamma\sqrt{\epsilon}}0$.
Applying \cref{mgs-local-unitary}, there exists a \(4\)-dimensional unitary
\[
\widetilde{U}_0:\H_4\to\H_4
\]
such that
\begin{align}
\widetilde{U}_0\widetilde{P}_{2,2}\widetilde{U}_0^*
&\approx_{\gamma\sqrt{\epsilon}} Z\otimes I, \label{eq:ms-local-canonical-1}\\
\widetilde{U}_0\widetilde{P}_{1,1}\widetilde{U}_0^*
&\approx_{\gamma\sqrt{\epsilon}} X\otimes I, \label{eq:ms-local-canonical-2}\\
\widetilde{U}_0\widetilde{P}_{1,2}\widetilde{U}_0^*
&\approx_{\gamma\sqrt{\epsilon}} I\otimes X. \label{eq:ms-local-canonical-3}
\end{align}

Define
\[
S_{i,j}:=\widetilde{U}_0\widetilde{P}_{i,j}\widetilde{U}_0^*.
\]
By \cref{mgs-obs-constraint} and \cref{eq:ms-local-canonical-2,eq:ms-local-canonical-3},
\[
S_{1,3}\approx_{\gamma\sqrt{\epsilon}} S_{1,1}S_{1,2}\approx_{\gamma\sqrt{\epsilon}} X\otimes X.
\]
Similarly, by \cref{mgs-obs-constraint} and \cref{eq:ms-local-canonical-1,eq:ms-local-canonical-3},
\[
S_{3,2}\approx_{\gamma\sqrt{\epsilon}} S_{1,2}S_{2,2}\approx_{\gamma\sqrt{\epsilon}} Z\otimes X.
\]

Next, since \(S_{2,1}\) commutes approximately with both \(X\otimes I\) and \(Z\otimes I\), its Pauli expansion contains only terms of the form \(I\otimes \sigma_{\alpha}\). Hence there exists a traceless binary observable
\[
W\in\H_2
\]
such that
\[
S_{2,1}\approx_{\gamma\sqrt{\epsilon}} I\otimes W.
\]
Moreover, by \cref{claim-mgs-anti-commute} and \cref{eq:ms-local-canonical-3},
\[
\{S_{2,1},I\otimes X\}\approx_{\gamma\sqrt{\epsilon}}0.
\]
Therefore
\[
\{W,X\}\approx_{\gamma\sqrt{\epsilon}}0.
\]
Applying \cref{anti-commute-to-ZX}, there exists a \(2\)-dimensional unitary \(V\) such that
\[
VWV^*\approx_{\gamma\sqrt{\epsilon}} Z,
\qquad
VXV^*\approx_{\gamma\sqrt{\epsilon}} X.
\]
Define
\[
\widetilde{U}:=(I\otimes V)\widetilde{U}_0.
\]
Then
\begin{align*}
\widetilde{U}\widetilde{P}_{1,1}\widetilde{U}^*
&\approx_{\gamma\sqrt{\epsilon}} X\otimes I,\\
\widetilde{U}\widetilde{P}_{1,2}\widetilde{U}^*
&\approx_{\gamma\sqrt{\epsilon}} I\otimes X,\\
\widetilde{U}\widetilde{P}_{2,1}\widetilde{U}^*
&\approx_{\gamma\sqrt{\epsilon}} I\otimes Z,\\
\widetilde{U}\widetilde{P}_{2,2}\widetilde{U}^*
&\approx_{\gamma\sqrt{\epsilon}} Z\otimes I.
\end{align*}
Using \cref{mgs-obs-constraint} once more, we obtain
\begin{align*}
\widetilde{U}\widetilde{P}_{1,3}\widetilde{U}^*
&\approx_{\gamma\sqrt{\epsilon}} X\otimes X,\\
\widetilde{U}\widetilde{P}_{2,3}\widetilde{U}^*
&\approx_{\gamma\sqrt{\epsilon}} Z\otimes Z,\\
\widetilde{U}\widetilde{P}_{3,1}\widetilde{U}^*
&\approx_{\gamma\sqrt{\epsilon}} X\otimes Z,\\
\widetilde{U}\widetilde{P}_{3,2}\widetilde{U}^*
&\approx_{\gamma\sqrt{\epsilon}} Z\otimes X,\\
\widetilde{U}\widetilde{P}_{3,3}\widetilde{U}^*
&\approx_{\gamma\sqrt{\epsilon}} Y\otimes Y.
\end{align*}
That is,
\[
\widetilde{U}\widetilde{P}_{i,j}\widetilde{U}^*
\approx_{\gamma\sqrt{\epsilon}} P_{i,j}^{\star}
\qquad\forall i,j\in[3].
\]

Finally, by item \(1\) of \cref{mgs-consist-row-col},
\[
\widetilde{Q}_{i,j}^{\,T}\approx_{\gamma\sqrt{\epsilon}}\widetilde{P}_{i,j}.
\]
Therefore
\[
\widetilde{U}\,\overline{\widetilde{Q}_{i,j}}\,\widetilde{U}^{T}
\approx_{\gamma\sqrt{\epsilon}} Q_{i,j}^{\star}
\qquad\forall i,j\in[3].
\]
This proves the lemma.
\end{proof}

\begin{lemma}\label{mgs-one-reg}
    For a Hermitian matrix $A\in\H_{4}^{\x n}$, if $A^2\approx_{\epsilon} I$ and $A$ contains only degree-one terms:
    \begin{align*}
        A=\sum_{i=1}^n a_i O^{\br{i}}\otimes I^{[n]\backslash \set{i}},
    \end{align*}
    where $a_i\geq 0$, $O^{\br{i}}\in\H_4$ is traceless, and $\widebar{\Tr}\!\left(\br{O^{\br{i}}}^2\right)=1$.
    Then there exists $\ell\in[n]$ such that $A\approx_{\epsilon} O^{\br{\ell}}\otimes I^{[n]\backslash\set{\ell}}$.
\end{lemma}
\begin{proof}
    Since $\widebar{\Tr}\!\left(\br{O^{\br{i}}}^2\right)=1$, we can write $\br{O^{\br{i}}}^2=I+B^{\br{i}}$, where $\Tr B^{\br{i}}=0$. We then have
    \begin{align*}
        A^2=~&\sum_{i} a_i^2 \br{O^{\br{i}}}^2\otimes I^{[n]\backslash\set{i}}\\
        &+\sum_{i,j} a_ia_jO^{\br{i}}\otimes O^{\br{j}}\otimes I^{[n]\backslash\set{i,j}}\\
        =~&\sum_{i} a_i^2 I+\sum_ia_i^2B^{\br{i}}\otimes I^{ [n]\backslash\set{i}}\\
        &+\sum_{i,j} a_ia_jO^{\br{i}}\otimes O^{\br{j}}\otimes I^{[n]\backslash\set{i,j}}
    \end{align*}
    From $A^2\approx_{\epsilon} I$ and Parseval's identity, we have
    \begin{align*}
        \sum_ia_i^2\approx_{\epsilon}1,\\
        \sum_{i\neq j}a_i^2a_j^2\approx_{\epsilon^2} 0.
    \end{align*}
        Since $\sum_{j}a_j^2\br{1-a_j^2} \approx_{\epsilon}\sum_{j\neq i}a_j^2a_i^2$, we have $\sum_{j}a_j^2\br{1-a_j^2}\approx_{\epsilon}0$, which implies 
     $\forall j,$ either $a_j^2\approx_{\epsilon} 0$ or $a_j^2\approx_{\epsilon} 1$. Combining this with $\sum_ia_i^2\approx_{\epsilon}1$, there exists one index $\ell\in[n]$ such that $a_\ell^2\approx_{\epsilon} 1$ 
     (which implies $a_\ell\approx_{\epsilon} 1$)
     and $\sum_{j\neq \ell}a_j^2\approx_{\epsilon^2} 0$. Then the conclusion holds.
\end{proof}

\begin{lemma}\label{mgs-local-unitary}
    Let $A,B,C\in \H_4$ be Hermitian and traceless matrices satisfying $A^2=B^2=C^2=I$. If $\set{A,B}\approx_{\epsilon}0,\left[A,C\right]\approx_{\epsilon}0,\left[B,C\right]\approx_{\epsilon}0$, then there exists a unitary $U$, such that
    \begin{align*}
        UAU^*\approx_{\epsilon} Z\otimes I,\qquad UBU^*\approx_{\epsilon} X\otimes I,\qquad UCU^*\approx_{\epsilon} I\otimes X.
    \end{align*}
\end{lemma}
\begin{proof}
    Without loss of generality, we can assume $A= \begin{bmatrix}
I_{2} &0 \\
0 & -I_{2}
\end{bmatrix}$. From the condition $\set{A,B}\approx_{\epsilon}0$  we have
\begin{align*}
    B\approx_{\epsilon}\begin{bmatrix}
0 &B_{12} \\
B_{12}^* & 0
\end{bmatrix},
\end{align*}
where $B_{12}\in\M_2$. By $B^2=I$ we have $B_{12}^*B_{12}\approx_{\epsilon} I,B_{12}B_{12}^*\approx_{\epsilon} I$. Let the singular value decomposition of $B_{12}$ be
\begin{align*}
    B_{12}=V\Lambda W^*,\qquad\Lambda=\begin{bmatrix}
\lambda_1 &0 \\
0 & \lambda_2
\end{bmatrix}.
\end{align*}
Since $B_{12}^*B_{12}\approx_{\epsilon} I$, we have $\Lambda^2\approx_{\epsilon} I$, thus
\begin{align*}
    \br{\lambda_1^2-1}^2+\br{\lambda_2^2-1}^2\approx_{\epsilon^2}0.
\end{align*}
Since \begin{multline*}
    \br{\lambda_1-1}^2+\br{\lambda_2-1}^2\\\leq \br{\lambda_1-1}^2\br{\lambda_1+1}^2+\br{\lambda_2-1}^2\br{\lambda_2+1}^2\approx_{\epsilon^2} 0,
\end{multline*}
we have $\Lambda\approx_{\epsilon} I$. Define $U_1=\begin{bmatrix}
V^* &0 \\
0 & W^*
\end{bmatrix}$, then we have
\begin{align*}
    U_1AU_1^*=Z\x I, \qquad U_1BU_1^*\approx_{\epsilon} X \x I.
\end{align*}
Since $\left[U_1CU_1^*,U_1AU_1^*\right]\approx_{\epsilon}0,\left[U_1CU_1^*,U_1BU_1^*\right]\approx_{\epsilon}0$, we have $U_1CU_1^*\approx_{\epsilon}\begin{bmatrix}
C_1 &0 \\
0 & C_1
\end{bmatrix}$ for some $C_1\in\H_2, \Tr\br{C_1}=0,C_1^2\approx_\epsilon I$. By \cref{sqrt-approx-project}, $C_1\approx_{\epsilon} \ketbra{u}-\ketbra{v}$ for orthogonal unit vectors $\ket{u},\ket{v}$.  
Then for unitary $V_1=\ket{+}\bra{u}+\ket{-}\bra{v}$, we have $V_1C_1V_1^*=X$. Let $U_2=\begin{bmatrix}
V_1 &0 \\
0 & V_1
\end{bmatrix}$. By calculation $U=U_2U_1$ satisfies the requirement.
\end{proof}
\bibliographystyle{apsrev4-2}
\bibliography{references}
\end{document}